\documentclass[a4paper,reqno]{amsart}
\usepackage[all]{xy}           
\usepackage{amssymb}           
\usepackage{hyperref}
\usepackage{eucal}
\usepackage[usenames,dvipsnames]{color}
\usepackage{graphicx}

\usepackage{diagrams}
\numberwithin{equation}{section}

\newtheorem{definition}{Definition}[section]
\newtheorem{lemma}[definition]{Lemma}
\newtheorem{theorem}[definition]{Theorem}
\newtheorem{proposition}[definition]{Proposition}

\newtheorem{remarkth}[definition]{Remark}
\newtheorem{example}[definition]{Example}
\newenvironment{remark}{\begin{remarkth}\upshape}{\hfill$\diamond$\end{remarkth}}
\renewcommand{\emph}[1]{{\bfseries\itshape{#1}}}

\usepackage{amssymb}

\newcount\ancho \newcount\anchom \newcount\anchoa
\newcount\anchob \newcount\altura

\newcommand{\ltilde}[3][0]{\altura=0 \advance\altura by #1
           \ancho=#2 \anchom=\ancho \divide\anchom by 2
           \anchoa=\ancho \divide\anchoa by 4
           \anchob=\anchom \advance\anchob by \anchoa
           \kern-3pt \begin{array}[b]{c}
           \begin{picture}(1,1)(\anchom,-\altura)
        \qbezier(0,2)(\anchoa,5)(\anchom,2)
        \qbezier(\anchom,2)(\anchob,-1)(\ancho,4)
        \qbezier(0,2)(\anchoa,4.5)(\anchom,1.8)
        \qbezier(\anchom,1.8)(\anchob,-1.5)(\ancho,4)
       \end{picture} \\[-4pt]{#3}
                       \end{array} \kern-4pt    }

\newcommand{\lhat}[3][0]{\altura=0 \advance\altura by #1
           \ancho=#2 \anchom=\ancho \divide\anchom by 2
           \anchoa=\ancho \divide\anchoa by 4
           \anchob=\anchom \advance\anchob by \anchoa
           \kern-3pt \begin{array}[b]{c}
           \begin{picture}(1,1)(\anchom,-\altura)
        \qbezier(0,2)(\anchoa,4)(\anchom,6)
        \qbezier(\anchom,6)(\anchob,4)(\ancho,2)
        \qbezier(0,2)(\anchoa,3.8)(\anchom,5.6)
        \qbezier(\anchom,5.6)(\anchob,3.8)(\ancho,2)
       \end{picture} \\[-4pt] {#3}
                       \end{array} \kern-4pt    }

\makeatletter
\newcommand\prol{\@ifstar{\@proldf}{\@prolpf}}  
\def\@prolpf{\@ifnextchar[{\@prolpf@wrt}{\@prolpf@}}
\def\@prolpf@wrt[#1]#2{\@ifnextchar[{\@prolpf@wrt@at{#1}{#2}}{\@prolpf@wrt@{#1}{#2}}}
\def\@prolpf@wrt@at#1#2[#3]{\prolsymbol^{#1}_{#3}#2}
\def\@prolpf@wrt@#1#2{\prolsymbol^{#1}#2}
\def\@prolpf@#1{\@ifnextchar[{\@prolpf@at{#1}}{\@prolpf@@{#1}}}
\def\@prolpf@at#1[#2]{\prolsymbol_{#2}#1}
\def\@prolpf@@#1{\prolsymbol#1}
\def\@proldf{\@ifnextchar[{\@proldf@wrt}{\@proldf@}}
\def\@proldf@wrt[#1]#2{\@ifnextchar[{\@proldf@wrt@at{#1}{#2}}{\@proldf@wrt@{#1}{#2}}}
\def\@proldf@wrt@at#1#2[#3]{\prolsymbol^{*#1}_{#3}#2}
\def\@proldf@wrt@#1#2{\prolsymbol^{*#1}#2}
\def\@proldf@#1{\@ifnextchar[{\@proldf@at{#1}}{\@proldf@@{#1}}}
\def\@proldf@at#1[#2]{\prolsymbol^*_{#2}#1}
\def\@proldf@@#1{\prolsymbol^*#1}
\def\prolsymbol{\mathcal{T}}
\makeatother

\newcommand{\pr}{\mathrm{pr}}
\begin{document}

\title[A new geometric formulation of Hamiltonian Classical Field theories]{A new canonical affine BRACKET formulation of Hamiltonian Classical Field theories of first order}

\author[F. Gay-Balmaz]{Fran\c{c}ois Gay-Balmaz}
\address{F. Gay-Balmaz:
Laboratoire de Meteorologie Dynamique \\ 
Ecole Normale Superieure, CNRS, F-75231 \\
Paris, France}
\email{gaybalma@lmd.ens.fr}

\author[J.\ C.\ Marrero]{Juan C.\ Marrero}
\address{J.\ C.\  Marrero:
ULL-CSIC Geometr\'{\i}a Diferencial y Mec\'anica Geom\'etrica\\
Departamento de Mate\-m\'a\-ti\-cas, Estad{\'\i}stica e IO, Secci\'on de
Ma\-te\-m\'a\-ti\-cas, Universidad de La Laguna \\
La Laguna, Tenerife, Canary Islands, Spain} 
\email{jcmarrer@ull.edu.es}

\author[N. Mart{\'\i}nez]{Nicol\'as Mart{\'\i}nez}
\address{N. Mart{\'\i}nez: 
Departamento de Matematicas, Universidad Nacional de Colombia \\ 
Bogot\'a, Colombia}
\email{nmartineza@unal.edu.co}

\thanks{J.C.M. thanks H. Bursztyn, A. Cabrera, S. Capriotti and D. Iglesias for stimulating discussions on the topics of this paper. J.C.M. has been partially supported by Ministerio de Econom\'{i}a y Competitividad (MINECO, Spain) under grant PGC2018-098265-B-C32}




\begin{abstract}
It has been a long standing question how to extend, in the finite-dimensional setting, the canonical Poisson bracket formulation from classical mechanics to classical field theories, in a completely general, intrinsic, and canonical way. In this paper, we provide an answer to this question by presenting a new completely canonical bracket formulation of Hamiltonian Classical Field Theories of first order on an arbitrary configuration bundle. It is obtained via the construction of the appropriate field-theoretic analogues of the Hamiltonian vector field and of the space of observables, via the introduction of a suitable canonical Lie algebra structure on the space of currents (the observables in field theories). This Lie algebra structure is shown to have a representation on the affine space of Hamiltonian sections, which yields an affine analogue to the Jacobi identity for our bracket. The construction is analogous to the canonical Poisson formulation of Hamiltonian systems although the nature of our formulation is linear-affine and not bilinear as the standard Poisson bracket. This is consistent with the fact that the space of currents and Hamiltonian sections are respectively, linear and affine. Our setting is illustrated with some examples including Continuum Mechanics and Yang-Mills theory.
\end{abstract}

\vspace{1cm}

\maketitle

\section{Introduction}

It is well-known that the Hamilton equations of classical mechanics are naturally formulated in terms of \textit{canonical geometric structures}, namely, canonical symplectic forms and canonical Poisson brackets. Given the configuration manifold $Q$ of the mechanical system, its  momentum phase space $T^*Q$ carries a canonical symplectic form which induces a vector bundle isomorphism $\sharp: T^*(T^*Q) \to T(T^*Q)$. Given the Hamiltonian $H \in C^\infty(T^*Q)$ of the system, the associated Hamilton equations of motion are determined by the Hamiltonian vector field $X_H$ intrinsically defined in terms of the canonical symplectic form as $X_H= \sharp dH$, i.e., we have the commutative diagram
\begin{equation}\label{diagram-1}
\xymatrix{T^*(T^*Q)\ar[rr]^{\sharp}&&T(T^*Q)\\&
T^*Q\ar[ul]^{dH}\ar[ur]_{X_H}&}
\end{equation}  
The time evolution of an observable $F\in C^\infty(T^*Q)$ along a solution $s: I \subseteq \mathbb{R} \to T^*Q$ of Hamilton's equations, is given by
\begin{equation}\label{Poisson_formulation} 
\frac{d}{dt}(F\circ s)= \{F, H\}\circ s,
\end{equation} 
where  $\{\cdot, \cdot\}: C^{\infty}(T^*Q) \times C^{\infty}(T^*Q) \to C^{\infty}(T^*Q)$ is the canonical Poisson bracket defined by 
\begin{equation}\label{Canonical_Poisson} 
\{F, G\}=\left\langle d F , X_G \right\rangle= \left\langle d F , \sharp dG \right\rangle = -\left\langle\sharp dF, dG \right\rangle.
\end{equation} 
Conversely, if a curve $s: I \subseteq \mathbb{R} \to T^*Q$ satisfies \eqref{Poisson_formulation} for any observables $F$, then $s$ is a solution of Hamilton's equations (see, for instance, \cite{AbMa,LeRo,LiMa}).
Recall that \eqref{Canonical_Poisson} defines a Lie algebra structure on $C^\infty(M)$, thereby satisfying the Jacobi identity
\begin{equation}\label{Jacobi} 
\{\{F, G\}, H\}=\{F, \{G, H\}\}- \{G, \{F, H\}\}.
\end{equation}
The canonical Poisson formulation \eqref{Poisson_formulation} of Hamilton's equation has been at the origin of many developments in the understanding of the geometry and dynamics of Hamiltonian systems and their quantization, as well as of many generalizations. In particular, the Poisson formulation is very appropriate for developing the geometric description of the Poisson reduction of a Hamiltonian system which is invariant under the action of a symmetry Lie group (see, for instance, \cite{MaRa}).

\medskip

When extending the geometric setting from classical mechanics to classical field theories, it is crucial to identify the geometric structures playing the role of these canonical structures. While the field-theoretic analogue of the canonical symplectic structure is well-known to be given by the canonical multisymplectic form on a space of finite dimension, the extended momentum phase space, it has been a long standing question how to extend the canonical Poisson bracket formulation from classical mechanics to classical field theories, in a completely general, intrinsic, and canonical way. Several contributions have been made in this direction as we will review below. In this paper will shall provide an answer to this question by constructing explicitly such a canonical bracket, giving its algebraic properties, and showing that it allows a canonical formulation of the Hamilton equations for field theories that naturally extends the formulation \eqref{Poisson_formulation} of classical mechanics. Such canonical bracket structures could be used as a starting point of a covariant canonical quantization. One key difference between the canonical bracket that we propose and the canonical Poisson bracket of classical mechanics is its linear-affine nature. This is compatible with the fact that the set of currents and Hamiltonian sections are, respectively, linear and affine spaces for field theories. This linear-affine nature already arises in the case of time-dependent Hamiltonian Mechanics.

\medskip

While in the present paper we focus on the extension to field theory of the canonical geometric setting \eqref{Poisson_formulation}--\eqref{Jacobi} governing the evolution equations, which is of finite dimensional nature, one may also focus on the geometric structures related to the space of solutions of these equations. 
For a large class of field theories, this infinite dimensional space admits a Poisson (or a presymplectic) structure. This is not a new theory and we will quote an old paper by Peierls \cite{Pe} (for a large list of contributions, see the recent papers  \cite{CiDiIbMaScZa1,CiDiIbMaScZa2} and the references therein). In this setting, the boundary conditions of the theory play an important role \cite{MaVi}. In fact, in \cite{MaVi}, Margalef-Bentabol and Villaseñor introduce the so-called ``relative bicomplex framework" and develop a geometric formulation of the covariant phase space methods with boundaries, which is used to endow the space of solutions with (pre)symplectic structures. These ideas are used to discuss formulations of Palatini gravity, General Relativity and Holst theories in the presence of boundaries \cite{BaMaVaVi1,BaMaVaVi2,BaMaVaVi3}. On the other hand, a classic research topic has been the relationship between finite and infinite dimensional approximation to Classical Field Theories (see Section \ref{Conclusions-Future}).

\medskip

Before explaining the main difficulties that emerge in the process of finding a canonical bracket formulation that extends \eqref{Poisson_formulation}--\eqref{Jacobi} to classical field theories, we quickly review below the previous contributions to the geometric formulation of Hamiltonian Classical Field Theories.

\subsection{Previous contributions on the geometric formulation of Hamiltonian classical field theories of first order.}
\label{Previous-contributions}
The geometric description starts with the choice of the configuration bundle $\pi: E \to M$, whose sections are the fields of the theory. The case of time-dependent mechanics corresponds to the special situation $E= Q\times \mathbb{R} \rightarrow \mathbb{R}$. For classical field theories, there are two useful generalisations of the notion of momentum phase space: the restricted multi-momentum  bundle ${\mathcal M}^0\pi$ and the extended multi-momentum bundle ${\mathcal M}\pi$. Both are vector bundles on $E$ with vector bundle projections denoted $\nu^0: {\mathcal M}^0\pi \to E$ and  $\nu: {\mathcal M}\pi \to E$, and there is a canonical line bundle projection $\mu: {\mathcal M}\pi \to {\mathcal M}^0\pi$. The main property of the extended multi-momentum bundle is that it admits a canonical multisymplectic structure $\omega_{{\mathcal M}\pi}$ of degree $m+1$ ($m$ being the dimension of $M$). The restricted multi-momentum bundle ${\mathcal M}^0\pi$, however, does not admit a canonical multisymplectic structure.

An important difference with Hamiltonian Mechanics is that for Hamiltonian Classical Field Theories we don't have a Hamiltonian function, but a Hamiltonian section $h: {\mathcal M}^0\pi \to {\mathcal M}\pi$ of the canonical projection  $\mu: {\mathcal M}\pi \to {\mathcal M}^0\pi$, see \cite{CaCaIb}. The corresponding evolution equations are the Hamilton-deDonder-Weyl equations. They form a system of partial differential equations of first order on ${\mathcal M}^0\pi$ with space of parameters $M$ and whose solutions are sections of the projection $\pi \circ \nu^0: {\mathcal M}^0\pi \to M$. More precisely, if we denote by
\[
h(x^{i}, u^{\alpha}, p_{\alpha}^{i}) = (x^{i}, u^{\alpha}, -H(x, u, p), p_{\alpha}^{i}),
\]
the local expression of the Hamiltonian section $h$, then the Hamilton-deDonder-Weyl equations are locally 
\begin{equation}\label{HdDW_intro} 
\begin{aligned} 
&\frac{\partial u^{\alpha}}{\partial x^{i}} = \frac{\partial H}{\partial p_{\alpha}^{i}}, \; \; \mbox{ for all } i \mbox{ and } \alpha,\\
&\sum_{i}\frac{\partial p_{\alpha}^{i}}{\partial x^{i}} = -\frac{\partial H}{\partial u^{\alpha}}, \; \; \mbox{ for all } \alpha.
\end{aligned} 
\end{equation} 
These equations go back, at least, to work of Volterra \cite{Vo1890a,Vo1890b}.
In the literature, we can find the following geometric descriptions of the Hamilton-deDonder-Weyl equations \eqref{HdDW_intro} associated to $h$:
\begin{itemize}
\item
From $h$ and the canonical multisymplectic structure one can produce a {\bf non-canonical} multisymplectic structure $\omega_h = h^*(\omega_{{\mathcal M}\pi})$ on ${\mathcal M}^0\pi$. Then, using $\omega_h$, a special type of Ehresmann connections can be introduced on the fibration $\pi\circ \nu^0: {\mathcal M}^0\pi \to M$ (which, in the present paper, will be called Hamiltonian connections for $h$) and the solutions of the evolution equations are the  integral sections of these connections (see \cite{DeMaMa2,DeMaMa1}; see also \cite{EcMuRo1,EcMuRo2}). The multisymplectic structure $\omega_h$ can also be used to directly characterize the sections of the projection $\pi \circ \nu^0: {\mathcal M}^0\pi \to M$ which are solutions of the Hamilton-deDonder-Weyl equations for $h$ (see \cite{CaCaIb,DeMaMa2,DeMaMa1,EcMuRo1,EcMuRo2}). From $\omega_h$, one can also define the "multisymplectic pseudo-brackets" and "multisymplectic brackets" of $(m-1)$-Hamiltonian forms which may be considered the field version of the Poisson bracket for functions in Classical Mechanics (see \cite{FoSa}). 

\item
Using an auxiliary Ehresmann connection $\nabla$ on the configuration bundle $\pi: E \to M$ and the Hamiltonian section $h$, one may produce a Hamiltonian energy $H^\nabla: {\mathcal M}^0\pi \to \mathbb{R}$ associated to $h$ and $\nabla$ and a {\bf non-canonical} multisymplectic structure on ${\mathcal M}^0\pi$ which allow us to describe the solutions of the evolution equations in a geometric form (see \cite{CaCaIb}; see also \cite{EcMuRo1,EcMuRo2}). In addition, it is possible to consider a suitable space of currents (a vector subspace of $m-1$-forms on ${\mathcal M}^0\pi$ which are horizontal with respect to the projection $\pi\circ \nu^0: {\mathcal M}^0\pi \to M$) and one may introduce a ``Poisson bracket" of a current and an Hamiltonian energy associated to $\nabla$. This bracket allows the description of the Hamilton-deDonder-Weyl equations as in Hamiltonian Mechanics (see \cite{CaMa}).
\item
The Hamiltonian section $h$ induces a canonical extended Hamiltonian density ${\mathcal F}_h$, which is a smooth $\pi^*(\Lambda^mT^*M)$-valued function defined on ${\mathcal M}\pi$ see \cite{CaGuMa,Gr}; see also \cite{EcLeMuRo} for the particular case when a volume form on $M$ is fixed. 
Then, using ${\mathcal F}_h$ and the canonical multisymplectic structure $\omega_{{\mathcal M}\pi}$ one write intrinsically a system of partial differential equations on ${\mathcal M}\pi$ whose solutions are sections of the fibration $\pi \circ \nu: {\mathcal M}\pi \to M$. The projection, via $\mu$, of these sections are the solutions of the Hamilton-deDonder-Weyl equations for $h$ (see \cite{EcLeMuRo}).

\item
From the configuration bundle $\pi: E \to M$ one can construct the phase bundle $\mathbb{P}(\pi)$, an affine bundle over ${\mathcal M}^0\pi$, and the differential $dh: {\mathcal M}^0\pi \to \mathbb{P}(\pi)$ of the Hamiltonian section $h$, as a section of $\mathbb{P}(\pi)$. In addition, an affine bundle epimorphism $A: J^1(\pi \circ \nu^0) \to \mathbb{P}(\pi)$ from the $1$-jet bundle of the fibration $\pi \circ \nu^0: {\mathcal M}^0\pi \to M$ onto $\mathbb{P}(\pi)$ may be also introduced. Then, the solutions of the Hamilton-deDonder-Weyl equations are the sections $s^0: M \to {\mathcal M}^0\pi$ whose first prolongation $j^1s^0: M \to J^1(\pi \circ \nu^0)$ is contained in the submanifold $A^{-1}(dh({\mathcal M}^0\pi))$ (see \cite{Gr,GrGr}; see also \cite{GrGrUr2,MaMeSa} for the particular case of time-dependent Hamiltonian Mechanics).
\end{itemize}

\subsection{The problem}\label{the-problem}
The previous comments lead naturally to the following question: 

\noindent {\it Does there exist a completely canonical geometric formulation of the Hamilton-deDonder-Weyl equations which is analogous to the standard Poisson bracket formulation of time-independent Hamiltonian Mechanics?}

A possible answer to this question could be the geometric formulation developed in \cite{CaMa} (see Section \ref{Previous-contributions}). However, this formulation is not canonical since most of the constructions in \cite{CaMa}
%
%
%
depend on the chosen auxiliary connection in the configuration bundle. In fact, in a previous paper \cite{MaSh} Marsden and Shkoller justify the use of this connection in the geometric formulation of the theory and one may find, in that paper (see \cite{MaSh}, page 554), the following cite:
   
{\it It is interesting that the structure of connection is not necessary to intrinsically define the
Lagrangian formalism (as shown in the preceding references), while for the intrinsic definition of a
covariant Hamiltonian the introduction of such a structure is essential. Of course, one can avoid
a connection if one is willing to confine ones attention to local coordinates.}   

However, in our paper, we will construct a bracket that does not use any auxiliary objects such as a connection in the configuration bundle and which is completely canonical, thereby giving an affirmative answer to the question above. 

\subsection{Answer to the problem and contributions of the paper}
In order to give an affirmative answer to the question in Section \ref{the-problem}, we will use the following previous contributions and results:
\begin{itemize}
\item
The construction of the phase space $\mathbb{P}(\pi)$ associated with the configuration bundle $\pi: E \to M$ and the differential of a Hamiltonian section $h: {\mathcal M}^0\pi \to {\mathcal M}\pi$ as a section of $\mathbb{P}(\pi)$ (see \cite{Gr}).

\item
The affine bundle epimorphism $A: J^1(\pi \circ \nu^0) \to \mathbb{P}(\pi)$ which was also introduced in \cite{Gr}.

\item
The notion of a Hamiltonian connection associated with a Hamiltonian section $h$. This type of objects were already considered in \cite{DeMaMa2,DeMaMa1} in order to characterize the solutions of the Hamilton-deDonder-Weyl equations for $h$ (although the authors of these papers did not use the terminology of a Hamiltonian connection).
\end{itemize}
We will combine the previous constructions as follows.

As a first step, we consider the affine bundle isomorphism
\[
\sharp^{\rm aff}: \mathbb{P}(\pi) \to J^1(\pi \circ \nu^0)/\operatorname{Ker} A
\]
where $\operatorname{Ker} A$ is the kernel of the affine bundle epimorphism $A: J^1(\pi \circ \nu^0) \to \mathbb{P}(\pi)$ and $\sharp^{\rm aff} = \hat{A}^{-1}$, with $\hat{A}: J^1(\pi \circ \nu^0)/\operatorname{Ker} A \to \mathbb{P}(\pi)$ the affine bundle isomorphism induced by $A$. Then, we introduce the section 
\[
\Gamma_h: {\mathcal M}^0\pi \to J^1(\pi \circ \nu^0)/\operatorname{Ker} A 
\]
of $J^1(\pi \circ \nu^0)/\operatorname{Ker} A$, 
given by
\[
\Gamma_h = \sharp^{\rm aff} \circ dh
\]
and we prove the following result: the section $\Gamma_h$ is canonically identified to the equivalence class of Ehresmann connections on the fibration $\pi \circ \nu^0: {\mathcal M}^0\pi \to M$ that are Hamiltonian connections for $h$, see Theorem \ref{Gamma-h-Ham-Connections-h}.

So, the section $\Gamma_h$ associated to the Hamiltonian section $h$ is the field-theoretic analogue to the Hamiltonian vector field $X_H$ associated to a Hamiltonian function $H$ in Classical Mechanics.
The following commutative diagram illustrates the situation  
\[
\xymatrix{{\mathbb{P}(\pi)}\ar[rr]^{\sharp^{\rm aff}}&&J^{1}(\pi \circ \nu^0)/
\operatorname{Ker} A\\&
{\mathcal M}^0\pi\ar[ul]^{dh}\ar[ur]_{\Gamma_h}&}
\]
The analogy with the corresponding diagram given in \ref{diagram-1} for classical mechanics is evident.

The next step is to introduce a suitable space of currents ${\mathcal O}$ (a vector subspace of $(m-1)$-forms on ${\mathcal M}^0\pi$ which are horizontal with respect to the fibration $\pi \circ \nu^0: {\mathcal M}^0\pi \to M$), in such a way that the restriction of the standard exterior differential to ${\mathcal O}$ takes values in the space of sections of the vector bundle $(J^1(\pi \circ \nu^0)/\operatorname{Ker} A)^+ = {\rm Aff}(J^1(\pi \circ \nu^0)/\operatorname{Ker} A, (\pi \circ \nu^0)^*(\Lambda^mT^*M))$, that is, we have the linear map
\[
d: {\mathcal O} \to \Gamma(J^1(\pi \circ \nu^0)/\operatorname{Ker} A)^+.
\]
The dual vector bundle $(J^1(\pi \circ \nu^0)/\operatorname{Ker} A)^+ $ is chosen so that these differentials can be canonically paired with the Hamiltonian connections $\Gamma _h$, thereby extending to the field-theoretic context the pairing $ \left\langle dF, X_H \right\rangle $ between the differential $dF$ of an observable and the Hamiltonian vector field $X_H$, see \eqref{Canonical_Poisson}.
This is our motivation for introducing the space of currents ${\mathcal O}$ and although it is different to the motivation in \cite{CaMa}, ${\mathcal O}$ just coincides with the space of currents in \cite{CaMa} (see Remark \ref{Currents-Marco-Jerry}).

Now, if $\Gamma(\mu)$ is the space of Hamiltonian sections, we can define the linear-affine canonical bracket
\begin{equation}\label{Lin_Aff_Bracket} 
\{\cdot, \cdot\}: {\mathcal O} \times \Gamma(\mu) \to \Gamma((\pi \circ \nu^0)^*(\Lambda^mT^*M)),
\end{equation} 
given by
\[
\{\alpha^0, h\} = \langle d\alpha^0, \Gamma_h \rangle = \langle d\alpha^0, \sharp^{\rm aff}(dh) \rangle, \; \; \mbox{ for } \alpha^0 \in {\mathcal O}.
\]
Then, one may prove that the evolution of any current  $ \alpha^0 \in \mathcal{O}$ along a solution $s^0: M \to {\mathcal M}^0\pi$ of the Hamilton-deDonder-Weyl equations for $h$ is given by
\begin{equation}\label{Poisson_formulation_fields} 
(s^0)^*(d\alpha^0) = \{\alpha^0, h\} \circ s^0.
\end{equation} 
Conversely, if $s^0: M \to {\mathcal M}^0\pi$ is such that \eqref{Poisson_formulation_fields} holds for all $\alpha^0 \in \mathcal{O}$, then $s^0$ is a solution of Hamilton-deDonder-Weyl equations. The canonical bracket formulation \eqref{Poisson_formulation_fields} is the field-theoretic analogue to the canonical Poisson bracket formulation \eqref{Poisson_formulation} of classical mechanics.

The previous tasks are performed in Section \ref{Currents-Ham-deDonder-Weyl} (see Theorem \ref{evolution-current}).
Here again, the analogy with the canonical Poisson formulation of classical mechanics (see (\ref{Poisson_formulation}) and (\ref{Canonical_Poisson})) is evident.

It is important to note the affine character of the canonical bracket $\{\cdot, \cdot\}$ in \eqref{Lin_Aff_Bracket}: the space $\Gamma(\mu)$ of Hamiltonian sections is an affine space modelled over the vector space $\Gamma((\pi \circ \nu^0)^*(\Lambda^mT^*M))$. Recalling that the canonical Poisson bracket on $T^*Q$ induces a Lie algebra structure on $C^{\infty}(T^*Q)$, a new question arises: 

{\it What are the algebraic properties of the bracket $\{\cdot, \cdot\}: {\mathcal O} \times \Gamma(\mu) \to \Gamma((\pi \circ \nu^0)^*(\Lambda^mT^*M))$?}

Related to this question, we will prove that ${\mathcal O}$ admits a canonical Lie algebra structure $\{\cdot, \cdot\}_{\mathcal O}$ (see Theorem \ref{Lie-algebra-O}) and that the linear map
\[
{\mathcal R}: {\mathcal O} \to {\rm Aff}\left(\Gamma(\mu), \Gamma((\pi \circ \nu^0)^*(\Lambda^mT^*M))\right)
\]
defined by
\[
{\mathcal R}(\alpha^0) = \{\alpha^0, \cdot \}
\]
is a representation of the Lie algebra $({\mathcal O}, \{\cdot, \cdot\}_{\mathcal O})$ on the affine space $\Gamma(\mu)$ (see Theorem \ref{affine-representation} and the property \eqref{affine_linear_Jacobi}).

The previous results will be applied to the following examples: time-dependent Hamiltonian systems, Continuum Mechanics (including fluid dynamics and nonlinear elasticity) and Yang-Mills theories. Some of the constructions developed in the paper are illustrated in the Diagram in Appendix \S\ref{diagram}.


\subsection{Structure of the paper}
The paper is structured as follows. In Section \ref{Ham-Class-Field-Theo}, we review the geometric formulation of the Hamilton-deDonder-Weyl equations using the multisymplectic structure on the phase space induced by the Hamiltonian section. In Section \ref{new-geom-formul-HCFT}, we introduce the canonical linear-affine bracket $\{\cdot, \cdot\}: {\mathcal O} \times \Gamma(\mu) \to \Gamma((\pi \circ \nu^0)^*(\Lambda^mT^*M))$ and we formulate the Hamilton-deDonder-Weyl equations using this bracket. In particular, we describe the evolution of a current along a solution of the Hamilton-deDonder-Weyl equations. In Section \ref{aff-repre-Lie-algebra-aff-space}, we introduce a Lie algebra structure on ${\mathcal O}$ and we prove that $\{\cdot, \cdot\}$ induces a representation of the Lie algebra ${\mathcal O}$ on the affine space $\Gamma(\mu)$ of Hamiltonian sections. In Section \ref{examples}, we apply the previous results to several examples. The paper closes with three appendices. In the first one, we review the definition of the $1$-jet bundle associated with a fibration, in the second one, we discuss the vertical  
lift of a section of a vector bundle as a vertical vector field on the total space and, in the third one, we present a Diagram which illustrates most of the relevant constructions in the paper.

\section{Hamiltonian Classical Field Theories of first order}\label{Ham-Class-Field-Theo}

In this section, we review some basic constructions and results on Hamiltonian Classical Field Theories of first order (for more details, see \cite{CaCaIb}).

\subsection{The restricted and extended multimomentum bundle associated with a fibration}\label{res-ext-multimomentum-bundle}

The \textit{configuration bundle} of a classical field theory is a fibration $\pi: E \to M$, that is, a surjective submersion from $E$ to $M$. We assume $dim \; M = m$ and $dim \; E = m+n$.

The \textit{extended multimomentum bundle} ${\mathcal M}\pi$ associated with the configuration bundle $\pi: E \to M$ is the vector bundle over $E$ whose fiber at the point $y \in E$ is
\[
{\mathcal M}_y\pi = \{ \varphi: J^1_y\pi \to \Lambda^mT^*_{\pi(y)}M \mid \varphi \mbox{ is affine }\}.
\]
Here, $J^1\pi = \cup_{y \in E}J^1_y\pi$ is the $1$-jet bundle of the fibration $\pi: E \to M$ (see Appendix \ref{1-jet-bundle}).

It is well-known that ${\mathcal M}\pi$ may be identified with the vector bundle $\Lambda_2^m(T^*E)$ over $E$, whose fiber at $y \in E$ is
\[
\Lambda^m_2(T_y^*E) = \{ \gamma \in \Lambda^m(T_y^*E) \mid i_ui_{u'} \gamma = 0, \forall u, u' \in V_y\pi\}.
\]  
In fact, if $\gamma \in \Lambda_2^m (T_y^*E)$ and $z: T_{\pi(y)}M \to T_yE \in J^1_y \pi$ then
\begin{equation}\label{duality}
\langle \gamma, z\rangle  = \Lambda^m z^*(\gamma).
\end{equation}
If $(x^{i}, u^{\alpha})$ are local coordinates on $E$ which are adapted with the fibration $\pi$, then $\gamma \in \Lambda^m_2(T_y^*E)$ reads locally
\[
\gamma = p \,d^m x + p^{i}_{\alpha} du^{\alpha} \wedge d^{m-1}x_i,
\]
where
\[ 
d^mx = dx^{1} \wedge \dots \wedge dx^m \; \; \mbox{ and } \; \; d^{m-1}x_j = i(\frac{\partial}{\partial x^j})d^mx.
\] 
So, $(x^{i}, u^{\alpha}, p, p_{\alpha}^{i})$ are local coordinates on ${\mathcal M}\pi$.

On ${\mathcal M}\pi \simeq \Lambda^m_2(T^*E)$ we can define a canonical $m$-form $\lambda_{{\mathcal M}\pi}$ as follows
\begin{equation}\label{Canonical-m-form}
\lambda_{{\mathcal M}\pi}(\gamma)(Y_1, \dots, Y_m) = \gamma((T_{\gamma}\nu)(Y_1), \dots , (T_{\gamma}\nu)(Y_m)),
\end{equation}
for $\gamma \in \Lambda_2^m(T^*E)$ and $Y_1, \dots, Y_m \in T_{\gamma}\Lambda^m_2(T^*E)$, with $\nu: {\mathcal M}\pi \to E$ the vector bundle projection.

From (\ref{Canonical-m-form}), $\lambda_{{\mathcal M}\pi}$ has the local expression
\[
\lambda_{{\mathcal M}\pi} = p d^mx + p_\alpha^{i} du^{\alpha} \wedge d^{m-1}x_i.
\]
The \textit{canonical multisymplectic structure} $\omega_{{\mathcal M}\pi}$ on ${\mathcal M}\pi$ is the $(m+1)$-form given by
\[
\omega_{{\mathcal M}\pi} = -d\lambda_{{\mathcal M}\pi}.
\]
Locally, we have
\begin{equation}\label{Local-multi-sym}
\omega_{{\mathcal M}\pi} =  -dp \wedge d^mx + du^{\alpha} \wedge dp_\alpha^{i} \wedge d^{m-1}x_i.
\end{equation}
It is clear that $\omega_{{\mathcal M}\pi}$ is closed and non-degenerate, that is, the vector bundle morphism
\[
\flat_{\omega_{{\mathcal M}\pi}}: T({\mathcal M}\pi) \to \Lambda^m(T^*({\mathcal M}\pi)), \; \; Y \in T_{\gamma}({\mathcal M}\pi) \mapsto i_Y\omega_{{\mathcal M}\pi}(\gamma) \in \Lambda^mT^*_\gamma({\mathcal M}\pi),
\]
is a linear monomorphism.

The \textit{restricted multimomentum bundle} ${\mathcal M}^0\pi$ is the vector bundle 
\[
{\mathcal M}^0\pi = (\pi^*(TM) \otimes V^*\pi) \otimes \pi^*(\Lambda^mT^*M) \simeq \pi^*(\Lambda^{m-1}T^*M) \otimes V^*\pi \simeq {\rm Lin}(\pi^*(T^*M) \otimes V\pi, \pi^*(\Lambda^mT^*M)).
\]
Local coordinates on ${\mathcal M}^0\pi$ are $(x^{i}, u^{\alpha}, p_{\alpha}^{i})$.

There is a canonical projection $\mu: {\mathcal M}\pi \to {\mathcal M}^0\pi$ given by
\[
\mu(\gamma) = \gamma^{l},
\]
for $\gamma: J^1_y\pi \to \Lambda^{m}T^*_{\pi(y)}M \in {\mathcal M}_y\pi$, where $\gamma^l: T^*_{\pi(y)}M \otimes V_y\pi \to \Lambda^{m}T^*_{\pi(y)}M$ is the linear map associated with $\gamma$.
The local expression of $\mu$ is
\[
\mu(x^{i}, u^{\alpha}, p, p_{\alpha}^{i}) = (x^{i}, u^{\alpha}, p_{\alpha}^{i}).
\]
Note that if $\gamma, \gamma' \in {\mathcal M}\pi = \Lambda_2^mT^*E$ then 
\begin{equation}\label{Fibers-mu}
\mu(\gamma) = \mu(\gamma') \Leftrightarrow \nu(\gamma) = \nu(\gamma') \mbox{ and } \exists ! \; \Omega\in \Lambda^mT^*_{\pi(\nu(\gamma))}M \mbox{ such that } \gamma' = \gamma + (\Lambda^mT^*_{\nu(\gamma)}\pi)(\Omega).
\end{equation}

\begin{remark}
Note that this last statement implies that $\mathcal{M}^0\pi\simeq {\mathcal M}\pi/\sim$. This situation recalls a particular case in the Poisson realm, the quotient of a symplectic manifold by a proper and free action of a symmetry Lie group inherits a Poisson structure. In this formalism, the quotient of the extended multimomentum bundle also has a  new version of  a multi-Poisson structure (see for example \cite{BCI}) which is defined via a Lie algebroid structure on a subbundle of $\Lambda^m T^*\mathcal{M}^0\pi$, when the base manifold $M$ is orientable. Note that, in such a case, if we fix a volume form on $M$, we have an action of the real line $\mathbb{R}$ on $\mathcal{M}^0\pi$, which preserves the multisymplectic structure, and $\mathcal{M}\pi$ is the space of orbits of this action. For the definition and details on the construction of the multi-Poisson structure,  we refer to \cite{BMR}.
\end{remark}

\subsection{Hamilton-deDonder-Weyl equations}

Given a configuration bundle $\pi: E \to M$, a \textit{Hamiltonian section} $h: {\mathcal M}^0\pi \to {\mathcal M}\pi$ is a smooth section of the canonical projection 
\[
\mu: {\mathcal M}\pi \to {\mathcal M}^0\pi.
\]
The local expression of $h$ is
\[
h(x^{i}, u^{\alpha}, p_{\alpha}^{i}) = (x^{i}, u^{\alpha}, -H(x, u, p), p_{\alpha}^{i}),
\]
where $H$ is a local real $C^{\infty}$-function on ${\mathcal M}^0\pi$.

Using the Hamiltonian section, we can define the $(m+1)$-form $\omega_h$ on ${\mathcal M}^0\pi$ given by
\begin{equation}\label{w-h}
\omega_h = h^*(\omega_{{\mathcal M}\pi}).
\end{equation}
The local expression of $\omega_h$ is 
\begin{equation}\label{Local-w-h}
\omega_h = du^{\alpha} \wedge dp_{\alpha}^{i} \wedge d^{m-1}x_i + dH \wedge d^mx = du^{\alpha} \wedge dp_{\alpha}^{i} \wedge d^{m-1}x_i + \frac{\partial H}{\partial u^{\alpha}} du^{\alpha} \wedge d^mx + \frac{\partial H}{\partial p_{\alpha}^{i}} dp_{\alpha}^{i} \wedge d^mx.
\end{equation} 
Note that if $m = 1$, $\omega_h$ is degenerate and the rank of its kernel is $1$. On the other hand, if $m \geq 2$, $\omega_h$ is non-degenerate which implies that it is multisymplectic.

\begin{proposition}\label{Prop_HdDW}
A (local) section $s^0: U\subseteq M \to {\mathcal M}^0\pi$ of the projection $\pi \circ \nu^0: {\mathcal M}^0\pi \to M$ is a solution of the Hamilton-deDonder-Weyl equations iff
\[
(s^0)^*(i_U\omega_h) = 0, \; \; \; \forall\; U \in \Gamma(V(\pi \circ \nu^0)),
\]
where $\Gamma(V(\pi \circ \nu^0))$ is the space of sections of the vertical bundle to $\pi \circ \nu^0: {\mathcal M}^0\pi \to M$.
\end{proposition}

\begin{proof} Using the local expression $s^0(x^{i}) = (x^{i}, u^{\alpha}(x), p_{\alpha}^{i}(x))$, 
it is routine to verify that $(s^0)^*(i_U\omega_h) = 0$ if and only if
\[
\frac{\partial u^{\alpha}}{\partial x^{i}} = \frac{\partial H}{\partial p_{\alpha}^{i}}, \; \; \mbox{ for all } i \mbox{ and } \alpha,
\]
\[
\sum_{i}\frac{\partial p_{\alpha}^{i}}{\partial x^{i}} = -\frac{\partial H}{\partial u^{\alpha}}, \; \; \mbox{ for all } \alpha.
\]
\end{proof}
\section{A new canonical bracket formulation of Hamiltonian Classical Field Theories of first order}\label{new-geom-formul-HCFT}


As we reviewed in the Introduction, the phase space of momenta for classical mechanics is the cotangent bundle $T^*Q$ of the configuration space $Q$, a smooth manifold of dimension $n$. The cotangent bundle $T^*Q$ carries a canonical symplectic structure which induces a vector bundle isomorphism $\flat: T(T^*Q) \to T^*(T^*Q)$ over the identity 
with inverse denoted $\sharp :T^*(T^*Q) \to T(T^*Q)$. The Hamiltonian is a real $C^{\infty}$-function on $T^*Q$ and the Hamiltonian vector field is given in terms of the differential $dH$ and the vector bundle isomorphism $\sharp $ as $X_H= \sharp dH$, see Diagram \eqref{diagram-1}.
For field theories, we don't have a Hamiltonian function, but a Hamiltonian section
\[
h: {\mathcal M}^0\pi \to {\mathcal M}\pi
\]
of the canonical projection $\mu: {\mathcal M}\pi \to {\mathcal M}^0\pi$.
So, the following questions arise when extending the previous construction to field theories:

{\bf Question 1}: What is the differential of $h$?

{\bf Question 2}: Where does the differential of $h$ take values? 

We will answer these questions in \S\ref{sec:phase bundle} by showing that the differential $dh$ of $h$ is a section of the phase bundle $\mathbb{P}(\pi)$ associated with the fibration $\mu: {\mathcal M}\pi \to {\mathcal M}^0\pi$. The bundle $\mathbb{P}(\pi)$ was introduced in \cite{Gr} and was used there to discuss a Tulczyjew triple for Classical Field Theories of first order.
This will allow us to define the field-theoretic analogue to the vector bundle isomorphism $\sharp$ in \S\ref{canonical-isomorphism} and the field-theoretic analogue $\Gamma_h$ to the Hamiltonian vector field $X_H$ in \S\ref{Gamma_h}. In particular, we will show that $\Gamma_h$ can be identified with the equivalence class of Hamiltonian Ehresmann connections associated to $h$.

\medskip

Going back to Classical Hamiltonian Mechanics, we recall that the set of observables is the space $C^\infty(T^*Q)$ and that the Hamilton equations can be equivalently formulated in the Poisson bracket form \eqref{Poisson_formulation} with respect to the canonical Poisson bracket giving by the formulas \eqref{Canonical_Poisson}. In view of this formulation, we need to find a suitable space of currents for field theories (the observables in field theories) such that their differentials take values in a bundle dual to the target bundle of $\Gamma_h$, this is the goal of \S\ref{Currents-Ham-deDonder-Weyl}. From this a canonical bracket can be obtained between currents and Hamiltonian sections. This construction is carried out in  \S\ref{suitable_LA_bracket}.

\subsection{The  phase bundle associated with a fibration and the differential of a Hamiltonian section}\label{sec:phase bundle}
Let $\pi: E \to M$ be the configuration bundle of the field theory and $h: {\mathcal M}^0\pi \to {\mathcal M}\pi$ be the Hamiltonian section. Then, although
\[
dim \; {\mathcal M}\pi = dim \; {\mathcal M}^0\pi + 1,
\]
$h$ cannot  be identified, in general, with a real $C^{\infty}$-function on ${\mathcal M}^0\pi$.
However, to $h$ we can associate an \textit{extended Hamiltonian density}
\[
{\mathcal F}_h: {\mathcal M}\pi \to \pi^*(\Lambda^mT^*M),
\]
defined as follows. If $\gamma \in {\mathcal M}\pi$ we have $\mu(\gamma) = \mu(h(\mu(\gamma)))$ and hence using (\ref{Fibers-mu}), we conclude that there exists a unique  $\Omega \in \Lambda^mT^*_{\pi(\nu(\gamma))}M$ such that $\gamma = h(\mu(\gamma)) + (\Lambda^mT^*_{\nu(\gamma)}\pi)(\Omega)$. We thus define
\[
{\mathcal F}_h(\gamma) = \Omega.
\]


\subsubsection{The differential of $\mathcal{F}_h$ and the extended phase bundle}

Note that ${\mathcal F}_h$ may be considered, in a natural way, as a $m$-form on ${\mathcal M}\pi$. Thus, we can take its exterior differential and we obtain a $(m+1)$-form on ${\mathcal M}\pi$ which is a section of the vector bundle 
\[
\Lambda^{m+1}_2(T^*({\mathcal M}\pi)) \to {\mathcal M}\pi.
\]
Now, it is easy to prove that the vector bundles $\Lambda^{m+1}_2(T^*({\mathcal M}\pi))$ and $V^*(\pi \circ \nu) \otimes (\pi \circ \nu)^*(\Lambda^mT^*M)$ are isomorphic. In fact, an isomorphism 
\[
\Psi: \Lambda^{m+1}_2(T^*({\mathcal M}\pi)) \to V^*(\pi \circ \nu) \otimes (\pi \circ \nu)^*(\Lambda^mT^*M)
\]
is given by
\[
\Psi(\tilde{\theta}): V_{\gamma}(\pi \circ \nu) \to \Lambda^m(T^*_{\pi(\nu(\gamma))}M), \; \; \tilde{U} \in V_{\gamma}(\pi \circ \nu) \to i_{\tilde{U}}\tilde{\theta} \in \Lambda^m(T^*_{\pi(\nu(\gamma))}M),
\]
for $\tilde{\theta} \in \Lambda^{m+1}_2(T^*_{\gamma}({\mathcal M}\pi))$ and $\gamma \in {\mathcal M}\pi$.
Note that $i_{\tilde{U}}\tilde{\theta} \in \Lambda^{m}_1(T^*_{\gamma}({\mathcal M}\pi))$, therefore, it induces an element of $\Lambda^m(T^*_{\pi(\nu(\gamma))}M)$.

We denote by $d^v{\mathcal F}_h$ the section of the vector bundle $V^*(\pi \circ \nu) \otimes (\pi \circ \nu)^*(\Lambda^mT^*M) \to {\mathcal M}\pi$ induced by the differential of ${\mathcal F}_h$. In local coordinates, if
\[
h(x^{i}, u^{\alpha}, p_{\alpha}^{i}) = (x^{i}, u^{\alpha}, -H(x, u, p), p_\alpha^{i})
\]
then
\begin{equation}\label{extended-Ham-den}
{\mathcal F}_h(x^{i}, u^{\alpha}, p, p_{\alpha}^{i}) = (p + H(x^{i}, u^{\alpha}, p_{\alpha}^{i}))d^mx
\end{equation}
and
\begin{equation}\label{d-v-F}
d^v{\mathcal F}_h(x^{i}, u^{\alpha}, p, p_{\alpha}^{i}) = \Big(dp + \frac{\partial H}{\partial u^{\alpha}} du^{\alpha} + \frac{\partial H}{\partial p_{\alpha}^{i}}dp_{\alpha}^{i}\Big)\otimes d^mx.
\end{equation}

Note that if $\mu: {\mathcal M}\pi \to {\mathcal M}^0\pi$ is the canonical projection, $\Phi$ is a $m$-form on $M$ and $\Phi^{\bf v} \in \Gamma(V\mu)$ is the vertical lift to ${\mathcal M}\pi$ (see Appendix \ref{ver-lift-sec}) then, using (\ref{d-v-F}) and (\ref{vertical-lift-dmx}), we deduce that
\begin{equation}\label{property_dv}
(d^v{\mathcal F}_h)(\Phi^{\bf v}) = (\pi \circ \nu)^*(\Phi).
\end{equation} 
This property of $d^v{\mathcal F}_h$ motivates the definition of the following affine subbundle of $V^*(\pi \circ \nu) \otimes (\pi \circ \nu)^*(\Lambda^mT^*M)$.

\begin{definition} The \textbf{extended phase bundle} of the configuration bundle $\pi: E \to M$ is the affine subbundle $\widetilde{\mathbb{P}(\pi)}$ of $V^*(\pi \circ \nu) \otimes (\pi \circ \nu)^*(\Lambda^mT^*M)$ whose fiber at the point $\gamma \in {\mathcal M}\pi$ is
\begin{equation}\label{P-Pi-tilde}
\widetilde{\mathbb{P}(\pi)}(\gamma) = \big\{ \tilde{\mathcal A}\in \operatorname{Lin}( V_{\gamma}(\pi \circ \nu), \Lambda^mT^*_{\pi(\nu(\gamma))}M)  
\mid \tilde{\mathcal A} (\Omega^{\bf v}(\gamma)) = \Omega, \forall \;\Omega \in \Lambda^mT^*_{\pi(\nu(\gamma))}M\big\}.
\end{equation}
\end{definition}
From \eqref{property_dv} we have
\[
d^v{\mathcal F}_h \in \Gamma(\widetilde{\mathbb{P}(\pi)}).
\]
Note that $\widetilde{\mathbb{P}(\pi)}$ is modelled over the vector bundle $V(\widetilde{\mathbb{P}(\pi)})$ whose fiber at the point $\gamma \in {\mathcal M}\pi$ is
\[
V(\widetilde{\mathbb{P}(\pi)})(\gamma) = \{ \tilde{\nu}\in \operatorname{Lin}(V_{\gamma}(\pi \circ \nu) , \Lambda^mT^*_{\pi(\nu(\gamma))}M) \mid \tilde{\nu}(\Omega^{\bf v}(\gamma)) = 0, \forall \;\Omega \in \Lambda^mT^*_{\pi(\nu(\gamma))}M \}.
\]
We remark that an element $\tilde{\mathcal A}$ of  $\widetilde{\mathbb{P}(\pi)}$ has the following local form
\[
\tilde{\mathcal A} = (dp + \tilde{\mathcal A}_\alpha du^{\alpha} + \tilde{\mathcal A}_{i}^{\alpha} dp_{\alpha}^{i})\otimes d^mx
\]
and a generic element $\tilde{\nu}$ of $V(\widetilde{\mathbb{P}(\pi)})$ has the local form
\[
\tilde{\nu} =  (\tilde{\mathcal A}_\alpha du^{\alpha} + \tilde{\mathcal A}_{i}^{\alpha} dp_{\alpha}^{i})\otimes d^mx.
\]
Therefore, the local coordinates on $\widetilde{\mathbb{P}(\pi)}$ and $V(\widetilde{\mathbb{P}(\pi)})$ are $(x^{i}, u^{\alpha}, p, p_{\alpha}^{i}; \tilde{\mathcal A}_\alpha, \tilde{\mathcal A}_{i}^{\alpha})$. In addition, 
\begin{equation}\label{d-v-F-h-local}
(d^v{\mathcal F}_h)(x^{i}, u^{\alpha}, p, p_{\alpha}^{i}) = \Big(x^{i}, u^{\alpha}, p, p_{\alpha}^{i}; \frac{\partial H}{\partial u^{\alpha}}, \frac{\partial H}{\partial p_\alpha^{i}}\Big).
\end{equation}


\subsubsection{The differential of a Hamiltonian section and the phase bundle}

Now, given a point $x \in M$, we can consider an action of the abelian group $\Lambda^mT^*_xM$ on the fiber $(\pi \circ \nu)^{-1}(x)$ defined as follows. If $\Phi \in \Lambda^m(T^*_{x}M)$ then we define $\Phi \; \cdot: (\pi \circ \nu)^{-1}(x) \to (\pi \circ \nu)^{-1}(x)$ by
\begin{equation}\label{Def-action-fibred}
\Phi \cdot \gamma = \gamma + (\Lambda^mT^*_{\nu(\gamma)}\pi)(\Phi), \; \; \mbox{ for } \gamma \in (\pi \circ \nu)^{-1}(x).
\end{equation}
In local coordinates, we get
\[
(\lambda d^mx)\cdot (x^{i}, u^{\alpha}, p, p_{\alpha}^{i}) = (x^{i}, u^{\alpha}, p + \lambda, p_{\alpha}^{i})
\]
and thus the quotient space ${\mathcal M}\pi / \pi^*(\Lambda^mT^*M)$ may be identified with the reduced multimomentum bundle ${\mathcal M}^0\pi$.

The tangent and cotangent lift of the previous action induces a fibred action of the vector bundle $(\pi\circ \nu)^*(\Lambda^mT^*M)$ on the vector bundles $V(\pi \circ \nu)$ and $V^*(\pi \circ \nu) \otimes (\pi\circ \nu)^*(\Lambda^mT^*M)$. In fact, if $\gamma \in \Lambda^m_2(T^*_yE)$, $\tilde{U} \in V_\gamma(\pi \circ \nu)$ and $\Phi \in \Lambda^mT^*_{\pi(y)}M$ then the tangent lift is
\begin{equation}\label{Vertical-action}
\Phi \; \cdot \tilde{U} = (T_{\gamma}(\Phi \cdot))(\tilde{U}) \in V_{\Phi \cdot \gamma}(\pi \circ \nu).
\end{equation}
Note that, using (\ref{Def-action-fibred}), (\ref{Vertical-action}) and (\ref{vertical-lift-dmx}), it follows that
\begin{equation}\label{Action-Vertical-lifts}
\Phi \cdot \Omega^{\bf v}_{\gamma} = \Omega^{\bf v}_{\Phi \cdot \gamma},
\end{equation}
for $\Omega \in \Lambda^mT^*_{\pi(y)}M$. If $(x^{i}, u^{\alpha}, p, p_{\alpha}^{i}; \dot{u}^{\alpha}, \dot{p}, \dot{p}_{\alpha}^{i})$ are local coordinates on $V(\pi \circ \nu)$, we have
\begin{equation}\label{R-action-vertical} 
(\lambda d^m x)\cdot (x^{i}, u^{\alpha}, p, p_{\alpha}^{i}; \dot{u}^{\alpha}, \dot{p}, \dot{p}_{\alpha}^{i}) = (x^{i}, u^{\alpha}, p + \lambda, p_{\alpha}^{i}; \dot{u}^{\alpha}, \dot{p}, \dot{p}_{\alpha}^{i}).
\end{equation}
In a similar way, if $\tilde{\theta} \in V^{*}_{\gamma}(\pi \circ \nu) \otimes \Lambda^mT^*_{\pi(\nu(\gamma))}M$ then the cotangent lift is
\begin{equation}\label{Vertical-Dual-action}
(\Phi \cdot \tilde{\theta}) (\tilde{U}') = \tilde{\theta}(-\Phi \cdot \tilde{U}'),
\end{equation}
for $\tilde{U}' \in V_{\Phi \cdot \gamma}(\pi \circ \nu)$.
From (\ref{P-Pi-tilde}) and (\ref{Action-Vertical-lifts}), we deduce that this action restricts to the extended phase bundle $\widetilde{\mathbb{P}(\pi)}$, and to the vector bundle $V(\widetilde{\mathbb{P}(\pi)})$. In local coordinates, we have
\[
(\lambda d^m x)\cdot (x^{i}, u^{\alpha}, p, p_{\alpha}^{i}; \tilde{\mathcal A}_{\alpha}, \tilde{\mathcal A}^{\alpha}_i) =  (x^{i}, u^{\alpha}, p + \lambda, p_{\alpha}^{i}; \tilde{\mathcal A}_{\alpha}, \tilde{\mathcal A}^{\alpha}_i).
\]
Taking the quotient with respect to the action, we can introduce the following definition.

\begin{definition}\label{Def-phase-bundle} The \textbf{phase bundle} of the configuration bundle $\pi: E \to M$ is defined by
\[
\displaystyle \mathbb{P}(\pi) = \frac{\widetilde{\mathbb{P}(\pi)}}{(\pi\circ \nu)^*(\Lambda^mT^*M)}.
\]
\end{definition}

We note that $\mathbb{P}(\pi)$ is an affine bundle over ${\mathcal M}^0\pi = {\mathcal M}\pi / \pi^*(\Lambda^mT^*M)$ modelled over the vector bundle $\displaystyle V(\mathbb{P}(\pi)) = \frac{V(\widetilde{\mathbb{P}(\pi))}}{(\pi\circ \nu)^*(\Lambda^mT^*M)}$. This bundle is isomorphic to the vector bundle 
\[
V^*(\pi \circ \nu^0) \otimes (\pi \circ \nu^0)^*(\Lambda^mT^*M) = \cup_{\gamma^0 \in {\mathcal M}^0\pi} \operatorname{Lin} (V_{\gamma^0}(\pi \circ \nu^0), \Lambda^mT^*_{\pi(\nu^0(\gamma^0))}M),
\]
an isomorphism being given by
\[
{\mathcal A} \in \operatorname{Lin}(V_{\gamma^0}(\pi \circ \nu^0), \Lambda^mT^*_{\pi(\nu^0(\gamma^0))}M) \to [\tilde{\mathcal A}] \in V(\mathbb{P}(\pi))_{\gamma^0},
\]
where $\tilde{\mathcal A}: V_{\gamma}(\pi \circ \nu) \to \Lambda^mT^*_{\pi(\nu(\gamma))}M$ is defined by
\[
\tilde{\mathcal A}(\tilde{U}) = {\mathcal A}((T_{\gamma}\mu)(\tilde{U})), \; \; \; \mbox{ for } \tilde{U} \in V_{\gamma}(\pi \circ \nu)
\]
with $\gamma \in {\mathcal M}\pi$ and $\mu(\gamma) = \gamma^0$. Local coordinates on $\mathbb{P}(\pi)$ and $V({\mathbb{P}(\pi))}$ are
\begin{equation}\label{loc_coord_phase_bundle} 
(x^{i}, u^{\alpha}, p_{\alpha}^{i}; {\mathcal A}_{\alpha}, {\mathcal A}_{i}^{\alpha}).
\end{equation}

It is clear that there exists a one-to-one correspondence between the space of sections of the affine bundle $\mathbb{P}(\pi) \to {\mathcal M}^0\pi$ and the set of sections of the extended phase bundle $\widetilde{\mathbb{P}(\pi)}$ associated with $\pi$, which are $(\pi \circ \nu)^*(\Lambda^mT^*M)$-equivariant.
So, if $h: {\mathcal M}^0\pi \to {\mathcal M}\pi$ is a Hamiltonian section then, using (\ref{d-v-F-h-local}), it is easy to see that the vertical differential $d^v{\mathcal F}_h$ is $(\pi \circ \nu)^*(\Lambda^mT^*M)$-equivariant and, therefore, it induces a section
\[
dh: {\mathcal M}^0\pi \to \mathbb{P}(\pi)
\]
of the phase bundle $\mathbb{P}(\pi)$. We can thus write the following definition.

\begin{definition}\label{differential-h}
The \textbf{differential of a Hamiltonian section} $h: {\mathcal M}^0\pi \to {\mathcal M}\pi$ is the section
\[
dh: {\mathcal M}^0\pi \to \mathbb{P}(\pi)
\]
defined by the following commutative diagram
\[
\xymatrix{{\mathcal M}\pi
\ar[dd]_{\mu} \ar[rr]^{d^v{\mathcal F}_h}&    &\widetilde{\mathbb{P}(\pi)} \ar[dd]^{\tilde{\mu}}\\
  &  & \\
{\mathcal M}^0\pi \ar[rr]^{dh}&  & \mathbb{P}(\pi) }
\]
where $\mu: {\mathcal M}\pi \to {\mathcal M}^0\pi$ and $\tilde{\mu}: \widetilde{\mathbb{P}(\pi)} \to \mathbb{P}(\pi)$ are the canonical projections. 
\end{definition}
The local expression of $dh$ is
\begin{equation}\label{d-v-H}
dh(x^{i}, u^{\alpha}, p_{\alpha}^{i}) = \Big(x^{i}, u^{\alpha}, p_{\alpha}^{i}; \frac{\partial H}{\partial u^{\alpha}}, \frac{\partial H}{\partial p_{\alpha}^{i}}\Big).
\end{equation}
So, we have given an answer to {\bf Questions 1 and 2} stated above.

From the previous definition, we get the map
\[
d: \Gamma(\mu) \to \Gamma(\mathbb{P}(\pi)),\quad h \in \Gamma(\mu) \to dh \in \Gamma(\mathbb{P}(\pi)).
\]
Note that $\Gamma(\mu)$ and $\Gamma(\mathbb{P}(\pi))$ are affine spaces modelled over the vector spaces $\Gamma((\pi \circ \nu^0)^*(\Lambda^mT^*M))$ and $\Gamma(V(\mathbb{P}(\pi)))$, respectively, and $d$ is an affine map. Later in the paper, we shall use the corresponding linear map $d^l: \Gamma((\pi \circ \nu^0)^*(\Lambda^mT^*M))\to \Gamma(V(\mathbb{P}(\pi)))$ defined as follows. If ${\mathcal F}^0 \in \Gamma((\pi \circ \nu^0)^*(\Lambda^mT^*M))$ is a $(\pi \circ \nu^0)^*(\Lambda^mT^*M)$-valued function on ${\mathcal M}^0\pi$ then it may be considered as a section of the vector bundle 
\[
\Lambda^m_1(T^*({\mathcal M}^0\pi)) \to {\mathcal M}^0\pi.
\]
So, we can take the standard differential $d{\mathcal F}^0$ and we obtain a section of the vector bundle 
\[
\Lambda^{m+1}_2(T^*({\mathcal M}^0\pi)) \to {\mathcal M}^0\pi.
\]
This vector bundle is isomorphic to $V^*(\pi \circ \nu^0) \otimes (\pi \circ \nu^0)^*(\Lambda^mT^*M) \to {\mathcal M}^0\pi$, an isomorphism 
\[
\Psi^0: \Lambda^{m+1}_2(T^*({\mathcal M}^0\pi)) \to V^*(\pi \circ \nu^0) \otimes (\pi \circ \nu^0)^*(\Lambda^mT^*M)
\]
is given by
\[
\Psi^0(\theta): V_{{\gamma}^0}(\pi \circ \nu^0) \to \Lambda^m(T^*_{\pi(\nu^0(\gamma^0))}M), \; \; U \to i_{U}\theta,
\]
for $\theta \in \Lambda^{m+1}_2(T^*_{\gamma^0}({\mathcal M}^0\pi))$ and $\gamma^0 \in {\mathcal M}^0\pi$.
Note that $i_{U}\theta \in \Lambda^{m}_1(T^*_{\gamma^0}({\mathcal M}^0\pi))$ and, therefore, it induces an element of $\Lambda^m(T^*_{\pi(\nu^0(\gamma^0))}M)$. We denote by $d^l{\mathcal F}^0$ the section of the vector bundle $V^*(\pi \circ \nu^0) \otimes (\pi \circ \nu^0)^*(\Lambda^mT^*M) \to {\mathcal M}^0\pi$ induced by the differential $d{\mathcal F}^0$ via the isomorphism $\Psi^0$. If locally
\[
{\mathcal F}^0(x^{i}, u^{\alpha}, p_{\alpha}^{i}) = F^0(x^{i}, u^{\alpha}, p_{\alpha}^{i}) \otimes d^m x,
\]
the local expression of $d^l{\mathcal F}^0$ is
\begin{equation}\label{d-l-F-0}
d^l{\mathcal F}^0(x^{i}, u^{\alpha}, p_{\alpha}^{i}) = \Big(x^{i}, u^{\alpha}, p_{\alpha}^{i}; \frac{\partial F^0}{\partial u^{\alpha}}, \frac{\partial F^0}{\partial p_{\alpha}^{i}}\Big).
\end{equation}

\subsubsection{Comments on the next steps}
The differential $dh: {\mathcal M}^0\pi \to \mathbb{P}(\pi)$ of a Hamiltonian section is the field theoretic analogue to the differential $dH: T^*Q \rightarrow T^*(T^*Q)$ of a Hamiltonian function in classical mechanics.
In the next section, we will introduce a quotient affine bundle $J^1(\pi \circ \nu^0)/ \operatorname{Ker A} \to {\mathcal M}^0\pi$ which is the field theoretic analogue to the tangent bundle $T(T^*Q)$ of the phase space in classical mechanics.
Recall that using the canonical symplectic structure of $T^*Q$, one can define a canonical vector bundle isomorphism 
\[
\sharp: T^*(T^*Q) \to T(T^*Q)
\]
and the Hamiltonian vector field $X_H$ on $T^*Q$ associated with a Hamiltonian function $H \in C^{\infty}(T^*Q)$ is given by $X_H = \sharp \circ dH$.
So, a natural question arises:

{\bf Question 3}: Does there exist an affine bundle isomorphism $\sharp^{\rm aff}: \mathbb{P}(\pi) \rightarrow J^{1}(\pi \circ \nu^0)/ \operatorname{Ker} A$ which, in the presence of a Hamiltonian section $h: {\mathcal M}^0\pi \to {\mathcal M}\pi$, allows us to introduce a distinguished section $\Gamma_h$ of the affine bundle $J^1(\pi \circ \nu^0)/\operatorname{Ker} A \to {\mathcal M}^0\pi$?

In the next Section \ref{canonical-isomorphism}, we will give an affirmative answer to {\bf Question 3} and we will discuss the relation between $\Gamma_h$ and the solutions of the Hamilton-deDonder-Weyl equations for $h$.
The section $\Gamma_h$ will play the role of $X_H$ in Hamiltonian Mechanics.


\subsection{The field-theoretic analogue to the canonical isomorphism $\sharp: T^*(T^*Q) \to T(T^*Q)$}\label{canonical-isomorphism} We will show that it is given by an affine bundle isomorphism $\sharp^{\rm aff}: \mathbb{P}(\pi) \rightarrow J^{1}(\pi \circ \nu^0)/ \operatorname{Ker} A$.

Let $J^1(\pi \circ \nu^0)$ be the $1$-jet bundle associated with the fibration $\pi \circ \nu^0: {\mathcal M}^0\pi \to M$ (see Appendix \ref{1-jet-bundle}). To define the quotient affine bundle, we shall use a construction in \cite{Gr}. In this paper, the author introduced an affine bundle epimorphism
\[
A: J^{1}(\pi \circ \nu^0) \to \displaystyle \mathbb{P}(\pi) = \frac{\widetilde{\mathbb{P}(\pi)}}{(\pi\circ \nu)^*(\Lambda^mT^*M)}
\]
over the identity of ${\mathcal M}^0\pi$. This epimorphism is constructed in several steps.

First, we consider the vector bundle monomorphism
\[
\tilde{\flat}: V(\pi \circ \nu) \to {\mathcal M}(\pi \circ \nu) \simeq \Lambda^m_2(T^*{\mathcal M}\pi) \subseteq \Lambda^m(T^*{\mathcal M}\pi)
\]
induced by the canonical multisymplectic structure $\omega_{{\mathcal M}\pi}$ as follows
\begin{equation}\label{b-tilde-multi}
\tilde{\flat}(\tilde{U}) = i_{\tilde{U}}\omega_{{\mathcal M}\pi}(\gamma), \; \; \mbox{ for } \tilde{U} \in V_\gamma(\pi \circ \nu).
\end{equation}
Note that, from (\ref{Local-multi-sym}), there exists a unique $m$-form at $\mu(\gamma)$ on ${\mathcal M}^0\pi$, which we denote by $\bar{\flat}(\tilde{U})$, such that
\begin{equation}\label{3.18'}
i_{\tilde{U}}\omega_{{\mathcal M}\pi}(\gamma) = (\Lambda^mT^*_\gamma \mu)(\bar{\flat}(\tilde{U})).
\end{equation}
This defines a vector bundle morphism
\[
\bar{\flat}: V(\pi \circ \nu) \to {\mathcal M}(\pi \circ \nu^0) \simeq \Lambda^m_2(T^*{\mathcal M}^0\pi) \subseteq \Lambda^m(T^*{\mathcal M}^0\pi)
\]
over the canonical projection $\mu: {\mathcal M}\pi \to {\mathcal M}^0\pi$. Using local coordinates $(x^{i}, u^\alpha, p, p_\alpha^{i}; \dot{u^\alpha}, \dot{p}, \dot{p}_\alpha^{i})$ and $(x^{i}, u^\alpha, p_\alpha^{i}; \bar{p}, \bar{p}_\alpha^{i}, \bar{p}_i^{\alpha j})$ on $V(\pi \circ \nu)$ and ${\mathcal M}(\pi \circ \nu^0)$, respectively, the local expression of $\bar{\flat}$ is
\begin{equation}\label{Local-b-tilde}
\bar{\flat}(x^{i}, u^\alpha, p, p_\alpha^{i}; \dot{u^\alpha}, \dot{p}, \dot{p}_\alpha^{i}) = (x^{i}, u^\alpha, p_\alpha^{i}; -\dot{p}, -\dot{p}_\alpha^{i}, \bar{p}_i^{\alpha j} = \dot{u}^\alpha \delta^j_i).
\end{equation}
If $\gamma \in {\mathcal M}\pi$, the following commutative diagram
\begin{equation}\label{Diagram-Comm-1}
\xymatrix{V_\gamma(\pi \circ \nu)\ar[rr]^{\tilde{\flat}_\gamma}\ar[dr]^{\bar{\flat}_\gamma}&&{\mathcal M}_\gamma(\pi \circ \nu) = \Lambda_2^mT_\gamma^*({\mathcal M}\pi)\\&
{\mathcal M}_{\mu(\gamma)}(\pi \circ \nu^0) = \Lambda_2^mT_{\mu(\gamma)}^*({\mathcal M}^0\pi)\ar[ur]_{\Lambda^mT^*_\gamma\mu}&}
\end{equation}
illustrates the relation between $\tilde{\flat}$ and $\bar{\flat}$.

We shall now use the vector bundle morphism $\bar{\flat}$ to construct $A$. For $\gamma^0 \in {\mathcal M}^0\pi$ and $Z^0\in J^1_{\gamma^0}(\pi \circ \nu^0)$ one first defines
\[
\widetilde{A(Z^0)} \in \widetilde{\mathbb{P}(\pi)}(\gamma) \subseteq V_{\gamma}^*(\pi \circ \nu) \otimes \Lambda^mT^*_{\pi(\nu(\gamma))}M,
\]
with $\gamma \in {\mathcal M}\pi$ such that $\mu(\gamma) = \gamma^0$, as follows:
\begin{equation}\label{3.15'}
\big\langle\widetilde{A(Z^0)}, \tilde{U} \big\rangle = - \Lambda^m((Z^0)^*)(\bar{\flat}(\tilde{U})),
\end{equation}
for $\tilde{U} \in V_\gamma(\pi \circ \nu)$.
Then, if $\tilde{\mu}: \widetilde{\mathbb{P}(\pi)} \to \mathbb{P}(\pi)$ is the canonical projection, we set
\begin{equation}\label{3.15''}
A(Z^0) = \tilde{\mu}(\widetilde{A(Z^0)}).
\end{equation}
Note that $A$ is well-defined and its local expression is
\begin{equation}\label{Local-A}
A(x^{i}, u^{\alpha}, p_{\alpha}^{i}; u^{\alpha}_j, p_{\alpha j}^{i}) = \Big(x^{i}, u^{\alpha}, p_{\alpha}^{i}; - \sum_{i} p^{i}_{\alpha i}, u^{\alpha}_j\Big).
\end{equation}
This proves that $A$ is an affine bundle epimorphism over the identity of ${\mathcal M}^0\pi$ (for more details, see \cite{Gr}).

Recall that the affine bundle $J^1(\pi \circ \nu^0) \to {\mathcal M}^0\pi$ is modelled over the vector bundle $V(J^1(\pi \circ \nu^0)) = (\pi \circ \nu^0)^*(T^*M) \otimes V(\pi \circ \nu^0) \to {\mathcal M}^0\pi$. We denote by $(x^{i}, u^{\alpha}, p_\alpha^{i}; u^\alpha_j, p^{i}_{\alpha j})$ the standard coordinates on $J^1(\pi \circ \nu^0)$ and $V(J^1(\pi \circ \nu^0))$ (see Appendix \ref{1-jet-bundle}).
From the local expression (\ref{Local-A}), it follows that the kernel of $A$ is a vector subbundle of $V(J^1(\pi \circ \nu^0))$ which is locally characterized by
\begin{equation}\label{Local-Ker-A}
\operatorname{Ker} A = \big\{ (x^{i}, u^{\alpha}, p_\alpha^{i}; u^\alpha_j, p^{i}_{\alpha j}) \in V(J^1(\pi \circ \nu^0)) \mid u^\alpha_j = 0, \sum_i p_{\alpha i}^{i} = 0, \forall \alpha, i\big\}.
\end{equation}
We can thus consider the quotient affine bundle $J^1(\pi \circ \nu^0) / \operatorname{Ker} A \to {\mathcal M}^0\pi $ which is modelled over the quotient vector bundle $(\pi \circ \nu^0)^*(T^*M) \otimes V(\pi \circ \nu^0) / \operatorname{Ker} A \to {\mathcal M}^0\pi$. From (\ref{Local-Ker-A}), we have that a local basis of sections for this vector bundle is
\[
\big\{\nu_{\alpha}^i = \big[dx^{i} \otimes \frac{\partial}{\partial u^{\alpha}}\big], \nu^\alpha = 1/m\big[ dx^{1} \otimes \frac{\partial}{\partial p_{\alpha}^1} + \dots + dx^m \otimes \frac{\partial}{\partial p_{\alpha}^m}\big] \big\}
\]
for $i \in \{1, \dots, m\}$ and $\alpha \in \{1, \dots, n\}$. Note that in the quotient vector bundle 
\[
\big[dx^{i} \otimes \frac{\partial}{\partial p_{\alpha}^i}\big] = \big[dx^{j} \otimes \frac{\partial}{\partial p_{\alpha}^j}\big], \; \; \forall \; i, j \mbox{ and } \alpha.
\]
Local coordinates associated to this basis of sections on the quotient vector bundle $(\pi \circ \nu^0)^*(T^*M) \otimes V(\pi \circ \nu^0)/ \operatorname{Ker} A\to {\mathcal M}^0\pi$ (and also on the quotient affine bundle $J^1(\pi \circ \nu^0)/ \operatorname{Ker} A$) are denoted $(x^{i}, u^{\alpha}, p_{\alpha}^{i}; \hat{u}^{\alpha}_j, \hat{p}_{\alpha})$ 

The affine bundle epimorphism $A: J^1(\pi \circ \nu^0) \to \mathbb{P}(\pi)$ induces an affine bundle isomorphism
\[
\hat{A}: J^1(\pi \circ \nu^0)/\operatorname{Ker}  A \to \mathbb{P}(\pi)
\]
and, from (\ref{Local-A}), we deduce that the local expression of $\hat{A}$ is
\begin{equation}\label{Local-A-hat}
\hat{A}(x^{i}, u^\alpha, p_\alpha^{i}; \hat{u}^\alpha_i, \hat{p}_\alpha) =  (x^{i}, u^\alpha, p_\alpha^{i}; -\hat{p}_\alpha, \hat{u}^\alpha_i).
\end{equation} 
By definition, the affine bundle isomorphism
\[
\sharp^{\rm aff}: \mathbb{P}(\pi) \to J^1(\pi \circ \nu^0) /\operatorname{Ker}  A
\]
is the inverse isomorphism to $\hat{A}: J^1(\pi \circ \nu^0) /\operatorname{Ker}  A \to \mathbb{P}(\pi)$.
If we consider the local coordinates $(x^{i}, u^{\alpha}, p_\alpha^{i}; {\mathcal A}_\alpha, {\mathcal A}^\alpha_i)$ on the phase bundle $\mathbb{P}(\pi)$ then, using (\ref{Local-A-hat}), it follows that
\begin{equation}\label{Local-sharp-Aff}
\sharp^{\rm aff}(x^{i}, u^{\alpha}, p_\alpha^{i}; {\mathcal A}_\alpha, {\mathcal A}^\alpha_i) = (x^{i}, u^{\alpha}, p_\alpha^{i}; {\mathcal A}^\alpha_i, -{\mathcal A}_\alpha).
\end{equation}

\subsection{The field-theoretic analogue to the Hamiltonian vector field $X_H:TQ \rightarrow T(T^*Q)$}\label{Gamma_h}
Let $h: {\mathcal M}^0\pi \to {\mathcal M}\pi$ be a Hamiltonian section. We have seen that the differential of $h$
\[
dh: {\mathcal M}^0\pi \to \mathbb{P}(\pi)
\]
is a section of the phase bundle $\mathbb{P}(\pi)$. So, we can define the section
\[
\Gamma_h: {\mathcal M}^0\pi \to J^1(\pi \circ \nu^0)/ \operatorname{Ker} A
\]
of the quotient affine bundle $J^1(\pi \circ \nu^0)/ \operatorname{Ker} A \to {\mathcal M}^0\pi$ by
\begin{equation}\label{Gamma-h-0}
\Gamma_h = \sharp^{\rm aff} \circ dh.
\end{equation}
Using (\ref{d-v-H}), (\ref{Local-sharp-Aff}) and the local basis of sections $\{\nu_\alpha^i, \nu^\alpha\}$ introduced above, we obtain that the local expression of $\Gamma_h$ is
\begin{equation}\label{Gamma-h}
\Gamma_h - [dx^{i} \otimes \frac{\partial}{\partial x^{i}}] = \frac{\partial H}{\partial p_{\alpha}^i} \nu^{i}_{\alpha} - \frac{\partial H}{\partial u^{\alpha}} \nu^{\alpha}.
\end{equation}

Now, we will show that $\Gamma_h$ plays the same role, in Hamiltonian Classical Field Theories of first order, that the Hamiltonian vector field associated with a Hamiltonian function in Classical Mechanics. This will give an affirmative answer to {\bf Question 3} in Section \ref{sec:phase bundle}.

For this purpose, we will discuss the relation between $\Gamma_h$ and the solutions of the Hamilton-deDonder-Weyl equations for $h$. This uses the notion of a Hamiltonian connection (see \cite{DeMaMa2,DeMaMa1}).

Let $\tau: N \to B$ be an arbitrary fibration and $H$ an Ehresmann connection on $\tau: N \to B$. Denote by $^H: N \times_B TB \to H \subseteq TN$ the horizontal lift induced by $H$ (see Appendix \ref{1-jet-bundle}). It is clear that if $1 \leq r \leq b = \operatorname{dim} B$ the previous map induces a vector bundle isomorphism between $N \times_B \Lambda^r(TB)$ and $\Lambda^r H$ which we also denote by
\[
^H: N \times_B \Lambda^r(TB) \to \Lambda^rH \subseteq \Lambda^r(TN).
\]
So, if $\chi \in \Gamma (N \times_B \Lambda^r(TB))$, the image $\chi^H$ of $\chi$ by the previous map is called the horizontal lift of $\chi$.
If $(b^{i}, n^{\alpha})$ are local coordinates on $N$ which are adapted to the fibration $\tau$, then the horizontal lift reads locally
\[
\left( \frac{\partial}{\partial b^{i}}\right)^H = \frac{\partial}{\partial b^{i}} + H^{\alpha}_i\frac{\partial}{\partial n^{\alpha}}, \; \; \forall i \in \{1, \dots, b\}
\]
and if
\[
\chi = \chi_{i_1 \dots i_r}\frac{\partial}{\partial b^{i_1}} \wedge \dots \wedge \frac{\partial}{\partial b^{i_r}} \in \Gamma(\Lambda^rTB)
\]
then its horizontal lift is
\[
\chi^H = \chi_{i_1 \dots i_r}
\left(\frac{\partial}{\partial b^{i_1}}\right)^H 
\wedge \dots \wedge \left(\frac{\partial}{\partial b^{i_r}}\right)^H
 = \chi_{i_1 \dots i_r}\left(\frac{\partial}{\partial b^{i_1}} + H_{i_1}^{\alpha_1} \frac{\partial}{\partial n^{\alpha_1}}\right) \wedge \dots \wedge \left(\frac{\partial}{\partial b^{i_r}} + H_{i_r}^{\alpha_r} \frac{\partial}{\partial n^{\alpha_r}}\right).
\]

Now, suppose that $\pi: E \to M$ is the configuration bundle of a Hamiltonian Classical Field Theory of first order, with $m = \operatorname{dim} M$, and consider the fibration $\pi \circ \nu^0: {\mathcal M}^0\pi \to M$. For a Hamiltonian section $h: {\mathcal M}^0\pi \to {\mathcal M}\pi$, we denote by $\omega_h$ the $(m+1)$-form on ${\mathcal M}^0\pi$ given by (\ref{w-h}).

Then, we may prove the following result.
\begin{lemma}
Let $h: {\mathcal M}^0\pi \to {\mathcal M}\pi$ be a Hamiltonian section, let $H$ be an Ehresmann connection on the fibration $\pi \circ \nu^0: {\mathcal M}^0\pi \to M$, and let $s^0: U \subseteq M \to {\mathcal M}^0\pi$ be a section horizontal with respect to $H$. Then, $s^0$ is a solution of the Hamilton-deDonder-Weyl equations  for $h$ if and only if
\begin{equation}\label{H-conn-1}
(s^0)^*\left(i_{\chi^H}\omega_h\right)= 0, \; \; \forall \;\chi \in \Gamma(\Lambda^m(TM)).
\end{equation}
\end{lemma}

\begin{proof}
Using that $s^0$ is horizontal with respect to $H$ we deduce that
\[
(s^0)^*(i_U \omega_h) = 0, \; \; \forall \;U \in \Gamma(V(\pi \circ \nu^0))
\]
if and only if 
\begin{equation}\label{H-conn-2}
(s^0)^*\left( i_{\chi^H}i_U \omega_h \right) = 0, \; \; \forall \;\chi \in \Gamma(\Lambda^m(TM)).
\end{equation}
So, it is clear that if (\ref{H-conn-1}) holds then $s^0$ is a solution of the Hamilton-deDonder-Weyl equations for $h$ by Proposition \ref{Prop_HdDW}.

Conversely, assume that $s^0$ is a solution of the Hamilton-deDonder-Weyl  equations for $h$. Then (\ref{H-conn-2}) holds. Thus, using that 
\[
T_{s_0(x)}{\mathcal M}^0\pi = \kern-3pt ({s_0(x)} \times T_xM)^H\oplus V_{s_0(x)}(\pi \circ \nu^0), \; \; \forall x \in U\subseteq M,
\]
we deduce that
\[
(s^0)^*\left(i_{\chi^H}\omega_h\right)(x) = 0, \; \; \forall \chi \in  \Gamma(\Lambda^mTM), \; \; \forall x \in U \subseteq M.
\]
This proves the result.
\end{proof} 
The previous result suggests the introduction of the following definition.
\begin{definition}\label{Hamiltonian-connection} Let $h: {\mathcal M}^0\pi \to {\mathcal M}\pi$ be a Hamiltonian section. An Ehresmann connection $^H: {\mathcal M}^0 \pi \times_M TM \to H \subseteq T({\mathcal M}^0\pi)$ on the fibration $\pi \circ \nu^0: {\mathcal M}^0\pi \to M$ is said to be a \textbf{Hamiltonian connection for $h$} if
\[
i_{\chi^H}\omega_h  = 0, \; \; \forall \chi \in \Gamma(\Lambda^mTM).
\]
\end{definition}
As a direct consequence of the definition, if locally the Hamiltonian section $h$ is
\[
h(x^{i}, u^{\alpha}, p_{\alpha}^{i}) = (x^{i}, u^{\alpha}, -H(x, u, p), p_{\alpha}^{i})
\]
and the Ehresmann connection is
\[
\left(\frac{\partial}{\partial x^{i}}\right)^H =  \frac{\partial}{\partial x^{i}} + H_{i}^{\alpha} \frac{\partial}{\partial u^{\alpha}} + H_{\alpha i}^j \frac{\partial}{\partial p_{\alpha}^{j}},
\]
then from (\ref{Local-w-h}), we have the equivalence
\begin{equation}\label{Ham-Connection}
^H: {\mathcal M}^0\pi \times_M TM \to T({\mathcal M}^0\pi) \mbox{ Hamiltonian connection for $h$ } \Leftrightarrow H^{\alpha}_i = \frac{\partial H}{\partial p_{\alpha}^{i}}, \; \; H^{i}_{\alpha i} = -\frac{\partial H}{\partial u^{\alpha}}.
\end{equation}
So, our definition of a Hamiltonian connection is equivalent to that introduced in \cite{DeMaMa2,DeMaMa1}. Note that a Hamiltonian connection $H$ for $h$ may be identified with a section $s^H: {\mathcal M}^0\pi \to J^1(\pi \circ \nu^0)$ of the affine bundle $J^{1}(\pi \circ \nu^0) \to {\mathcal M}^0\pi$ (see Appendix \ref{1-jet-bundle}). Moreover, if $\Gamma_h$ is the section of the quotient affine bundle $J^1(\pi \circ \nu^0) / \operatorname{Ker} A \to {\mathcal M}^0\pi$ defined in (\ref{Gamma-h-0}) then, using (\ref{Gamma-h}) and (\ref{Ham-Connection}), we obtain the following result.

\begin{proposition}\label{Gamma-h-Ham-connections}
Let $s^H: {\mathcal M}^0\pi \to J^1(\pi \circ \nu^0)$ be the section of the affine bundle $J^1(\pi \circ \nu^0) \to {\mathcal M}^0\pi$ induced by a Hamiltonian connection $H$ for $h$ and $p: J^1(\pi \circ \nu^0) \to J^1(\pi \circ \nu^0) / \operatorname{Ker} A $ be the canonical projection. Then
\[
\Gamma_h = p \circ s^H.
\]
\end{proposition}

From (\ref{Local-Ker-A}) and (\ref{Ham-Connection}), we also get the following result.

\begin{proposition}\label{different-Ham-conn}
If $H, \bar{H}: {\mathcal M}^0\pi \to J^{1}(\pi \circ \nu^0) $ are Hamiltonian connections for the same Hamiltonian section $h: {\mathcal M}^0\pi \to {\mathcal M}\pi$, then they satisfy
\[
\bar{H} - H \in \Gamma(\operatorname{Ker} A).
\]
\end{proposition}
Finally, from Propositions \ref{Gamma-h-Ham-connections} and \ref{different-Ham-conn}, it follows the following characterization of Hamiltonian connections for $h$.
\begin{theorem}\label{Gamma-h-Ham-Connections-h}
Let $h: {\mathcal M}^0\pi \to {\mathcal M}\pi$ be a Hamiltonian section and let $H$ be an Ehresmann connection for the fibration $\pi \circ \nu^0: {\mathcal M}^0\pi \to M$. Then, $H$ is a Hamiltonian connection for $h$ if and only if
\[
\Gamma_h = p \circ s^H.
\]
\end{theorem}


Theorem \ref{Gamma-h-Ham-Connections-h} suggests the introduction of the following definition.
\begin{definition}\label{Equiv-class-Ham-Connections-h}
Let $h: {\mathcal M}^0\pi \to {\mathcal M}\pi$ be a Hamiltonian section. Then, the section $\Gamma_h: {\mathcal M}^0\pi \to J^{1}(\pi \circ \nu^0)/ \operatorname{Ker}A$ is called the \textbf{equivalence class of the Hamiltonian connections for $h$}.
\end{definition}
The following commutative diagram illustrates the results obtained in Sections \ref{sec:phase bundle}, \ref{canonical-isomorphism}, \ref{Gamma_h}
\[
\xymatrix{{\mathbb{P}(\pi)}\ar[rr]^{\sharp^{\rm aff}}&&J^{1}(\pi \circ \nu^0)/
\operatorname{Ker} A\\&
{\mathcal M}^0\pi\ar[ul]^{dh}\ar[ur]_{\Gamma_h}&}
\]
It is the field-theoretic analogue to Diagram \ref{diagram-1} for Hamiltonian Mechanics.

The last step is to introduce a suitable space of currents for Hamiltonian Classical Field Theories of first order and a suitable canonical bracket formulation for the evolution of such currents along the solution of the Hamilton-deDonder-Weyl equations. This is the aim of the next two subsections.


\subsection{A suitable space of currents for Hamiltonian Classical Field Theories}\label{Currents-Ham-deDonder-Weyl} We shall define a space of currents for Hamiltonian Classical Field Theories of first order, which plays the same role that the space of observables in Hamiltonian Mechanics.

Recall that in Hamiltonian Mechanics, the Hamiltonian vector field $X_H$ is a section of the vector bundle $T(T^*Q) \to T^*Q$ and the space of observables is the set $C^{\infty}(T^*Q)$ of real $C^{\infty}$-functions on $T^*Q$. Given an observable $F \in  C^{\infty}(T^*Q)$, we can consider a section $dF$ (the differential of $F$) of the dual bundle $T^*(T^*Q) \to T^*Q$ to $T(T^*Q) \to T^*Q$ and the evolution of the observable $F$ along a solution $s: I \subseteq \mathbb{R} \to T^*Q$ of Hamilton's equation is given as
\[
\frac{d}{dt} (F\circ s)= \langle dF, X_H \rangle \circ s.
\]
When written for all observables $F$, the previous equations are equivalent to the Hamilton equations.

Our goal is to carry out these construction for Hamiltonian Classical Field theories. As we have seen, given a Hamiltonian section $h: {\mathcal M}^0\pi \to {\mathcal M}\pi$, the object corresponding to the Hamiltonian vector field $X_H$  is the section $\Gamma_h$ of the quotient affine bundle $J^{1}(\pi \circ \nu^0)/\operatorname{Ker}A \to {\mathcal M}^0\pi$. So, we need to overcome the following two steps:

\noindent{\bf First step}: Describe the dual vector bundle $(J^{1}(\pi \circ \nu^0)/ \operatorname{Ker} A)^+$ to the affine bundle $J^{1}(\pi \circ \nu^0)/ \operatorname{Ker A} \to {\mathcal M}^0\pi$.

\noindent{\bf Second step}: Introduce a space ${\mathcal O}$ of currents and a differential operator
\begin{equation}\label{2nd_step} 
d: {\mathcal O} \to \Gamma\big((J^{1}(\pi \circ \nu^0)/\operatorname{Ker} A)^+\big)
\end{equation} 
on this space, such that the evolution of a current $\alpha^0 \in {\mathcal O}$ along a solution $s^0: U \subseteq M \to {\mathcal M}^0\pi$ is given by
\[
(s^0)^*(d\alpha^0) = \langle d\alpha^0, \Gamma_h \rangle  \circ s^0.
\]
We will show that $s^0$ satisfies these equations for any $\alpha\in \mathcal{O}$ if and only if $s^0$ is a solution of the Hamilton-deDonder-Weyl equations.

\noindent{\bf First step}: Let $A: J^{1}(\pi \circ \nu^0) \to \mathbb{P}(\pi)$ (respectively, $\hat{A}: J^1(\pi \circ \nu^0)/\operatorname{Ker} A \to \mathbb{P}(\pi)$) be the affine bundle epimorphism (respectively, isomorphism) considered in Section \ref{canonical-isomorphism}.
Denote by $\mathbb{P}(\pi)^+$ and $J^1(\pi \circ \nu^0)^+$ the vector bundles over ${\mathcal M}^0\pi$ defined by
\begin{align*} 
\mathbb{P}(\pi)^+& = {\rm Aff}(\mathbb{P}(\pi), (\pi \circ \nu^0)^*(\Lambda^mT^*M))\\
J^1(\pi \circ \nu^0)^+ &= {\rm Aff}(J^1(\pi \circ \nu^0),  (\pi \circ \nu^0)^*(\Lambda^mT^*M)).
\end{align*} 
It is clear that $A$ induces the vector bundle morphism
\[
A^+: \mathbb{P}(\pi)^+ \to J^1(\pi \circ \nu^0)^+ = {\mathcal M}(\pi \circ \nu^0),\qquad 
\langle A^+(U^0), Z^0 \rangle = \langle U^0, A(Z^0) \rangle
\]
for $U^0 \in \mathbb{P}(\pi)^+_{\gamma^0}$ and $Z^0 \in J^1_{\gamma^0}(\pi \circ \nu^0)$, with $\gamma^0 \in {\mathcal M}^0\pi$. 
Since $A$ is an epimorphism, we deduce that $A^+$ is a vector bundle monomorphism. In addition, the image of $A^+$ is the vector subbundle of $J^1(\pi \circ \nu^0)^+$ whose fiber at the point $\gamma^0 \in {\mathcal M}^0\pi$ is
\[
A^+(\mathbb{P}(\pi)^+)_{\gamma^0} = \big\{\theta^0 \in {\rm Aff}(J^1_{\gamma^0}(\pi \circ \nu^0),  \Lambda^mT^*_{\pi(\nu^0(\gamma^0))}M) \mid \big\langle(\theta^0)^{l}, \operatorname{Ker} A \cap V_{\gamma^0}(J^1(\pi \circ \nu^0))\big\rangle = 0\big\}.
\]
Here, $(\theta^0)^l: V_{\gamma^0}(J^1(\pi \circ \nu^0)) \to \Lambda^mT^*_{\pi(\nu^0(\gamma^0))}M$ denotes the linear map associated with the affine map $\theta^0: J^1_{\gamma^0}(\pi \circ \nu^0) \to \Lambda^mT^*_{\pi(\nu^0(\gamma^0))}M$.
So, we have a vector bundle isomorphism
\[
A^+: \mathbb{P}(\pi)^+ \to A^+(\mathbb{P}(\pi)^+) \subseteq {\mathcal M}(\pi \circ \nu^0)
\]
over the identity of ${\mathcal M}^0\pi$.

Now, denote by $\hat{A}^+: \mathbb{P}(\pi)^+ \to (J^{1}(\pi \circ \nu^0)/ \operatorname{Ker} A)^+$ the vector bundle isomorphism induced by $\hat{A}: J^{1}(\pi \circ \nu^0)/ \operatorname{Ker} A \to \mathbb{P}(\pi)$. Then, it is clear that the vector bundles $A^+(\mathbb{P}(\pi)^+)$ and $(J^{1}(\pi \circ \nu^0)/ \operatorname{Ker} A)^+$ can be identified and, under this identification, $\hat{A}^+$ is just the vector bundle isomorphism $A^+: \mathbb{P}(\pi)^+ \to A^+(\mathbb{P}(\pi)^+) \subseteq {\mathcal M}(\pi \circ \nu^0)$.

The following commutative diagram
\[
\xymatrix{{\mathbb{P}(\pi)^+}\ar[rr]^{\hat{A}^+}\ar[dr]_{A^+}&&(J^{1}(\pi \circ \nu^0)/ \operatorname{Ker} A)^+\ar[dl]^{\simeq}\\&
A^+(\mathbb{P}(\pi)^+) &}
\]
illustrates the situation.

It is desirable to have an explicit realisation of dual vector bundle $(J^{1}(\pi \circ \nu^0)/ \operatorname{Ker} A)^+$. As we know
\[
J^1(\pi \circ \nu^0)^+ = {\mathcal M}(\pi \circ \nu^0) \simeq \Lambda^m_2(T^*({\mathcal M}^0\pi))
\]
(see Section \ref{res-ext-multimomentum-bundle}). So it is possible to describe $A^+(\mathbb{P}(\pi)^+)$ and, therefore, $(J^{1}(\pi \circ \nu^0)/ \operatorname{Ker} A)^+$, as a certain vector subbundle $L$ of $\Lambda^m_2(T^*({\mathcal M}^0\pi))$. We shall now give such a description. 

First of all, using (\ref{P-Pi-tilde}), it follows that the vector bundle
\[
\widetilde{\mathbb{P}(\pi)}^+ = {\rm Aff} \big(\widetilde{\mathbb{P}(\pi)}, (\pi \circ \nu)^*(\Lambda^mT^*M)\big)
\]
is isomorphic to the vertical bundle $V(\pi \circ \nu)$ of the fibration $\pi \circ \nu: {\mathcal M}\pi \to M$.
An isomorphism
\[
\tilde{\mathcal I}: V(\pi \circ \nu) \to \widetilde{\mathbb{P}(\pi)}^+
\]
is given by
\begin{equation}\label{Def-I-tilde}
\langle\tilde{\mathcal I}(\tilde{U}), \tilde{\mathcal A}\rangle = \tilde{\mathcal A}(\tilde{U}), \; \; \mbox{ for } \tilde{U} \in V_{\gamma}(\pi \circ \nu) \mbox{ and } \tilde{\mathcal A} \in \widetilde{\mathbb{P}(\pi)}_\gamma,
\end{equation}
with $\gamma \in {\mathcal M}\pi$.

Now, if we consider the standard fibred actions of $(\pi \circ \nu)^*(\Lambda^mT^*M)$ on $V(\pi \circ \nu)$ and $\widetilde{\mathbb{P}(\pi)}^+$ then, it is clear that $\tilde{\mathcal I}$ is $(\pi \circ \nu)^*(\Lambda^mT^*M)$-equivariant and, thus, it induces a vector bundle isomorphism
\[
{\mathcal I}: V(\pi \circ \nu)/(\pi \circ \nu)^*(\Lambda^mT^*M) \to \widetilde{\mathbb{P}(\pi)}^+/(\pi \circ \nu)^*(\Lambda^mT^*M)
\]
over the identity of ${\mathcal M}^0\pi$.
Then, from Definition \ref{Def-phase-bundle}, we deduce that the quotient vector bundle $\widetilde{\mathbb{P}(\pi)}^+/(\pi \circ \nu)^*(\Lambda^mT^*M)$ is isomorphic to $\mathbb{P}(\pi)^+$. So, we have a vector bundle isomorphism
\[
{\mathcal I}: V(\pi \circ \nu)/(\pi \circ \nu)^*(\Lambda^mT^*M) \to \mathbb{P}(\pi)^+
\]
which is characterized by the following condition
\begin{equation}\label{I-mu-A-tilde}
\langle{\mathcal I}[\tilde{U}], \tilde{\mu}(\tilde{\mathcal A})\rangle = \langle\tilde{\mathcal I}(\tilde{U}), \tilde{\mathcal A} \rangle
\end{equation}
for $\tilde{U} \in V_{\gamma}(\pi \circ \nu)$ and $\tilde{\mathcal A} \in \widetilde{\mathbb{P}(\pi)}_\gamma$, with $\gamma \in {\mathcal M}\pi$ and $\tilde{\mu}: \widetilde{\mathbb{P}(\pi)} \to \mathbb{P}(\pi)$ the canonical projection.

We now consider the composition
\[
A^+ \circ {\mathcal I}: V(\pi \circ \nu)/(\pi \circ \nu)^*(\Lambda^mT^*M) \to A^+(\mathbb{P}(\pi)^+) \subseteq {\mathcal M}(\pi \circ \nu^0) = \Lambda^m_2(T^*{\mathcal M}^0\pi)
\]
of the two vector bundles isomorphisms $A^+$ and $\mathcal{I}$ defined above and show that $A^+ \circ {\mathcal I}$ can be expressed in a simple way, which allows to describe its image $L$ explicitly.

Consider the vector bundle morphism $\bar{\flat}: V(\pi \circ \nu) \to {\mathcal M}(\pi \circ \nu^0) \simeq \Lambda^m_2(T^*{\mathcal M}^0\pi)$ defined in Section \ref{canonical-isomorphism} which is characterized by Eq. (\ref{3.18'}) and has the local expression \eqref{Local-b-tilde}. Using (\ref{R-action-vertical}), we deduce that $\bar{\flat}$ induces the vector bundle morphism
\[
\flat: V(\pi \circ \nu)/(\pi \circ \nu)^*(\Lambda^mT^*M) \to {\mathcal M}(\pi \circ \nu^0) \simeq \Lambda^m_2(T^*({\mathcal M}^0\pi))
\]
over the identity of ${\mathcal M}^0\pi$ given by
\begin{equation}\label{b-b-tilde}
\flat[\tilde{U}] = \bar{\flat}(\tilde{U}), \; \; \mbox{ for } \tilde{U} \in V(\pi \circ \nu).
\end{equation}
So, if $\gamma \in {\mathcal M}\pi$ and ${\mathcal J}_\gamma: V_\gamma(\pi \circ \nu) \to (V(\pi \circ \nu)/(\pi \circ \nu)^*(\Lambda^mT^*M))_{\mu(\gamma)}$ is the canonical isomorphism between the fibers by $\gamma$ and $\mu(\gamma)$ of the vector bundles $V(\pi \circ \nu) \to {\mathcal M}\pi$ and $V(\pi \circ \nu)/(\pi \circ \nu)^*(\Lambda^mT^*M) \to {\mathcal M}^0\pi$, then the following diagram
\begin{equation}\label{Diagram-Comm-2}
\xymatrix{V_\gamma(\pi \circ \nu)\ar[rr]^{\bar{\flat}_\gamma}\ar[dr]^{{\mathcal J}_\gamma}&&{\mathcal M}_{\mu(\gamma)}(\pi \circ \nu^0) = \Lambda_2^mT_{\mu(\gamma)}^*({\mathcal M}^0\pi)\\&
(V(\pi \circ \nu)/(\pi \circ \nu)^*(\Lambda^mT^*M))_{\mu(\gamma)}
\ar[ur]_{\flat_\gamma}&}
\end{equation}
is commutative
\begin{proposition}\label{Definition-L}
We have the equality
\begin{equation}\label{b-A-+-I}
- A^+ \circ {\mathcal I}= \flat.
\end{equation}
This implies that the vector bundle $(J^1(\pi \circ \nu^0)/ \operatorname{Ker} A)^+$ is isomorphic to the vector subbundle $L$ of ${\mathcal M}(\pi \circ \nu^0) \simeq \Lambda^m_2(T^*({\mathcal M}^0\pi))$ given by
\begin{equation}\label{Description-L}
L  = \bar{\flat}(V(\pi \circ \nu)).
\end{equation}
In particular, a local basis of sections of the vector subbundle $L$ is
\begin{equation}\label{local_L} 
\big\{d^mx, du^\alpha \wedge d^{m-1}x_i, \sum_i dp_\alpha^{i} \wedge d^{m-1}x_i\big\}.
\end{equation} 

\end{proposition}
\begin{proof}
If $\gamma \in {\mathcal M}\pi$, $\tilde{U} \in V_\gamma(\pi \circ \nu)$ and $Z^0 \in J^1_{\mu(\gamma)}(\pi \circ \nu^0)$ then, from (\ref{3.15''}), we obtain
\[
\langle A^+({\mathcal I}[\tilde{U}]), Z^0 \rangle = \langle {\mathcal I}[\tilde{U}], A(Z^0)\rangle = \langle{\mathcal I}[\tilde{U}], \tilde{\mu}(\widetilde{A(Z^0)})\rangle.
\]
Therefore, using (\ref{Def-I-tilde}) and (\ref{I-mu-A-tilde}), it follows that
\[
\langle A^+({\mathcal I}[\tilde{U}]), Z^0 \rangle = \langle\tilde{\mathcal I}(\tilde{U}), \widetilde{A(Z^0)} \rangle = \langle \widetilde{A(Z^0)}, \tilde{U}\rangle.
\]
So, from (\ref{duality}), (\ref{3.15'}) and (\ref{b-b-tilde}), we conclude that
\[
\langle A^+({\mathcal I}[\tilde{U}]), Z^0 \rangle = -(\Lambda^m(Z^0)^*)(\bar{\flat}(\tilde{U})) = -(\Lambda^m(Z^0)^*)(\flat[\tilde{U}]) = -\langle\flat[\tilde{U}], Z^0\rangle.
\]
This proves (\ref{b-A-+-I}).

From this, (\ref{Description-L}) follows using (\ref{b-b-tilde}) and (\ref{b-A-+-I}). 

Finally, the local expression is obtained by using (\ref{Local-b-tilde}) and (\ref{Description-L}).
\end{proof}

\noindent{\bf Second step}: The previous result together with \eqref{2nd_step} suggests the introduction of the following definition.

\begin{definition}\label{Def-currents}
The \textbf{space of currents} of a Hamiltonian Field Theory with configuration bundle $\pi: E \to M$ is
\[
{\mathcal O} = \big\{\alpha^0 \in \Gamma(\Lambda^{m-1}_1(T^*({\mathcal M}^0\pi))) \mid d\alpha^0 \in \Gamma(L)\big\}.
\]
\end{definition}

\begin{example}\label{currents}{\rm 
\textbf{i)} Let $\alpha$ be a $(m-1)$-form on $E$ which is semi-basic with respect to the projection $\pi: E \to M$. Then,
\[
\alpha^0 = (\nu^0)^*(\alpha) \in {\mathcal O}.
\]
Indeed, for $\alpha = \alpha^{i}(x, u) d^{m-1}x_i$,
we have
\[
d\alpha = \left(\frac{\partial \alpha^{i}}{\partial x^{i}}\right) d^mx + \left(\frac{\partial \alpha^{i}}{\partial u^{\alpha}}\right)du^{\alpha} \wedge d^{m-1}x_{i}
\]
which implies that $d\alpha \in \Gamma(L)$, see \eqref{local_L}.

\textbf{ii)} Let $Y$ be a section of the vector bundle $V\pi \to E$, that is, $Y$ is a vector field on $E$ and
\[
(T_y\pi)(Y(y)) = 0, \; \; \forall y \in E.
\]
Define the $(m-1)$-form $\hat{Y}$ on ${\mathcal M}^0\pi = \operatorname{Lin}(V\pi, \pi^*(\Lambda^{m-1}T^*M))$ as follows
\[
\hat{Y}(\gamma^0)(Z_1, \dots, Z_{m-1}) = \gamma^0(Y(\nu^0(\gamma^0)))\left(T_{\gamma^0}(\pi \circ \nu^0)(Z_1), \dots, T_{\gamma^0}(\pi \circ \nu^0)(Z_{m-1})\right),
\]
for $\gamma^0 \in {\mathcal M}^0\pi$ and $Z_1, \dots, Z_{m-1} \in T_{\gamma^0}{\mathcal M}^0\pi$. If the local expression of $Y$ is
\[
Y = Y^{\alpha}(x, u)\frac{\partial}{\partial u^{\alpha}}
\]
it follows that $\hat{Y} = \left(Y^{\alpha}(x, u)p_{\alpha}^{i}\right) d^{m-1}x_i$. Thus,
\[
d\hat{Y} = \left(\frac{\partial Y^{\alpha}}{\partial x^{i}}p_{\alpha}^{i}\right) d^mx + \left(\frac{\partial Y^{\alpha}}{\partial u^{\beta}}p_{\alpha}^{i}\right) du^{\beta} \wedge d^{m-1}x_i + Y^{\alpha} dp_{\alpha}^{i} \wedge d^{m-1}x_{i},
\]
which implies that $d\hat{Y} \in \Gamma(L)$ and $\hat{Y} \in {\mathcal O}$.
}
\end{example}

In the following theorem, we give the explicit description of the currents for the case when $m \geq 2$. Note that if $m = 1$ then ${\mathcal O} = C^{\infty}({\mathcal M}^0\pi)$.

\begin{theorem}\label{descrip-current}
If $m \geq 2$ then a section $\alpha^0$ of the vector bundle $\Lambda^{m-1}_1(T^*({\mathcal M}^0\pi)) \to {\mathcal M}^0\pi$ is a current if and only if there exists a unique $\pi$-semibasic $(m-1)$-form $\alpha$ on $E$ and a unique $\pi$-vertical vector field $Y$ on $E$ such that
\[
\alpha^0 = \hat{Y} + (\nu^0)^*(\alpha).
\]
\end{theorem}
\begin{proof}
It is clear that if $Y$ is a $\pi$-vertical vector field and $\alpha$ is a $\pi$-semibasic $(m-1)$-form on $E$ then $\alpha^0 = \hat{Y} + (\nu^0)^*(\alpha)$ is a current (see Examples \ref{currents}).

Conversely, suppose that $\alpha^0$ is a current. The local expression of $\alpha^0$ is $\alpha^0 = \alpha^{0i}(x^j, u^\beta, p_\beta^j) d^{m-1}x_i$
and
\[
d\alpha^0 = \left(\frac{\partial \alpha^{0i}}{\partial x^{i}}\right) d^mx + \left(\frac{\partial \alpha^{0i}}{\partial u^{\beta}}\right) du^{\beta} \wedge d^{m-1}x_i + \left(\frac{\partial \alpha^{0i}}{\partial p_\beta^j}\right)dp_{\beta}^j \wedge d^{m-1}x_i.
\]
Thus, using \eqref{local_L}, we deduce that
\begin{equation}\label{First-deriv}
\displaystyle \frac{\partial \alpha^{0i}}{\partial p_\beta^j} = 0 \mbox{ and } \frac{\partial \alpha^{0i}}{\partial p_\beta^i} = \frac{\partial \alpha^{0j}}{\partial p_\beta^j} \mbox{ if } i \neq j.
\end{equation}
This implies that
\[
\displaystyle \frac{\partial^2 \alpha^{0i}}{\partial (p_\alpha^{i})^2} = \frac{\partial^2 \alpha^{0j}}{\partial p_\alpha^{i} \partial p_\alpha^j}=0, \mbox{ with } i \neq j
\]
so if $m \geq 2$ we conclude that
\begin{equation}\label{Second-deriv}
\displaystyle \frac{\partial^2 \alpha^{0i}}{\partial (p_\alpha^{i})^2} = 0, \mbox{ for all } i.
\end{equation}
Therefore, from (\ref{First-deriv}) and (\ref{Second-deriv}), it follows that
\[
\alpha^{0i}(x^j, u^\beta, p_\beta^j) = Y^\alpha(x, u)p_\alpha^{i} + \alpha^{i}(x, u), \mbox{ for all } i.
\]
Consequently, we have proved that there exists a local $\pi$-vertical vector field $Y = Y^\alpha(x, u) \displaystyle \frac{\partial}{\partial u^\alpha}$ and a local $\pi$-semibasic $(m-1)$-form $\alpha = \alpha^{i}(x, u)d^{m-1}x_i$ on $E$ such that
\[
\alpha^0 = \hat{Y} + (\nu^0)^*(\alpha).
\]
Note that $Y$ and $\alpha$ are unique. Then, this last fact also proves the global result.
\end{proof}

\begin{remark}\label{Currents-Marco-Jerry} 
{\rm (i)} Note that ${\mathcal O}$ is a $C^\infty(E)$-module.\\
{\rm (ii)} 
In \cite{CaMa}, the authors consider as a space of currents the set of horizontal Poisson $(m-1)$-forms on ${\mathcal M}^0\pi$. Moreover, they prove that a $(m-1)$-form $F$ of this type may be described as
\[
F = \hat{Y} + (\nu^0)^*(\alpha') + \beta',
\]
where $X$ is a vertical vector field on $E$, $\alpha'$ is a $\pi$-semibasic $(m-1)$-form on $E$ and $\beta'$ is a closed $(\pi \circ \nu^0)$-semibasic $(m-1)$-form on ${\mathcal M}^0\pi$. Now, it is easy to prove that, under the previous conditions, there exists a unique closed $(m-1)$-form $\beta$ on $M$ such that $\beta' = (\pi \circ \nu^0)^*(\beta) = (\nu^0)^*(\pi^*(\beta))$. So, if we take $\alpha = \alpha' + \pi ^* \beta$, we conclude that
\[
F = \hat{Y} + (\nu^0)^*(\alpha).
\]
The previous discussion shows that ${\mathcal O}$ is just the space of currents which was considered in \cite{CaMa}. 
\end{remark}

\subsection{A suitable linear-affine bracket and the Hamilton-deDonder-Weyl equations}\label{suitable_LA_bracket}
We consider the linear-affine bracket
\[
\{\cdot, \cdot\}: {\mathcal O} \times \Gamma(\mu) \to \Gamma((\pi \circ \nu^0)^*(\Lambda^mT^*M))
\]
defined by
\begin{equation}\label{def-bracket}
\{\alpha^0, h\} = \langle d\alpha^0, \Gamma_h\rangle  = \langle d\alpha^0, \sharp^{\rm aff}(dh) \rangle.
\end{equation}

Assume that $m \geq 2$, that the local expression of the Hamiltonian section $h\in \Gamma(\mu)$ is
\[
h(x^{i}, u^{\alpha}, p_\alpha^{i}) = (x^{i}, u^{\alpha}, -H(x, u, p), p_\alpha^{i})
\]
and that the local expression of the current $\alpha^0\in {\mathcal O}$ is
\[
\alpha^0(x^{i}, u^{\alpha}, p_{\alpha}^{i}) = (Y^{\alpha}(x, u)p_\alpha^{i} + \beta^{i}(x, u))d^{m-1}x_i
\]
with $Y^\alpha$ and $\beta^{i}$ local real $C^{\infty}$-functions on $E$. Then, using (\ref{Gamma-h}) and (\ref{def-bracket}), we obtain the local expression of the linear-affine bracket \eqref{def-bracket} as 
\begin{equation}\label{Local-bracket}
\begin{array}{l}
\{\left(Y^\alpha(x, u)p_\alpha^{i} + \beta^{i}(x, u)\right)d^{m-1}x_i, h\} \\[5pt]
\;= \left(\displaystyle \frac{\partial \beta^{i}}{\partial x^{i}} + \frac{\partial Y^{\alpha}}{\partial x^{i}}p_\alpha^{i}  + \left( \frac{\partial \beta^{i}}{\partial u^\alpha} + \frac{\partial Y^{\beta}}{\partial u^\alpha}p_\beta^{i} \right)\frac{\partial H}{\partial p_\alpha^{i}} - \frac{\partial H}{\partial u^\alpha} Y^\alpha \right) \otimes d^mx.
\end{array}
\end{equation}

Note that if we write the current as $\alpha^0(x^{i}, u^{\alpha}, p_{\alpha}^{i}) = \alpha^{i} d^{m-1}x_i$, with $\alpha^i(x,u,p)=Y^{\alpha}(x, u)p_\alpha^{i} + \beta^{i}(x, u)$, the bracket takes the elegant form
\[
\{\alpha^id^{m-1}x_i, h\}= \left( \frac{\partial \alpha^{i}}{\partial x^{i}}   +  \frac{\partial \alpha^{i}}{\partial u^\alpha} \frac{\partial H}{\partial p_\alpha^{i}} - \frac{1}{m}\frac{\partial H}{\partial u^\alpha}  \frac{\partial \alpha^{i}}{\partial p^i_\alpha} \right) \otimes d^mx.
\]

As we know, $\Gamma(\mu)$ is an affine space which is modelled over the vector space $\Gamma((\pi \circ \nu^0)^*(\Lambda^mT^*M))$. Therefore, using (\ref{def-bracket}), it follows that the bilinear bracket
\[
\{\cdot, \cdot\}_l: {\mathcal O} \times \Gamma((\pi \circ \nu^0)^*(\Lambda^mT^*M)) \to \Gamma((\pi \circ \nu^0)^*(\Lambda^mT^*M))
\]
associated with the linear-affine bracket $\{\cdot, \cdot\}$ is given by
\[
\{\alpha^0, {\mathcal F}^0\}_l = \langle\mu^0(d\alpha^0), \sharp^{\rm lin}(d^l{\mathcal F}^0)\rangle.
\]
Here, $d^l{\mathcal F}^0$ is the vertical differential of ${\mathcal F}^0$ (see (\ref{d-l-F-0})), 
\[
\sharp^{\rm lin}: V(\mathbb{P}(\pi)) \to \displaystyle \frac{(\pi \circ \nu^0)^*(\Lambda^m T^*M) \otimes V(\pi \circ \nu^0)}{\operatorname{Ker} A}
\]
 is the vector bundle isomorphism associated with the affine bundle isomorphism $\sharp^{\rm aff}: \mathbb{P}(\pi) \to J^1(\pi \circ \nu^0)/\operatorname{Ker} A$, and
 \[
 \mu^0: {\rm Aff}\Big(\frac{J^1(\pi \circ \nu^0)}{\operatorname{Ker}  A}, (\pi \circ \nu^0)^*(\Lambda^mT^*M)\Big) \to {\rm Lin}\Big( \displaystyle \frac{(\pi \circ \nu^0)^*(T^*M) \otimes V(\pi \circ \nu^0)}{\operatorname{Ker}  A}, (\pi \circ \nu^0)^*(\Lambda^mT^*M)\Big)
\] 
is the canonical projection.
In local coordinates, we have
\begin{equation}\label{Local-exp-bracket-l}
\begin{array}{l}
\{ \left(Y^\alpha(x, u)p_\alpha^{i} + \beta^{i}(x, u)\right)d^{m-1}x_i, {\mathcal F}^0(x^{i}, u^{\alpha}, p_\alpha^{i})\otimes d^mx \}_l\\[5pt]
\;= \left(\left(\displaystyle  \frac{\partial \beta^{i}}{\partial u^\alpha} + \frac{\partial Y^{\beta}}{\partial u^\alpha}p_\beta^{i}\right) \displaystyle \frac{\partial {\mathcal F}^0}{\partial p_\alpha^{i}} - \frac{\partial {\mathcal F}^0}{\partial u^\alpha}Y^\alpha\right) \otimes d^mx.
\end{array}
\end{equation}
Again, if we write the current as $\alpha^0(x^{i}, u^{\alpha}, p_{\alpha}^{i}) = \alpha^{i} d^{m-1}x_i$, with $\alpha^i(x,u,p)=(Y^{\alpha}(x, u)p_\alpha^{i} + \beta^{i}(x, u))$, the bilinear bracket takes the form
\[
\{\alpha^id^{m-1}x_i, {\mathcal F}^0(x^{i}, u^{\alpha}, p_\alpha^{i})\otimes d^mx\}_l= \left( \frac{\partial \alpha^{i}}{\partial u^\alpha} \frac{\partial {\mathcal F}^0}{\partial p_\alpha^{i}} - \frac{1}{m} \frac{\partial {\mathcal F}^0}{\partial u^\alpha}  \frac{\partial \alpha^{i}}{\partial p^i_\alpha} \right) \otimes d^mx.
\]
On the other hand, if $m = 1$ then the space of currents is
\[
{\mathcal O} = \Gamma((\pi \circ \nu^0)^*(T^*M)) = C^{\infty}({\mathcal M}^0\pi)
\]
and, using (\ref{Gamma-h}) and (\ref{def-bracket}), we deduce that the linear-affine bracket
\[
\{\cdot, \cdot\}: C^{\infty}({\mathcal M}^0\pi) \times \Gamma ( \mu ) \to C^{\infty}({\mathcal M}^0\pi)
\]
and the bilinear bracket
\[
\{\cdot, \cdot\}_l: C^{\infty}({\mathcal M}^0\pi) \times C^{\infty}({\mathcal M}^0\pi) \to C^{\infty}({\mathcal M}^0\pi)
\]
are locally given by
\begin{equation}\label{Aff-bracket-m-1}
\{f^0, h\} = \displaystyle \frac{\partial f^0}{\partial x} + \frac{\partial f^0}{\partial u^\alpha}\frac{\partial H}{\partial p_\alpha} - \frac{\partial f^0}{\partial p_\alpha}\frac{\partial H}{\partial u^\alpha}
\end{equation}
and
\begin{equation}\label{Linear-bracket-m-1}
\{f^0, g^0\}_l = \displaystyle \frac{\partial f^0}{\partial u^\alpha}\frac{\partial g^0}{\partial p_\alpha} - \frac{\partial f^0}{\partial p_\alpha}\frac{\partial g^0}{\partial u^\alpha}
\end{equation}
for $f^0, g^0 \in C^{\infty}({\mathcal M}^0\pi)$ and $h \in \Gamma(\mu)$.

The following result extends to the field-theoretic context the canonical Poisson bracket formulation of Hamilton's equations.

\begin{theorem}\label{evolution-current}
Let $h: {\mathcal M}^0\pi \to {\mathcal M}\pi$ be a Hamiltonian section and $s^0: M \to {\mathcal M}^0\pi$ a (local) section of the projection $\pi \circ \nu^0: {\mathcal M}^0\pi \to M$. Then, $s^0$ is a solution of the Hamilton-deDonder-Weyl equations for $h$ if and only if
\begin{equation}\label{Current-Ham-eqs}
(s^0)^*(d\alpha^0) = \{\alpha^0, h\}\circ s^0, \; \; \; \forall \;\alpha^0 \in {\mathcal O}.
\end{equation}
\end{theorem}
\begin{proof}
Suppose that $m \geq 2$ and that
\[
s^0(x^{i}) = (x^{i}, u^\alpha(x), p_\alpha^{i}(x)), \quad  \alpha^0(x^{i}, u^\alpha, p_\alpha^{i}) = \left(Y^\alpha(x, u)p_\alpha^{i} + \beta^{i}(x, u)\right)d^{m-1}x_i,
\]
with $Y^\alpha$, $\beta^{i}$ local real $C^\infty$-functions on $E$.
Then,
\[
(s^0)^*(d\alpha^0) = \left(\frac{\partial \beta^{i}}{\partial x^{i}} + \frac{\partial Y^{\alpha}}{\partial x^{i}}p_\alpha^{i} + \left( \frac{\partial \beta^{i}}{\partial u^\alpha} + \frac{\partial Y^{\beta}}{\partial u^\alpha}p_\beta^{i} \right)\frac{\partial u^\alpha}{\partial x^{i}} + \frac{\partial p_\alpha^{i}}{\partial x^{i}} Y^\alpha \right) \otimes d^mx.
\]
Thus, using (\ref{Local-bracket}), we conclude that (\ref{Current-Ham-eqs}) hold if and only if
\[
\displaystyle \frac{\partial u^\alpha}{\partial x^{i}} = \frac{\partial H}{\partial p_\alpha^{i}}, \; \; \; \; \frac{\partial p_\alpha^{i}}{\partial x^{i}} = -\frac{\partial H}{\partial u^\alpha},
\]
or, equivalently, $s^0$ is a solution of the Hamilton-deDonder-Weyl equations for $h$.

If $m =1$ the result is proved in a similar way using (\ref{Aff-bracket-m-1}).
\end{proof}
\subsection{A remark on boundary conditions}\label{BC}
If the boundary $\partial M$ of the base space $M$ of the configuration bundle $\pi: E \to M$ is not empty, then the Hamilton-deDonder-Weyl equations for a Hamiltonian section can be supplemented by boundary conditions.

The boundary of the configuration space $E$ is just
\[
\partial E = \pi^{-1}(\partial M)
\]
in such a way that 
\[
\pi_{|\partial E}: \partial E \to \partial M
\]
is again a fibration.

In a similar way, the restricted multimomentum bundle ${\mathcal M}^0\pi$ is a manifold with boundary,
\[
\partial({\mathcal M}^0\pi) = (\nu^0)^{-1}(\partial E) = (\pi \circ \nu^0)^{-1}(\partial M),
\]
and we have fibrations
\[
(\nu^0)_{|\partial({\mathcal M}^0\pi)}: \partial({\mathcal M}^0\pi) \to \partial E, \; \; \; (\pi \circ \nu^0)_{|\partial({\mathcal M}^0\pi)}: \partial({\mathcal M}^0\pi) \to \partial M.
\]
A boundary condition for the Hamiltonian Classical Field theory is given by specifying a subbundle $B^0 \rightarrow \partial M$ of $(\pi \circ \nu^0)_{|\partial({\mathcal M}^0\pi)}: \partial({\mathcal M}^0\pi) \to \partial M$, such that $B^0 \rightarrow B_E^0:= (\nu^0)_{|\partial({\mathcal M}^0\pi)}( B^0 )$ is a subbundle of $(\nu^0)_{|\partial({\mathcal M}^0\pi)}: \partial({\mathcal M}^0\pi) \to \partial E$. In such a case, we will consider only sections $s^0: M \to {\mathcal M}^0\pi$ such that
\begin{equation}\label{boundary-cond-sections}
s^0(\partial M) \subseteq B^0.
\end{equation}
A standard assumption in the literature for the subbundle $B^0$ is
\[
i_{B^0}^*\omega_h = 0,
\]
where $i_{B^0}: B^0 \to {\mathcal M}^0\pi$ is the canonical inclusion and $h: {\mathcal M}^0\pi \to {\mathcal M}\pi$ is the Hamiltonian section (see, for instance, \cite{AsIbSp,LeMaSa,IbSp1,IbSp2}; see also \cite{BiSnFi} for boundary conditions in the Lagrangian formalism).

From (\ref{boundary-cond-sections}), we deduce that among all the Hamiltonian connections
\[
^{H}: {\mathcal M}^0\pi \times_M TM \to T({\mathcal M}^0\pi)
\]
we should only consider those whose restriction to $\partial({\mathcal M}^0\pi) \times_{\partial M} T(\partial M)$ takes values in the tangent bundle $TB^0$, that is, $^{H}$ should induce a monomorphism of vector bundles
\[
^{H}: \partial({\mathcal M}^0\pi) \times_{\partial M} T(\partial M) \to TB^0.
\]
This remark is sufficient for the purposes in this paper.

A more detailed discussion of boundary conditions for a Hamiltonian Classical Field theory of first order and its relation with the section $\Gamma_h$ of the quotient affine bundle $J^1(\pi \circ \nu^0) / \operatorname{Ker} A$ and with the theory of covariant Peierls brackets \cite{Pe} in the space of the solutions will be postponed to a future publication (see the next Section \ref{Conclusions-Future}).

\section{The affine representation of the Lie algebra of currents on the affine space of Hamiltonian sections}\label{aff-repre-Lie-algebra-aff-space}

In this section, we prove that the space of currents ${\mathcal O}$ of a Hamiltonian Field theory of first order admits a Lie algebra structure and we show that the linear affine bracket $\{\cdot, \cdot\}$ introduced in Section \ref{suitable_LA_bracket} (see (\ref{def-bracket})) induces an affine representation of ${\mathcal O}$ on the affine space of Hamiltonian sections.

We first review the notion of an affine representation of a Lie algebra on an affine space (for more details, see \cite{HaLe}).

Let $A$ be an affine space modelled over the vector space $V$.
The vector space of affine maps of $A$ on $V$, ${\rm Aff}(A, V)$, is a Lie algebra and the Lie bracket on ${\rm Aff}(A, V)$ is given by
\[
[\varphi, \psi] = \varphi^{l} \circ \psi - \psi^{l} \circ \varphi,
\]
for $\varphi, \psi \in {\rm Aff}(A, V)$, where $\varphi^l, \psi^l: V \to V$ are the linear maps associated with $\varphi, \psi$, respectively.

An affine representation of  a real Lie algebra $\frak{g}$ on $A$ is a Lie algebra morphism
\[
\mathcal{R}: \frak{g} \to {\rm Aff(A, V)}.
\]

We first note that if $\pi: E \to M$ is a configuration bundle with ${\rm dim}M = 1$ then it is easy to prove the following facts (see Section \ref{time-dependent-Ham-mech} for the particular case when $M$ is the real line $\mathbb{R}$ and $E = \mathbb{R} \times Q$):
\begin{itemize}
\item
$J^1\pi$ is an affine subbundle of corank $1$ of the tangent bundle $TE \to E$ which is modelled over the vertical bundle $V\pi \to E$ to the fibration $\pi: E \to M$.
\item
The restricted multimomentum bundle is just the dual bundle $V^*\pi \to E$ to $V\pi \to E$.
\item
The extended multimomentum bundle is the cotangent bundle $T^*E \to E$ of $E$.
\item
$\Gamma(\mu)$ is an affine space which is modelled over the vector space ${\mathcal O} = C^{\infty}(V^*\pi)$ of the currents. 
\end{itemize}
So, in this case, we have a Lie algebra structure on ${\mathcal O} = C^{\infty}(V^*\pi)$. In fact, the Lie bracket on ${\mathcal O}$ is just the Poisson bracket $\{\cdot, \cdot\}_l$ on $V^*\pi$ given by (\ref{Linear-bracket-m-1}). Moreover, using (\ref{Aff-bracket-m-1}) and (\ref{Linear-bracket-m-1}), we deduce that the linear-affine bracket
\[
\{\cdot, \cdot\}: C^{\infty}(V^*\pi) \times \Gamma(\mu) \to C^{\infty}(V^*\pi)
\]
induces an affine representation of the Lie algebra $(C^{\infty}(V^*\pi), \{\cdot, \cdot\}_l)$ on the affine space $\Gamma(\mu)$. More explicitly, we have
\[
\{\{\alpha_0, \beta_0\}_l, h\}=\{\alpha_0, \{\beta_0, h\}\}_l- \{\beta_0, \{\alpha_0, h\}\}_l,
\]
for $\alpha_0, \beta_0 \in C^{\infty}(V^*\pi)$ and $h\in \Gamma(\mu)$.


Therefore, in the rest of this section, we will assume the following hypothesis:

\noindent{\bf Assumption}: In what follows, we will suppose that ${\rm dim}M \geq 2$.

First we introduce a Lie algebra structure on ${\mathcal O}$, then we show that the linear affine bracket $\{\cdot, \cdot\}: {\mathcal O} \times \Gamma(\mu) \to \Gamma((\pi \circ \nu^0)^*(\Lambda^mT^*M))$ induces an affine representation of the Lie algebra $({\mathcal O}, \{\cdot, \cdot\}_{\mathcal O})$ on the affine space $\Gamma(\mu)$.


\subsection{The Lie algebra structure on the space of currents ${\mathcal O}$}

The construction of the Lie bracket is made in several steps which involve the definition of a vertical vector field on ${\mathcal M}^0\pi$ associated to a current.

\subsubsection{Definition of the vertical vector field on ${\mathcal M}^0\pi$ associated to a current} Let $\flat: V(\pi \circ \nu)/(\pi \circ \nu)^*(\Lambda^mT^*M)  \to L \subseteq {\mathcal M}(\pi \circ \nu^0) = \Lambda^m_2(T^*({\mathcal M}^0\pi))$ be the vector bundle isomorphism over the identity of ${\mathcal M}^0\pi$ given by (\ref{b-b-tilde}) and denote by $\sharp: L \to V(\pi \circ \nu)/(\pi \circ \nu)^*(\Lambda^mT^*M)$ the inverse morphism. If $\alpha^0 \in {\mathcal O}$ then, from Definition \ref{Def-currents}, we have $\alpha^0 \in \Gamma(\Lambda^{m-1}_1(T^*{\mathcal M}^0\pi))$ and $d\alpha^0 \in \Gamma(L)$. So, we can consider the section $\sharp(d\alpha^0)$ of the vector bundle $V(\pi \circ \nu)/(\pi \circ \nu)^*(\Lambda^mT^*M) \to {\mathcal M}^0\pi = {\mathcal M}\pi/\pi^*(\Lambda^mT^*M)$.
\begin{remark}\label{justification-notation}
The notation $\sharp(d\alpha^0)$ is justified by the following fact. The vector bundle $V(\pi \circ \nu)/(\pi \circ \nu)^*(\Lambda^mT^*M)$ is canonically isomorphic to $\mathbb{P}(\pi)^+$ (an isomorphism ${\mathcal I}$ between these vector bundles is characterized by condition (\ref{I-mu-A-tilde})). So, if $h: {\mathcal M}^0\pi \to {\mathcal M}\pi$ is a Hamiltonian section then $dh \in \mathbb{P}(\pi)$ and
\[
\langle \sharp(d\alpha^0), dh \rangle \in \Gamma((\pi \circ \nu^0)^*(\Lambda^mT^*M)).
\]
Moreover, as will be proved later (see the next Lemma \ref{sharp-tilde-bracket}), we can write the linear-affine bracket as
\[
\{\alpha^0, h\} = -\langle\sharp(d\alpha^0), dh \rangle.
\]
The reader can compare the previous expression with Eq. (\ref{Canonical_Poisson}) for the definition of the canonical Poisson bracket on $T^*Q$.
\end{remark}

Note that a section of the vector bundle $V(\pi \circ \nu)/(\pi \circ \nu)^*(\Lambda^mT^*M) \to {\mathcal M}^0\pi = {\mathcal M}\pi/\pi^*(\Lambda^mT^*M)$ can be identified with a section $U: {\mathcal M}\pi \to V(\pi \circ \nu)$ of $V(\pi \circ \nu) \to {\mathcal M}\pi$ which is equivariant with respect to the fibred actions of $\pi^*(\Lambda^mT^*M)$ on ${\mathcal M}\pi$ and on $V(\pi \circ \nu)$.

By applying this observation to $\sharp(d\alpha^0)$, we denote by
\[
\tilde{\mathcal H}_{\alpha^0}: {\mathcal M}\pi \to V(\pi \circ \nu)
\]
the equivariant vector field on ${\mathcal M}\pi$ associated with the section $\sharp(d\alpha^0)$. Since $\tilde{\mathcal H}_{\alpha^0}$ is equivariant, it follows that it is $\mu$-projectable to a vertical vector field 
\[
{\mathcal H}_{\alpha^0}: {\mathcal M}^0\pi \to V(\pi \circ \nu^0)
\]
on ${\mathcal M}^0\pi$.

We now present a description of $\tilde{\mathcal H}_{\alpha^0}$ in terms of $\alpha^0$.
Let $\tilde{\flat}: V(\pi \circ \nu) \to {\mathcal M}(\pi \circ \nu) = \Lambda^m_2(T^*({\mathcal M}\pi))$ be the vector bundle monomorphism given by (\ref{b-tilde-multi}). Denote by $\tilde{L}$ the image of $V(\pi \circ \nu)$ by $\tilde{\flat}$, so that
\[
\tilde{\flat}: V(\pi \circ \nu) \to \tilde{L} \subseteq {\mathcal M}(\pi \circ \nu) = \Lambda^m_2(T^*{\mathcal M}\pi)
\]
is a vector bundle isomorphism over the identity of ${\mathcal M}\pi$.
From (\ref{Local-multi-sym}), it follows that a local basis of $\Gamma(\tilde{L})$ is
\[
\{d^mx, du^\alpha \wedge d^{m-1}x_i, dp_\alpha^{i} \wedge d^{m-1}x_i\}.
\]
Note that if $\gamma \in {\mathcal M}\pi$ then
\[
(\Lambda^mT^*_\gamma\mu)(L_{\mu(\gamma)}) = \tilde{L}_{\gamma}.
\]
In fact, if $\theta$ is a section of ${\mathcal M}(\pi \circ \nu^0) = \Lambda^m_2(T^*{\mathcal M}^0\pi)$ then
\begin{equation}\label{L-mu-L-tilde}
\theta \in \Gamma(L) \Leftrightarrow \mu^*\theta \in \Gamma(\tilde{L}).
\end{equation}
Let $\bar{\flat}: V(\pi \circ \nu) \to L \subseteq {\mathcal M}(\pi \circ \nu^0) = \Lambda_2^m(T^*{\mathcal M}^0\pi)$ be the vector bundle morphism over the canonical projection $\mu: {\mathcal M}\pi \to {\mathcal M}^0\pi$ which is characterized by Eq. (\ref{3.18'}). If $\gamma \in {\mathcal M}\pi$ and ${\mathcal J}_\gamma: V_\gamma(\pi \circ \nu) \to (V(\pi \circ \nu)/(\pi \circ \nu)^*(\Lambda^mT^*M))_{\mu(\gamma)}$ is the canonical isomorphism between the fibers by $\gamma$ and $\mu(\gamma)$ of the vector bundles $V(\pi \circ \nu) \to {\mathcal M}\pi$ and $V(\pi \circ \nu)/(\pi \circ \nu)^*(\Lambda^mT^*M) \to {\mathcal M}^0\pi$ then, using (\ref{Diagram-Comm-1}) and (\ref{Diagram-Comm-2}), it follows that the following diagram
\begin{equation}\label{commutative-diagram}
\xymatrix{V_\gamma(\pi \circ \nu)\ar[rrr]^{\tilde{\flat}_\gamma}\ar[drrr]^{\bar{\flat}_\gamma} \ar[d]^{{\mathcal J}_\gamma}&&&\tilde{L}_\gamma\\
(V(\pi \circ \nu)/(\pi \circ \nu)^*(\Lambda^mT^*M))_{\mu(\gamma)} \ar[rrr]_{\hspace{1.5cm}\flat_{\mu(\gamma)}} &&& L_{\mu(\gamma)} \ar[u]_{\Lambda^mT^*_\gamma \mu}}
\end{equation}
is commutative. This implies that
\[
\tilde{\flat}_\gamma(\tilde{\mathcal H}_{\alpha^0}(\gamma)) = (\Lambda^mT^*_\gamma \mu)\big(\flat_{\mu(\gamma)}(\sharp(d\alpha^0)(\mu(\gamma)))\big) = (\Lambda^mT^*_\gamma \mu)((d\alpha^0)(\mu(\gamma))),
\]
or, in other words, $\tilde{\mathcal H}_{\alpha^0}$ satisfies the following condition
\begin{equation}\label{Def-H-tilde-alpha-0}
i_{\tilde{\mathcal H}_{\alpha^0}}\omega_{{\mathcal M}\pi} = \mu^*(d\alpha^0).
\end{equation}
Note that, since $\omega_{{\mathcal M}\pi}$ is non-degenerate, (\ref{Def-H-tilde-alpha-0}) may be considered as a definition of the equivariant vector field $\tilde{\mathcal H}_{\alpha^0}$.

So, in conclusion, for a current $\alpha^0 \in {\mathcal O}$ we have the following objects:
\begin{itemize}
\item
An equivariant vertical vector field $\tilde{\mathcal H}_{\alpha^0}: {\mathcal M}\pi \to V(\pi \circ \nu)$ on ${\mathcal M}\pi$, which is characterized by condition (\ref{Def-H-tilde-alpha-0}).
\item
The induced section $\sharp(d\alpha^0)$ of the vector bundle 
\[
\mathbb{P}(\pi)^+ \simeq V(\pi \circ \nu) / (\pi \circ \nu)^*(\Lambda^mT^*M) \to {\mathcal M}^0\pi = {\mathcal M}\pi / \pi^*(\Lambda^mT^*M).
\]
\item
The vertical vector field on ${\mathcal M}^0\pi$
\[
{\mathcal H}_{\alpha^0}: {\mathcal M}^0\pi \to V(\pi \circ \nu^0)
\]
which is the projection, via $\mu: {\mathcal M}\pi \to {\mathcal M}^0\pi$, of $\tilde{\mathcal H}_{\alpha^0}$.
\end{itemize}

We now present the local expressions of the vector fields $\tilde{\mathcal H}_{\alpha^0}$ and ${\mathcal H}_{\alpha^0}$.
As we have seen (see Section \ref{Currents-Ham-deDonder-Weyl}), the local expression of an element $\alpha^0 \in {\mathcal O}$ is
\[
\alpha^0(x^{i}, u^{\alpha}, p_{\alpha}^{i}) = \left(Y^{\alpha}(x^{j}, u^{\beta})p_\alpha^{i} + \alpha^{i}(x^{j}, u^{\beta})\right)d^{m-1}x_i
\]
with $Y^\alpha$ and $\alpha^{i}$ local real $C^{\infty}$-functions on $E$. We have
\begin{equation}\label{dif-current}
d\alpha^0 = \left(\frac{\partial \alpha^{i}}{\partial x^{i}} + \frac{\partial Y^{\alpha}}{\partial x^{i}}p_\alpha^{i} \right) d^mx + \left(\frac{\partial \alpha^{i}}{\partial u^\beta} + \frac{\partial Y^{\alpha}}{\partial u^\beta}p_\alpha^{i} \right) du^\beta \wedge d^{m-1}x_i + Y^{\alpha} dp_\alpha^{i} \wedge d^{m-1}x_i.
\end{equation}
 Thus, using (\ref{Local-multi-sym}) and (\ref{Def-H-tilde-alpha-0}), we deduce that 
 \begin{equation}\label{sharp-tilde}
 \tilde{\mathcal H}_{\alpha^0} = Y^\alpha \frac{\partial}{\partial u^\alpha} - \left(\frac{\partial \alpha^{i}}{\partial x^{i}} + \frac{\partial Y^{\alpha}}{\partial x^{i}}p_\alpha^{i} \right) \frac{\partial}{\partial p} - \left(\frac{\partial \alpha^{i}}{\partial u^\beta} + \frac{\partial Y^{\alpha}}{\partial u^\beta}p_\alpha^{i} \right) \frac{\partial}{\partial p_\beta^{i}}.
 \end{equation}
 Therefore, it follows that
  \begin{equation}\label{sharp-current-restricted}
 {\mathcal H}_{\alpha^0} = Y^\alpha \frac{\partial}{\partial u^\alpha} - \left(\frac{\partial \alpha^{i}}{\partial u^\beta} + \frac{\partial Y^{\alpha}}{\partial u^\beta}p_\alpha^{i} \right) \frac{\partial}{\partial p_\beta^{i}}.
 \end{equation}
 Following the terminology in \cite{CaCaIb} (see also \cite{GoSt}), Eq. (\ref{Def-H-tilde-alpha-0}) implies that $\mu^*\alpha^0$ is a Hamiltonian $(m-1)$-form and that the vector field $\tilde{\mathcal H}_{\alpha^0}$ is a Hamiltonian vector field on the multisymplectic manifold $({\mathcal M}\pi, \omega_{{\mathcal M}\pi})$.

\subsubsection{Definition of the Lie bracket $\{\cdot, \cdot\}_{\mathcal O}$}
As in \cite{CaCaIb}, we consider the $(m-1)$-form $\{\mu^*(\alpha^0), \mu^*(\beta^0)\}^{\tilde{}}_{\mathcal O}$ on the multisymplectic manifold $({\mathcal M}\pi, \omega_{{\mathcal M}\pi})$ defined by
\[
\{\mu^*(\alpha^0), \mu^*(\beta^0)\}^{\tilde{}}_{\mathcal O} = -i_{\tilde{\mathcal H}_{\alpha^0}} i_{\tilde{\mathcal H}_{\beta^0}} \omega_{{\mathcal M}\pi} = -i_{\tilde{\mathcal H}_{\alpha^0}}(\mu^*(d\beta^0)).
\]
Using the fact that $\tilde{\mathcal H}_{\alpha^0}$ is $\mu$-projectable, it follows that such a $(m-1)$-form is basic with respect to $\mu$. In fact,
\[
\{\mu^*(\alpha^0), \mu^*(\beta^0)\}_{\mathcal O}^{\tilde{}} = -\mu^*(i_{{\mathcal H}_{\alpha^0}}d\beta^0).
\]
This equation suggests the introduction of the the $(m-1)$-form $\{\alpha^0, \beta^0\}_{\mathcal O}$ on ${\mathcal M}^0\pi$ given by
\begin{equation}\label{def-bracket-0}
\{\alpha^0, \beta^0\}_{\mathcal O} = - i_{{\mathcal H}_{\alpha^0}}d\beta^0.
\end{equation}
Using the local expressions
\[
\alpha^0(x^{i}, u^{\alpha}, p_\alpha^{i}) = \left( Y^\alpha(x, u)p_\alpha^{i} + \alpha^{i}(x, u) \right)d^{m-1}x_i, \; \; \beta^0(x^{i}, u^{\alpha}, p_\alpha^{i}) = \left( Z^\alpha(x, u)p_\alpha^{i} + \beta^{i}(x, u) \right)d^{m-1}x_i
\]
we deduce from (\ref{dif-current}) and (\ref{sharp-current-restricted}) that
\begin{equation}
\begin{array}{l}\label{Local-bracket-current}
\{\left( Y^\alpha(x, u)p_\alpha^{i} + \alpha^{i}(x, u) \right)d^{m-1}x_i, \left( Z^\alpha(x, u)p_\alpha^{i} + \beta^{i}(x, u) \right)d^{m-1}x_i\}_{\mathcal O}  \\
[5pt]
\;= - \left(\big(\displaystyle Y^\beta \frac{\partial Z^\alpha}{\partial u^\beta} -  Z^\beta \frac{\partial Y^\alpha}{\partial u^\beta}\big)p_\alpha^{i} + \big(\displaystyle Y^\beta \frac{\partial \beta^{i}}{\partial u^\beta} -  Z^\beta \frac{\partial \alpha^{i}}{\partial u^\beta}\big)\right)d^{m-1}x_i 
\end{array}
\end{equation}
So, it is clear that $\{\alpha, \beta\}_{\mathcal O} \in {\mathcal O}$.

Moreover, we can prove the following result.
\begin{theorem}\label{Lie-algebra-O}
The bracket $\{\cdot, \cdot\}_{\mathcal O}$ given by
\[
\{\alpha^0, \beta^0\}_{\mathcal O} = -i_{{\mathcal H}_{\alpha^0}}(d\beta^0), \; \; \; \mbox{ for } \alpha^0, \beta^0 \in {\mathcal O},
\]
defines a Lie algebra structure on the space of currents ${\mathcal O}$.
\end{theorem}
\begin{proof}
Using Theorem \ref{descrip-current}, we deduce that there exists and isomorphism between the $C^{\infty}(E)$-modules ${\mathcal O}$ and $\Gamma(V\pi) \times \Gamma(\Lambda^{m-1}_1T^*E)$. In addition, from (\ref{Local-bracket-current}), it follows that under the previous isomorphism, the bracket $\{\cdot, \cdot\}_{\mathcal O}$ is given by
\[
\{(Y, \alpha), (Z, \beta)\}_{\mathcal O} = -([Y, Z], {\mathcal L}_Y\beta - {\mathcal L}_Z\alpha) =  -([Y, Z], i_Yd\beta - i_Zd\alpha),
\]
for $(Y, \alpha), (Z, \beta) \in \Gamma(V\pi) \times \Gamma(\Lambda^{m-1}_1T^*E)$, where $[\cdot, \cdot]$ is the Lie bracket of vector fields in $E$ and ${\mathcal L}$ is the Lie derivative operator. Using this definition of $\{\cdot, \cdot\}_{\mathcal O}$, it is easy to prove that $\{\cdot, \cdot\}_{\mathcal O}$ induces a Lie algebra structure on ${\mathcal O}$.
\end{proof}

Note that if we write the observables locally as $\alpha^0(x^{i}, u^{\alpha}, p_\alpha^{i}) =\alpha^{0i}d^{m-1}x_i$ and $\beta ^0(x^{i}, u^{\alpha}, p_\alpha^{i}) =\beta ^{0i}d^{m-1}x_i$ with $ \alpha^{0i}(x, u, p)=  Y^\alpha(x, u)p_\alpha^{i} + \alpha^{i}(x, u)  $ and $\beta ^{0i}(x, u, p)=  Z^\alpha(x, u)p_\alpha^{i} + \beta^{i}(x, u) $, then the Lie bracket $\{\cdot, \cdot\}_{\mathcal O}$ has the local expression
\[
\{\alpha^{0i}d^{m-1}x_i, \beta ^{0i}d^{m-1}x_i\}_{\mathcal O}=\left( \frac{\partial \alpha^{0j}}{\partial u^ \gamma } \frac{\partial \beta ^{0i}}{\partial p_ \gamma ^{j}} - \frac{\partial \beta ^{0j}}{\partial u^ \gamma }  \frac{\partial \alpha^{0i}}{\partial p^j_ \gamma } \right)d^{m-1}x_i 
\]
which is reminiscent of the local expression of the canonical Poisson bracket.
  
As a final remark on the definition of the morphism $\flat$ and the bracket $\{\cdot,\cdot\}_\mathcal{O}$ on the currents, we can derive a version of the classical result that any Poisson structure on a manifold induces a Lie algebroid structure on the cotangent bundle of the manifold. In this case, we will obtain a Lie algebroid structure on the vector bundle $L$ over $M^0\pi$. 

Indeed, it is clear that the vector bundle $V(\pi \circ \nu) \to \mathcal{M}\pi$ admits a Lie algebroid structure. The Lie bracket in the space $\Gamma(V(\pi \circ \nu))$ of sections is just the restriction of the standard Lie bracket to $(\pi \circ \nu)$-vertical vector fields and the anchor map is the inclusion $\Gamma(V(\pi \circ \nu)) \hookrightarrow \frak{X}(\mathcal{M}\pi)$. So, using the vector bundle isomorphism $\tilde{\flat}: V(\pi \circ \nu) \to \tilde{L} \subseteq \mathcal{M}(\pi \circ \nu)$, we can induce a Lie algebroid structure on the vector bundle $\tilde{L}$. In fact, a direct computation proves that the Lie bracket in the space of sections of $\tilde{L}$, $\Gamma(\tilde{L})$, is given by
\begin{equation}\label{bracket-L-tilde}
[\tilde{\theta}, \tilde{\theta}']_{\tilde{L}} = \mathcal{L}_{\tilde{\rho}(\tilde{\theta})}\tilde{\theta}' -  \mathcal{L}_{\tilde{\rho}(\tilde{\theta'})}\tilde{\theta} + d(i_{\tilde{\rho}(\tilde{\theta}')}\tilde{\theta}), \; \; \mbox{ for } \tilde{\theta}, \tilde{\theta}' \in \Gamma(\tilde{L}),
\end{equation}
where $\tilde{\rho} = \tilde{\flat}^{-1}: \Gamma(\tilde{L}) \to \Gamma(V(\pi \circ \nu))$ is the anchor map.

On the other hand, the space of sections of the vector bundle $V(\pi \circ \nu) / (\pi \circ \nu)^*(\Lambda^m(T^*M)) \to \mathcal{M}^0\pi$ may be identified with the $(\pi \circ \nu)$-vertical and $(\pi \circ \nu)^*(\Lambda^m(T^*M))$-equivariant vector fields on $\mathcal{M}\pi$. Thus, using that the standard Lie bracket of vector fields is closed for this subspace, we may induce a Lie algebroid structure on the vector bundle $V(\pi \circ \nu) / (\pi \circ \nu)^*(\Lambda^m(T^*M)) \to \mathcal{M}^0\pi$. Note that every $(\pi \circ \nu)$-vertical and $(\pi \circ \nu)^*(\Lambda^m(T^*M))$-equivariant vector field on $\mathcal{M}\pi$ is $\mu$-projectable a vector field on $\mathcal{M}^0\pi$ and this fact determines the anchor map of the Lie algebroid $V(\pi \circ \nu) / (\pi \circ \nu)^*(\Lambda^m(T^*M)) \to \mathcal{M}^0\pi$.

Now, using the vector bundle isomorphism 
\[
\flat: V(\pi \circ \nu) / (\pi \circ \nu)^*(\Lambda^m(T^*M)) \to L \subseteq \mathcal{M}(\pi \circ \nu^0)
\]
we can translate the previous Lie algebroid structure to a Lie algebroid structure on $L$. In fact, if $\theta \in \Gamma(L)$ then, from the commutativity of the diagram (\ref{commutative-diagram}), it follows that vector field $\tilde{\rho}(\mu^*(\theta))$ on $\mathcal{M}\pi$ is $\mu$-projectable to a vector field $\rho(\theta) \in \Gamma(V(\pi \circ \nu^0))$. This fact determines the anchor map $\rho: \Gamma(L) \to \Gamma(V(\pi \circ \nu^0))$ of the Lie algebroid $L$. Moreover, from (\ref{bracket-L-tilde}), we deduce that the Lie bracket $[\cdot, \cdot]$ on $\Gamma(L)$ is given by
\begin{equation}\label{bracket-L}
[\theta, \theta']_{L} = \mathcal{L}_{\rho(\theta)}\theta' -  \mathcal{L}_{\rho(\theta')}\theta + d(i_{\rho(\theta')}\theta), \; \; \mbox{ for } \theta, \theta' \in \Gamma(L).
\end{equation} 
Finally, if $\alpha^0, \beta^0 \in \mathcal{O}$ then, using (\ref{def-bracket-0}), (\ref{bracket-L}) and the fact that
\[
\rho(d\alpha^0) = \mathcal{H}_{\alpha^0}, \; \; \rho(d\beta^0) = \mathcal{H}_{\beta^0},
\]
we conclude that
\begin{equation}\label{bracket-exact-currents}
d\{\alpha^0, \beta^0\}_{\mathcal O} = -[ d\alpha^0, d\beta^0]_L, \; \; \mbox{ for } \alpha^0, \beta^0 \in {\mathcal O}.
\end{equation}
Eqs (\ref{bracket-L}) and (\ref{bracket-exact-currents}) are reminiscent of the properties of the Lie algebroid structure
on the cotangent bundle of a Poisson manifold (in particular, the cotangent bundle of an arbitrary manifold).

\subsection{The linear-affine bracket is an affine representation}
In this subsection, we show that the linear-affine bracket $\{\cdot, \cdot\}: {\mathcal O} \times \Gamma(\mu) \to \Gamma((\pi \circ \nu^0)^*(\Lambda^mT^*M))$ induces an affine representation of the Lie algebra of currents on the affine space of Hamiltonian sections.

\begin{theorem}\label{affine-representation}
The map
\begin{equation}\label{def-representation}
{\mathcal R}: {\mathcal O} \to {\rm Aff}\left(\Gamma(\mu), (\pi \circ \nu^0)^*(\Lambda^mT^*M)\right),\quad
{\mathcal R}(\alpha^0) = \{\alpha^0, \cdot \}
\end{equation}
is an affine representation of the Lie algebra $({\mathcal O}, \{\cdot, \cdot\}_{\mathcal O})$ on the affine space $\Gamma(\mu)$.
\end{theorem}
More explicitly, we have
\begin{equation}\label{affine_linear_Jacobi} 
\{\{\alpha^0,\beta^0\}_{\mathcal O}, h \}= \{ \alpha^0, \{\beta^0,h\} \}_l -  \{ \beta^0, \{\alpha^0,h\} \}_l, \; \; \mbox{ for } \alpha^0, \beta^0 \in {\mathcal O} \mbox{ and } h \in \Gamma(\mu).
\end{equation}

In order to prove this theorem, we will use the following results.
\begin{lemma}\label{anti-algebra-hom}
The linear map
\[
\tilde{\mathcal H}: {\mathcal O} \to \Gamma(V(\pi \circ \nu)), \; \; \alpha^0 \to \tilde{\mathcal H}(\alpha^0) : = \tilde{\mathcal H}_{\alpha^0}
\]
is a Lie algebra anti-morphism between the Lie algebras $({\mathcal O}, \{\cdot, \cdot\}_{\mathcal O})$ and $(\Gamma(V(\pi \circ \nu)), [\cdot, \cdot])$, where $[\cdot, \cdot ]$ is the standard Lie bracket of vector fields. 
\end{lemma}
\begin{lemma}\label{sharp-tilde-bracket}
Let $\alpha^0$ be an element of ${\mathcal O}$.
\begin{enumerate}
\item
If $h \in \Gamma(\mu)$ is a Hamiltonian section then
\begin{equation}\label{exp-bracket-l-a}
\mu^*\{\alpha^0, h\} = -(d^v{\mathcal F}_h)(\tilde{\mathcal H}_{\alpha^0}), 
\end{equation}
where ${\mathcal F}_h \in C^{\infty}\big({\mathcal M}\pi, (\pi \circ \nu)^*(\Lambda^mT^*M)\big)$ is the extended Hamiltonian density associated with $h$ and $d^v{\mathcal F}_h$ is the vertical differential of ${\mathcal F}_h$. So, we have
\begin{equation}\label{exp-bracket-l-a-1}
\{\alpha^0, h\} = -\langle\sharp(d\alpha^0), dh \rangle.
\end{equation}
\item
If ${\mathcal F}^0 \in \Gamma((\pi \circ \nu^0)^*(\Lambda^mT^*M))$ then
\begin{equation}\label{exp-bracket-l-l}
\{\alpha^0, {\mathcal F}^0\}_l = -(d^l{\mathcal F}^0)({\mathcal H}_{\alpha^0}),
\end{equation}
where $d^l{\mathcal F}^0$ is the linear part of the differential of ${\mathcal F}^0$.
\end{enumerate}
\end{lemma}
\begin{remark}
Note that $\{\alpha^0, h\} \in \Gamma((\pi \circ \nu^0)^*(\Lambda^mT^*M))$ and, thus, $\mu^*\{\alpha^0, h\} \in \Gamma((\pi \circ \nu)^*(\Lambda^mT^*M))$. In fact,
\[
\mu^*\{\alpha^0, h\}(\gamma) = \{\alpha^0, h\}(\mu(\gamma)), \; \; \; \mbox{ for } \gamma \in {\mathcal M}\pi.
\]
\end{remark}
\begin{proof} ({\em of Theorem \ref{affine-representation}})
It is clear that ${\mathcal R}$ is a linear map.
So, we must prove that ${\mathcal R}$ is a Lie algebra morphism. For this purpose, we will use Lemmas \ref{anti-algebra-hom} and \ref{sharp-tilde-bracket}.

If $\alpha^0, \beta^0$ are currents and $h \in \Gamma(\mu)$ is a Hamiltonian section then, using (\ref{def-representation}) and (\ref{exp-bracket-l-a}), we obtain
\[
\mu^*\left({\mathcal R}(\{\alpha^0, \beta^0\}_{\mathcal O}) h \right) = -(d^v{\mathcal F}_h)(\tilde{\mathcal H}_{\{\alpha^0, \beta^0\}_{\mathcal O}}).
\]
Therefore, from Lemma \ref{anti-algebra-hom}, it follows that
\[
\mu^*\left(  {\mathcal R}(\{\alpha^0, \beta^0\}_{\mathcal O})h \right) = \tilde{\mathcal H}_{\alpha^0}\left((d^v{\mathcal F}_h)(\tilde{\mathcal H}_{\beta^0})\right) - \tilde{\mathcal H}_{\beta^0}\left((d^v{\mathcal F}_h)(\tilde{\mathcal H}_{\alpha^0})\right).
\]
Now, using (\ref{exp-bracket-l-a}) and the fact the vector fields $\tilde{\mathcal H}_{\alpha^0}$ and $\tilde{\mathcal H}_{\beta^0}$ are $\mu$-projectable on the vector fields ${\mathcal H}_{\alpha^0}$ and ${\mathcal H}_{\beta^0}$, respectively, we deduce
\begin{align*}
\mu^*\left(  {\mathcal R}(\{\alpha^0, \beta^0\}_{\mathcal O})h \right) & =  \mu^*\left(-{\mathcal H}_{\alpha^0}(\{\beta^0, h\}) + {\mathcal H}_{\beta^0}(\{\alpha^0, h\})\right) \\
& = \mu^*\left(-d^l\{\beta^0, h\}({\mathcal H}_{\alpha^0}) + d^l\{\alpha^0, h\}({\mathcal H}_{\beta^0}) \right).
\end{align*}
So, from (\ref{exp-bracket-l-l}), we obtain that
\[
\mu^*\left(  {\mathcal R}(\{\alpha^0, \beta^0\}_{\mathcal O})h \right) = \mu^*\left(\{\alpha^0, \{\beta^0, h\}\}_l -  \{\beta^0, \{\alpha^0, h\}\}_l    \right).
\]
Finally, using (\ref{def-representation}) and the fact that
\[
{\mathcal R}(\alpha^0)^{l} = \{\alpha^0, \cdot \}_l, \; \; \; {\mathcal R}(\beta^0)^{l} = \{\beta^0, \cdot \}_l,
\]
we conclude that
\begin{align*} 
{\mathcal R}(\{\alpha^0, \beta^0\}_{\mathcal O})h   & =  \left({\mathcal R}(\alpha^0)^l({\mathcal R}(\beta^0)h) -  {\mathcal R}(\beta^0)^l({\mathcal R}(\alpha^0)h) \right) \\
& =  \big({\mathcal R}(\alpha^0)^l \circ {\mathcal R}(\beta^0) - {\mathcal R}(\beta^0)^l \circ {\mathcal R}(\alpha^0)\big)h,
\end{align*}
which proves the result.
\end{proof}
\begin{proof}({\em of Lemma \ref{anti-algebra-hom}})
It is clear that the map $\tilde{\mathcal H}: {\mathcal O} \to \Gamma(V(\pi \circ \nu))$ is linear. So, we will prove that
\[
[\tilde{\mathcal H}_{\alpha^0}, \tilde{\mathcal H}_{\beta^0}] = -\tilde{\mathcal H}_{\{\alpha^0, \beta^0\}_{\mathcal O}},
\]
for $\beta^0 \in {\mathcal O}$.

Now, we have that
\[
i_{[\tilde{\mathcal H}_{\alpha^0}, \tilde{\mathcal H}_{\beta^0}]}\omega_{{\mathcal M}\pi} = \left({\mathcal L}_{\tilde{\mathcal H}_{\alpha^0}} \circ i_{\tilde{\mathcal H}_{\beta^0}} - i_{\tilde{\mathcal H}_{\beta^0}} \circ {\mathcal L}_{\tilde{\mathcal H}_{\alpha^0}}\right)\omega_{{\mathcal M}\pi}.
\]
Thus, using (\ref{Def-H-tilde-alpha-0}) and the fact that $\omega_{{\mathcal M}\pi}$ is closed, we deduce that
\[
i_{[\tilde{\mathcal H}_{\alpha^0}, \tilde{\mathcal H}_{\beta^0}]}\omega_{{\mathcal M}\pi}  = {\mathcal L}_{\tilde{\mathcal H}_{\alpha^0}}(\mu^*(d\beta^0)) - i_{\tilde{\mathcal H}_{\beta^0}}d(\mu^*(d\alpha^0))  
 =  {\mathcal L}_{\tilde{\mathcal H}_{\alpha^0}}(\mu^*(d\beta^0)) = d(i_{\tilde{\mathcal H}_{\alpha^0}}(\mu^*(d\beta^0))).
\]
But, since the vector field $\tilde{\mathcal H}_{\alpha^0}$ on ${\mathcal M}\pi$ is $\mu$-projectable over the vector field ${\mathcal H}_{\alpha^0}$ on ${\mathcal M}^0\pi$, it follows that
\[
i_{[\tilde{\mathcal H}_{\alpha^0}, \tilde{\mathcal H}_{\beta^0}]}\omega_{{\mathcal M}\pi}  = d(\mu^*(i_{{\mathcal H}_{\alpha^0}}d\beta^0)).
\]
Therefore, using (\ref{def-bracket-0}), we obtain that
\[
i_{[\tilde{\mathcal H}_{\alpha^0}, \tilde{\mathcal H}_{\beta^0}]}\omega_{{\mathcal M}\pi}  = -\mu^*(d\{\alpha^0, \beta^0\}_{\mathcal O}).
\]
So, from (\ref{Def-H-tilde-alpha-0}), it follows that
\[
i_{[\tilde{\mathcal H}_{\alpha^0}, \tilde{\mathcal H}_{\beta^0}]}\omega_{{\mathcal M}\pi}  = -i_{\tilde{\mathcal H}_{\{\alpha^0, \beta^0\}_{\mathcal O}}}\omega_{{\mathcal M}\pi}
\]
and, since $\omega_{{\mathcal M}\pi}$ is non-degenerate, this implies the result.
\end{proof}
\begin{proof} ({\em of Lemma \ref{sharp-tilde-bracket}})
Suppose that
\[
\alpha^0(x^{i}, u^{\alpha}, p_\alpha^{i}) = \left(Y^\alpha(x, u)p_\alpha^{i} + \alpha^{i}(x, u)\right)d^{m-1}x_i.
\]
(i) If
\[
h(x^{i}, u^{\alpha}, p_\alpha^{i}) = (x^{i}, u^{\alpha}, -H(x^{j}, u^\beta, p_\beta^j), p_\alpha^{i})
\]
then, from (\ref{d-v-F}) and (\ref{sharp-tilde}), we have
\[
(d^v{\mathcal F}_h)(\tilde{\mathcal H}_{\alpha^0}) = \left( Y^\alpha \frac{\partial H}{\partial u^\alpha} - \left( \frac{\partial Y^\beta}{\partial u^\alpha} p_\beta^{i} + \frac{\partial \alpha^{i}}{\partial u^\alpha}\right) \frac{\partial H}{\partial p_\alpha^{i}} - \left(\frac{\partial \alpha^{i}}{\partial x^{i}} + \frac{\partial Y^\alpha}{\partial x^{i}}p_\alpha^{i}\right)\right)d^mx.
\]
This, by (\ref{Local-bracket}), implies
\[
\mu^*\{\alpha^0, h\} = -(d^v{\mathcal F}_h)(\tilde{\mathcal H}_{\alpha^0}).
\]
Finally, using (\ref{exp-bracket-l-a}), Definition \ref{differential-h}, and the fact that $\tilde{\mathcal H}_{\alpha^0}$ is $\mu$-projectable on the vector 
field ${\mathcal H}_{\alpha^0}$, we deduce that (\ref{exp-bracket-l-a-1}) also holds.

(ii) Suppose that
\[
{\mathcal F}^0(x^{i}, u^\alpha, p_\alpha^{i}) = F^0(x^{i}, u^\alpha, p_\alpha^{i})d^mx.
\]
Then, we have
\[
d^l{\mathcal F}^0 = \left(\frac{\partial F^0}{\partial u^\alpha} du^\alpha + \frac{\partial F^0}{\partial p_\alpha^{i}} dp_\alpha^{i}\right)\otimes d^mx.
\]
Thus, using (\ref{sharp-current-restricted}), it follows
\[
(d^l{\mathcal F}^0)({\mathcal H}_{\alpha^0}) =  \left(Y^\alpha \frac{\partial F^0}{\partial u^\alpha} - \left( \frac{\partial Y^\beta}{\partial u^\alpha} p_\beta^{i} + \frac{\partial \alpha^{i}}{\partial u^\alpha}\right) \frac{\partial F^0}{\partial p_\alpha^{i}}\right)d^mx.
\]
Therefore, from (\ref{Local-exp-bracket-l}), we deduce the result.
\end{proof}

\section{Examples}\label{examples}
\subsection{Time-dependent Hamiltonian Mechanics}\label{time-dependent-Ham-mech}
In this section, we will use the following terminology. Let $\tau: V \to P$ be a vector bundle. Then, we can consider the vector bundle
\[
id_{\mathbb{R}} \times \tau: \mathbb{R} \times V \to \mathbb{R} \times P.
\]
The sections of this vector bundle are just the time-dependent sections of $\tau: V \to P$. For this reason, the vector bundle $ id_{\mathbb{R}} \times \tau: \mathbb{R} \times V \to \mathbb{R} \times P$ will be called the time-dependent vector bundle associated with $\tau: V \to P$.

For time-dependent Mechanics, the base space of the configuration bundle is the real line $\mathbb{R}$, that is, we have a fibration $\pi:E \to \mathbb{R}$. This fibration is trivializable but not canonically trivializable. In fact, if one chose\textcolor{blue}{s}  a reference frame one may trivialize the fibration. This means that $E$ may be identified with a global product $\mathbb{R} \times Q$ and, under this identification, $\pi: E \simeq \mathbb{R} \times Q \to \mathbb{R}$ is the canonical projection on the first factor. For simplicity, in what follows, we will assume that this is our starting point although all the constructions in this section may be extended, in a natural way, if we don't chose a reference frame (for an affine formulation of frame-independent Mechanics, we remit to \cite{GrGrUr1,GrGrUr2,GrGrUr3,GrGrUr4,IgMaPaSo,MaMeSa}.)

So, if the configuration bundle is trivial, the $1$-jet bundle $J^1\pi = J^1(pr_1) \to E = \mathbb{R} \times Q$ may be identified with the affine subbundle of $T(\mathbb{R} \times Q) = \mathbb{R} \times \mathbb{R} \times TQ$
\[
J^1(pr_1) = \big\{ \frac{\partial}{\partial t}_{|t} + v_q \in T_t\mathbb{R} \times T_qQ \mid (t, q) \in \mathbb{R} \times Q \mbox{ and } v_q \in T_qQ\big \}.
\] 
Thus, $J^1(pr_1)$ is isomorphic to the vector bundle $\mathbb{R} \times  TQ \to \mathbb{R} \times Q$.

An Ehresmann connection $^{H}: (\mathbb{R} \times Q) \times_\mathbb{R} T\mathbb{R} \to H \subseteq T(\mathbb{R} \times Q)$ on the fibration 
$pr_1:  \mathbb{R} \times Q \to \mathbb{R}$ is completely determined by a vector field $\Gamma$ on $\mathbb{R} \times Q$ satisfying
\[
\left\langle  dt, \Gamma \right\rangle  = 1.
\]
In fact,  the horizontal subbundle associated with the connection is of rank $1$ and generated by the vector field $\Gamma$.

The extended (resp. restricted) multimomentum bundle may be identified with the cotangent bundle $T^*(\mathbb{R} \times Q) = \mathbb{R} \times \mathbb{R} \times T^*Q$ (resp. the vector bundle $\mathbb{R} \times T^*Q$) and, under this identification, the multisymplectic structure on $\mathbb{R} \times \mathbb{R} \times T^*Q$ is just the canonical symplectic structure $\omega_{(\mathbb{R} \times Q)}$ on $\mathbb{R}\times \mathbb{R} \times T^*Q$.

We have a principal $\mathbb{R}$-action on $T^*(\mathbb{R} \times Q) = \mathbb{R} \times \mathbb{R} \times T^*Q$ given by
\[
p' \cdot (t, p, \alpha_q) = (t, p+p', \alpha_q)
\]
for $(t, p, \alpha_q) \in \mathbb{R} \times T^*_qQ$ and $p' \in \mathbb{R}$. The principal bundle projection is just the canonical projection
\[
\mu: T^*(\mathbb{R} \times Q) = \mathbb{R} \times \mathbb{R} \times T^*Q \to V^*(pr_1) = \mathbb{R} \times T^*Q, \; \; (t, p, \alpha_q) \to (t, \alpha_q).
\]
Note that $\mathbb{R} \times T^*Q$ admits a Poisson structure of corank $1$ which is induced by the canonical symplectic structure on $T^*Q$. In fact, the $\mathbb{R}$-action on the extended multimomentum bundle preserves the symplectic form and the canonical projection $\mu$ is a Poisson map.

In this case, a Hamiltonian section $h: V^*(pr_1) = \mathbb{R} \times T^*Q \to T^*(\mathbb{R} \times Q) = \mathbb{R} \times \mathbb{R} \times T^*Q$ may be identified with a global time-dependent Hamiltonian function $H: \mathbb{R} \times T^*Q \to \mathbb{R}$. In addition, the couple
\[
(\omega_h = h^*\omega_{(\mathbb{R} \times Q)}, \eta = (\pi \circ \nu^0)^*(dt))
\]
is a cosymplectic structure on $V^*(pr_1) = \mathbb{R} \times T^*Q$ and, thus, we can consider the Reeb vector field $\Gamma_H$ of $(\omega_H = \omega_h, \eta)$.
\begin{remark}
We recall that a cosymplectic structure on a manifold $P$ of odd dimension $2p + 1$ is a couple $(\omega, \theta)$, where $\omega$ is a closed $2$-form, $\theta$ is a closed $1$-form and $\theta \wedge \omega^p$ is a volume form on $P$. The Reeb vector field $\Gamma$ associated with the structure $(\omega, \theta)$ is the vector field on $P$ which is completely characterized by the conditions
\[
i_{\Gamma}\omega = 0, \quad i_{\Gamma}\eta = 1.
\]
\end{remark}
Note that the $2$-form $\omega_H$ on $\mathbb{R} \times T^*Q$ is given by
\[
\omega_H = \omega_Q + dH \wedge dt,
\]
where $\omega_Q$ is the canonical symplectic structure on $T^*Q$. So, the Reeb vector field $\Gamma_H$ of the cosymplectic structure $(\omega_H, \eta)$ on $\mathbb{R} \times T^*Q$ is
\[
(t, \alpha) \to \Gamma_H(t, \alpha) = \frac{\partial}{\partial t}_{|t} + X_{H(t, \cdot)}(\alpha),
\]
with $X_{H(t, \cdot)}$ the Hamiltonian vector field on $T^*Q$ associated with the function
\[
H(t, \cdot): T^*Q \to \mathbb{R}, \; \; \beta \to H(t, \beta).
\]
Thus, if $(t, q^{i}, p_{i})$ are canonical coordinates on $\mathbb{R} \times T^*Q$ then
\[
\Gamma_H = \frac{\partial}{\partial t} + \sum_{i}\left(\frac{\partial H}{\partial p_i} \frac{\partial}{\partial q^{i}} - \frac{\partial H}{\partial q^{i}} \frac{\partial}{\partial p_{i}}\right)
\]
and the integral curves $t \to (t, q^{i}(t), p_{i}(t))$ of $\Gamma_H$ are just of the solution of the Hamilton equations for $H$, that is,
\[
\displaystyle \frac{dq^{i}}{dt} = \frac{\partial H}{\partial p_i}, \; \; \; \frac{dp_{i}}{dt} = -\frac{\partial H}{\partial q^{i}}.
\]
Therefore, $\Gamma_H$ is the evolution vector field associated with the time-dependent Hamiltonian function $H: \mathbb{R} \times T^*Q \to \mathbb{R}$ (for more details see, for instance, \cite{ChLeMa,LeMaMa}).


In addition, the vector bundle $V^*(\pi \circ \nu)$ may be identified with the time-dependent cotangent bundle to $\mathbb{R} \times T^*Q$,
\[
id_{\mathbb{R}} \times \pi_{\mathbb{R} \times T^*Q}: \mathbb{R} \times T^*(\mathbb{R} \times T^*Q) \to  \mathbb{R} \times (\mathbb{R} \times T^*Q),
\]
where $\pi_{\mathbb{R} \times T^*Q}: T^*(\mathbb{R} \times T^*Q) \to \mathbb{R} \times T^*Q$ is the canonical projection.

Under the previous identification, the principal $\mathbb{R}$-action on $V^*(\pi \circ \nu) \simeq  \mathbb{R} \times T^*(\mathbb{R} \times T^*Q)\simeq  \mathbb{R} \times (\mathbb{R} \times \mathbb{R}  \times T^*T^*Q)$ is given by
\[
p' \cdot (t,(p, p_q, \tilde{\gamma})) = (t ,(p+p', p_q,  \tilde{\gamma})),
\]
for $p' \in \mathbb{R}$ and $(t, (p, p_q,\tilde{\gamma})) \in \mathbb{R} \times ( \mathbb{R} \times \mathbb{R} \times T^*T^*Q)$.

Moreover, the extended phase bundle $\widetilde{\mathbb{P}(\pi)}$ is
\[
\widetilde{\mathbb{P}(\pi)} = \{(t, dp_{|p} + \gamma) \mid (t, p) \in \mathbb{R} \mbox{ and } \gamma\in T^*(T^*Q) \}.
\]
In other words, $\widetilde{\mathbb{P}(\pi)}$ may be identified with the time-dependent vector bundle associated with the vector bundle
$id_{\mathbb{R}} \times \pi_{T^*Q}: \mathbb{R} \times T^*T^*Q \to \mathbb{R} \times T^*Q$, that is,
\[
id_{\mathbb{R}} \times (id_{\mathbb{R}} \times \pi_{T^*Q}): \widetilde{\mathbb{P}(\pi)} \simeq \mathbb{R} \times (\mathbb{R} \times T^*T^*Q) \to  T^*( \mathbb{R} \times Q)\simeq \mathbb{R} \times \mathbb{R} \times T^*Q
\]
and the principal $\mathbb{R}$-action on $\widetilde{\mathbb{P}(\pi)}\simeq \mathbb{R} \times (\mathbb{R} \times T^*T^*Q)$ is given by
\[
p' \cdot (t, (p, \tilde{\gamma})) = (t , p+p' \tilde{\gamma})
\]
for $p', t, p \in \mathbb{R}$ and $\tilde{\gamma} \in T^*(T^*Q)$.

Thus, the phase bundle $\mathbb{P}(\pi)$ may be identified with the time-dependent cotangent bundle to $T^*Q$
\[
id_{\mathbb{R}} \times \pi_{T^*Q}: \mathbb{P}(\pi) \simeq  \mathbb{R} \times T^*T^*Q \to \mathbb{R} \times T^*Q.
\]
Then, the differential $dh$ of $h$ is just the vertical differential $d^vH$ (with respect to the projection $pr_1: \mathbb{R} \times T^*Q \to \mathbb{R}$) of the Hamiltonian function $H$, that is,
\[
d^vH: \mathbb{R} \times T^*Q \to \mathbb{R} \times T^*T^*Q, \; \; (t, \alpha) \to d^vH(t, \alpha) = (t, dH_{(t, \cdot)}(\alpha)).
\] 
On the other hand, the vector bundle $L$, which was introduced in Proposition \ref{Definition-L}, is isomorphic to the cotangent bundle to $\mathbb{R} \times T^*Q$.

Moreover, as we know, the $1$-jet bundle to the projection
\[
\pi \circ \nu^0: \mathbb{R} \times T^*Q \to \mathbb{R}
\]
may be identified with the affine subbundle of $T(\mathbb{R} \times T^*Q)$ given by
\begin{equation}\label{1-jet-bundle-phase-space}
\big\{\frac{\partial}{\partial t}_{|t} + X_\alpha \in T_t \mathbb{R} \times T_\alpha T^*Q \mid (t, \alpha) \in \mathbb{R} \times T^*Q\big\}
\end{equation}
or, equivalently, with the time-dependent tangent bundle to $T^*Q$
\[
id_{\mathbb{R}} \times \tau_{T^*Q}: \mathbb{R} \times TT^*Q \to \mathbb{R} \times T^*Q.
\]
Under all the previous identifications, the affine bundle isomorphism $A: J^1(\pi \circ \nu^0) \simeq \mathbb{R} \times TT^*Q \to \mathbb{P}(\pi) \simeq \mathbb{R} \times T^*T^*Q$ is given by
\[
A(t, X) = (t, i_X\omega_Q), \; \; \mbox{ for } (t, X) \in \mathbb{R} \times TT^*Q,
\]
with $\omega_Q$ the canonical symplectic structure of $T^*Q$. Thus, in this case, the vector subbundle $\operatorname{Ker} A$ is trivial and this implies that there exists a unique Hamiltonian connection for the hamiltonian section $h$. In fact, the horizontal subbundle of such a connection is generated by the evolution vector field $\Gamma_H$.

In addition, if $\sharp^{\rm aff}: \mathbb{P}(\pi) \to J^{1}(\pi \circ \nu^0)$ is the inverse morphism of $A: J^{1}(\pi \circ \nu^0) \to \mathbb{P}(\pi)$ it is clear that, under the identification of $J^1(\pi \circ \nu^0)$ with the affine subbundle of $T(\mathbb{R} \times T^*Q)$ given by (\ref{1-jet-bundle-phase-space}), the image of the section $dh$ of $\mathbb{P}(\pi)$ is just the evolution vector field $\Gamma_H$.

On the other hand, the space ${\mathcal O}$ of currents is the set of smooth real functions on $\mathbb{R} \times T^*Q$ (the space of observables in Classical Mechanics) 
\[
{\mathcal O} \simeq C^{\infty}(\mathbb{R} \times T^*Q)
\]
and it is clear that the space $\Gamma(\mu)$ of sections of the projection
\[
\mu: T^*(\mathbb{R} \times Q) \to \mathbb{R} \times T^*Q
\]
may be also identified with $C^{\infty}(\mathbb{R} \times T^*Q)$. Then, the linear-affine bracket
\[
\{\cdot, \cdot\}: {\mathcal O} \times \Gamma(\mu) \to C^{\infty}(\mathbb{R} \times T^*Q)
\]
is given by
\[
\{F^0, H\} = \langle dF^0,\Gamma_H\rangle = \frac{\partial F^0}{\partial t} + \sum_{i}\left(\frac{\partial F^0}{\partial q^{i}} \frac{\partial H}{\partial p_i} - \frac{\partial F^0}{\partial p_{i}} \frac{\partial H}{\partial q^i} \right),
\]
for $(F^0, H) \in {\mathcal O} \times \Gamma(\mu) \simeq C^{\infty}(\mathbb{R} \times T^*Q) \times C^{\infty}(\mathbb{R} \times T^*Q)$.

This bracket was considered in \cite{GrGrUr1,GrGrUr2,MaMeSa}.

In addition, the bracket $\{\cdot, \cdot\}_{\mathcal O}$ on the space of observables ${\mathcal O} \simeq C^{\infty}(\mathbb{R} \times T^*Q)$ is just the standard Poisson bracket $\{\cdot, \cdot\}_{\mathbb{R} \times T^*Q}$ induced by the canonical symplectic structure of $T^*Q$. In other words,
\[
\{F^0, G^0\}_{\mathbb{R} \times T^*Q} = \sum_{i}\left(\frac{\partial F^0}{\partial q^{i}} \frac{\partial G^0}{\partial p_i} - \frac{\partial F^0}{\partial p_{i}} \frac{\partial G^0}{\partial q^i} \right)
\]
for $F^0, G^0 \in C^{\infty}(\mathbb{R} \times T^*Q)$.

Finally, the affine representation ${\mathcal R}$ of the Lie algebra $(C^{\infty}(\mathbb{R} \times T^*Q), \{\cdot, \cdot\}_{\mathbb{R} \times T^*Q})$ on the affine space $\Gamma(\mu) \simeq C^{\infty}(\mathbb{R} \times T^*Q)$ is given by
\[
{\mathcal R}: C^{\infty}(\mathbb{R} \times T^*Q) \to {\rm Aff}( C^{\infty}(\mathbb{R} \times T^*Q), C^{\infty}(\mathbb{R} \times T^*Q)),  \quad  {\mathcal R}(F^0) H =  \langle dF^0,\Gamma_H\rangle.
\]
Concretely, we have
\[
\{\{F_0, G_0\}_{\mathbb{R} \times T^*Q}, H\}=\{F_0, \{G_0, H\}\}_{\mathbb{R} \times T^*Q}- \{G_0, \{F_0, H\}\}_{\mathbb{R} \times T^*Q},
\]

\subsection{A particular case: the configuration bundle is trivial and the base space is orientable}\label{particular-case}
In this section we will assume that $E = M \times Q$, $\pi: E \to M$ is the canonical projection $pr_1: M \times Q \to M$
on the first factor. In this case, the affine bundle
\[
J^1(pr_1) \to M \times Q
\]
can be identified with the vector bundle 
\[
T^*M \otimes TQ = \cup_{(x, y) \in M \times Q}{\rm Lin}(T_xM, T_yQ).
\]
Let us further assume that $M$ is orientable, with $m = dim M \geq 2$, and fix a volume form $\mbox{vol} \in \Omega^m(M)$ on $M$. We denote by $\chi_{\mbox{vol}}$ the $m$-vector on $M$ which is characterized by the condition
\[
i(\chi_{\mbox{vol}})\mbox{vol} = 1.
\] 
Using the volume form $\mbox{vol}$ on $M$, we have
\[
\Lambda^mT_x^*M \simeq \mathbb{R}, \;\; \; \forall x \in M,
\]
and the vector bundle $\Lambda^mT^*M \to M$ may be trivialized as the trivial line vector bundle $M \times \mathbb{R} \to M$. Using $\mbox{vol}$ again, the reduced multimomentun bundle ${\mathcal M}^0\pi$ is isomorphic to the vector bundle. 
\[
\Lambda^{m-1}T^*M \otimes T^*Q = \cup_{(x, y) \in M \times Q}{\rm Lin}(T_yQ, \Lambda^{m-1}(T_x^*M))
\]
We can also identify it with
\[
TM \otimes T^*Q = \cup_{(x, y) \in M \times Q}\operatorname{Lin} (T^*_{x}M, T^*_yQ ).
\]
As in the general case, we will denote by $\nu^0: \Lambda^{m-1}T^*M \otimes T^*Q \to M \times Q$ the vector bundle projection. We will see that this space admits a multisymplectic structure.
\begin{proposition}
Let $\lambda_{{\mathcal M}^0\pi}$ be the $m$-form on ${\mathcal M}^0\pi$  given by
\begin{align*}
\lambda_{{\mathcal M}^0\pi}(\gamma^0)(Z_1^0, \dots, Z_m^0) & = \sum_{i=1}^m (-1)^{i+1}\left(\gamma^0(T_{\gamma^0}(pr_2 \circ \nu^0)(Z_i^0))\right)\\
&\qquad   \Big(T_{\gamma^0}(pr_1 \circ \nu^0)(Z_1^0), \dots, \widehat{T_{\gamma^0}(pr_1 \circ \nu^0)(Z_i^0)}, \dots, T_{\gamma^0}(pr_1 \circ \nu^0)(Z_m^0)\Big),
\end{align*}
for $\gamma^0 \in {\rm Lin}(T_yQ, \Lambda^{m-1}(T_x^*M))$ and $Z_1^0, \dots, Z_m^0 \in T_{\gamma^0}(\Lambda^{m-1}T^*M \otimes T^*Q)$, where $pr_2: M \times Q \to Q$ is the projection on the second factor. Then, $\omega_{{\mathcal M}^0\pi} = -d\lambda_{{\mathcal M}^0\pi}$ is a multisymplectic structure on ${\mathcal M}^0\pi$.
\end{proposition}
\begin{proof}
A direct computation proves that the local expression of $\lambda_{{\mathcal M}^0\pi}$ is
\[
\lambda_{{\mathcal M}^0\pi} = p_\alpha^{i} du^\alpha \wedge d^{m-1}x_i.
\]
Thus, the local expression of $\omega_{{\mathcal M}^0\pi}$ is
\begin{equation}\label{Local-exp-omega-0}
\omega_{{\mathcal M}^0\pi} = du^\alpha \wedge dp_\alpha^{i} \wedge d^{m-1}x_i.
\end{equation}
Therefore,
\[
i_{\frac{\partial}{\partial x^j}}\omega_{{\mathcal M}^0\pi} = du^{\alpha} \wedge dp_\alpha^{i} \wedge d^{m-2}x_{ij}, \; i_{\frac{\partial}{\partial u^\alpha}}\omega_{{\mathcal M}^0\pi} = dp_\alpha^{i} \wedge d^{m-1}x_{i}, \; i_{\frac{\partial}{\partial p_\alpha^{i}}}\omega_{{\mathcal M}^0\pi} = -du^{\alpha} \wedge d^{m-1}x_{i},
\]
which implies that $\omega_{{\mathcal M}^0\pi}$ is a multisymplectic structure on ${\mathcal M}^0\pi$. 
\end{proof}
On the other hand, the extended multimomentum bundle ${\mathcal M}\pi$ may be identified with the Withney sum of the vector bundles 
$\Lambda^mT^*M \times Q \to M \times Q$ and $\Lambda^{m-1}T^*M \otimes T^*Q \to M \times Q$, that is,
\[
{\mathcal M}\pi \simeq (\Lambda^mT^*M \times Q) \oplus (\Lambda^{m-1}T^*M \otimes T^*Q).
\]
So, using the volume form $\mbox{vol}$, we deduce that
\[
{\mathcal M}\pi \simeq \mathbb{R} \times {\mathcal M}^0\pi.
\]
Under this identification, the canonical multisymplectic structure $\omega_{{\mathcal M}\pi}$ on ${\mathcal M}\pi$ is
\[
 \omega_{{\mathcal M}\pi} = \omega_{{\mathcal M}^0\pi} - dp \wedge \mbox{vol},
 \]
where $p$ is the canonical coordinate on $\mathbb{R}$. Here, we also denote by $\omega_{{\mathcal M}^0\pi}$ and $\mbox{vol}$ the pullbacks to ${\mathcal M}\pi$ of $\omega_{{\mathcal M}^0\pi}$ and $\mbox{vol}$, respectively. 

Moreover, a Hamiltonian section $h: {\mathcal M}^0\pi \to {\mathcal M}\pi$ is just a global Hamiltonian function $H: {\mathcal M}^0\pi \to \mathbb{R}$ on ${\mathcal M}^0\pi$ and the $(m+1)$-form $\omega_h$ on ${\mathcal M}^0\pi$ is
\begin{equation}\label{omega-H-h}
\omega_h = \omega_H = \omega_{{\mathcal M}^0\pi} + dH \wedge \mbox{vol}.
\end{equation}

On the other hand, under the identification between ${\mathcal M}\pi$ and $\mathbb{R} \times {\mathcal M}^0\pi$ and using (\ref{vertical-lift-dmx}) (see Appendix \ref{ver-lift-section-vector-bundle}), it follows that
\begin{equation}\label{vol-vertical}
(\mbox{vol})^{\bf v} = \frac{\partial}{\partial p},
\end{equation} 
where $p$ is the standard coordinate on $\mathbb{R}$.

So, the extended phase bundle $\widetilde{\mathbb{P}(pr_1)}$ is isomorphic to the affine bundle over ${\mathcal M}\pi \simeq \mathbb{R} \times {\mathcal M}^0\pi$
\[
\cup_{(p, \gamma^0) \in \mathbb{R} \times {\mathcal M}^0\pi}\Big\{\tilde{\mathcal A} \in V^*_{(p, \gamma^0)}(pr_1 \circ \nu) \mid \tilde{\mathcal A}(\frac{\partial}{\partial p}_{|p}) = 1 \Big\},
\] 
where $V(pr_1 \circ \nu)$ is the vertical bundle of the fibration $pr_1 \circ \nu: {\mathcal M}\pi \to M$.

Now, using the previous identifications, we have that the fibred action of $\pi^*(\Lambda^mT^*M)$ on ${\mathcal M}(pr_1)$ is just the standard action of $\mathbb{R}$ on $\mathbb{R} \times {\mathcal M}^0\pi$. Therefore, since
\[
V_{(p, \gamma^0)}(pr_1 \circ \nu) \simeq {\rm span}\frac{\partial}{\partial p}_{|p} \oplus V_{\gamma^0}(pr_1 \circ \nu^0),
\]
for $(p, \gamma^0) \in \mathbb{R} \times {\mathcal M}^0(pr_1)$, we deduce that the phase bundle $\mathbb{P}(pr_1)$ is isomorphic to the vector bundle $V^*(pr_1 \circ \nu^0)$. An isomorphism 
\[
V^*(pr_1 \circ \nu^0) \to \mathbb{P}(pr_1) = \frac{\widetilde{\mathbb{P}(pr_1)}}{\mathbb{R}}
\]
between these spaces is given by
\[
{\mathcal A}_{\gamma^0} \in V^*_{\gamma^0}(pr_1 \circ \nu^0) \to [dp_{|p} + A_{\gamma^0}] \in \mathbb{P}(pr_1)_{\gamma^0},
\]
for $\gamma^0 \in {\mathcal M}^0\pi$, with $p$ an arbitrary real  number.

Thus, the image of the vertical differential $d^vH$ of $H$ under the previous isomorphism is just the equivalence class induced by the vertical differential $d^v{\mathcal F}_h$ of the extended Hamiltonian density ${\mathcal F}_h = p + H$. This implies that, under the identification between $V^*(pr_1 \circ \nu^0)$ and $\mathbb{P}(pr_1)$, the differential $dh$ of $h$ (as a section of the affine bundle $\mathbb{P}(pr_1) \to {\mathcal M}^0(pr_1)$) is just $d^vH$ (as a section of the vector bundle $V^*(pr_1 \circ \nu^0) \to {\mathcal M}^0(pr_1)$). 

Next, following Section \ref{canonical-isomorphism} (see (\ref{b-tilde-multi}) and (\ref{3.18'})), we will define the vector bundle monomorphism 
\[
\flat^0: V(pr_1 \circ \nu^0) \to {\mathcal M}(pr_1 \circ \nu^0) = \Lambda^m_2(T^*({\mathcal M}^0\pi))
\]
as follows
\[
\flat^0(U) = -i_U\omega_{{\mathcal M}^0\pi}(\gamma^0),
\]
for $U \in V_{\gamma^0}(pr_1 \circ \nu^0)$ and $\gamma^0 \in {\mathcal M}^0\pi$. We have
\[
\flat^0(\frac{\partial}{\partial u^\alpha}) = -(dp_\alpha^{i} \wedge d^{m-1}x_i), \; \; \flat^0(\frac{\partial}{\partial p_\alpha^{i}}) = du^\alpha  \wedge d^{m-1}x_i.
\]

Now, we consider the restriction to $J^1(pr_1 \circ \nu^0)$ of the dual morphism of $\flat^0$, that is,
\[
(\flat^0)^*: J^{1}(pr_1 \circ \nu^0) \to V^*(pr_1 \circ \nu^0).
\]
A direct computation proves that
\[
(\flat^0)^*(Z^0) = (-1)^{m+1}\left(i_{(\Lambda^mZ^0)(\chi_{\mbox{vol}}(\gamma^0))}\omega_{{\mathcal M}^0\pi}(\gamma^0)\right)_{|V_{\gamma^0}(pr_1 \circ \nu^0)},
\]
for $Z^0: T_{pr_1(\nu^0(\gamma^0))}M \to T_{\gamma^0}({\mathcal M}^0\pi) \in J^{1}_{\gamma^0}(pr_1 \circ \nu^0)$ and $\gamma^0 \in {\mathcal M}^0\pi$.  Thus, if $(x^{i}, u^\alpha, p_\alpha^{i})$ are local coordinates on ${\mathcal M}^0\pi$ such that
\[
\mbox{vol} = d^mx
\]
then, using (\ref{Local-exp-omega-0}), we deduce that
\[
(\flat^0)^*(x^{i}, u^{\alpha}, p_\alpha^{i}; u^\alpha_j, p^{i}_{\alpha j}) = \big(x^{i}, u^{\alpha}, p_\alpha^{i}; -\sum_i p^{i}_{\alpha i}, u^\alpha_i\big).
\]
This, from (\ref{Local-A}) and under the identification between $\mathbb{P}(\pr_1)$ and $V^*(pr_1 \circ \nu^0)$, implies that $(\flat^0)^* = A$, with $A: J^{1}(pr_1 \circ \nu^0) \to \mathbb{P}(pr_1)$ the affine bundle epimorphism given by (\ref{3.18'}), (\ref{3.15'}) and (\ref{3.15''}).

So, if $\widehat{(\flat^0)^*}: J^{1}(pr_1 \circ \nu^0)/ \operatorname{Ker} A \to V^*(pr_1 \circ \nu^0)$ is the affine bundle isomorphism induced by $(\flat^0)^*: J^{1}(pr_1 \circ \nu^0) \to V^*(pr_1 \circ \nu^0)$ then, under the identification between $\mathbb{P}(pr_1)$ and $V^*(pr_1 \circ \nu^0)$, we conclude that 
\[
\widehat{(\flat^0)^*} = \hat{A}.
\]
We will denote by
\[
\sharp^0: V^*(pr_1 \circ \nu^0) \to J^{1}(pr_1 \circ \nu^0)/\operatorname{Ker} A
\]
the inverse morphism of $\widehat{(\flat^0)^*}: J^{1}(pr_1 \circ \nu^0)/\operatorname{Ker} A \to V^*(pr_1 \circ \nu^0)$.

Now, from (\ref{omega-H-h}) and Definition \ref{Hamiltonian-connection}, we deduce that a connection $^H: {\mathcal M}^0\pi \times_M TM \to H \subseteq T({\mathcal M}^0\pi)$ on the fibration $pr_1 \circ \nu^0: {\mathcal M}^0\pi \to M$ is Hamiltonian if and only if
\[
(-1)^{m+1}\big(i_{\chi_{\mbox{vol}}^H}\omega_{{\mathcal M}^0\pi}\big)_{|V(pr_1 \circ \nu^0)} = d^vH,
\]
where $d^vH$ is the vertical differential of $H$ with respect to the projection $pr_1 \circ \nu^0$.

Moreover, using that $(\flat^0)^* = A$, it follows that the vector subbundle $L$ of ${\mathcal M}(pr_1 \circ \nu^0)$ introduced in Proposition \ref{Definition-L} is
\[
L_{\gamma^0} = {\rm span}(\mbox{vol}(\gamma^0) ) \oplus \flat_0(V_{\gamma^0}(pr_1 \circ \nu^0)), \; \; \mbox{ for } \gamma^0 \in {\mathcal M}^0\pi
\]
In addition, as we know (see {\bf first step} in Section \ref{Currents-Ham-deDonder-Weyl}), we have that
\[
\left(J^{1}(pr_1 \circ \nu^0)/\operatorname{Ker} A\right)^+ = {\rm Aff}\left(J^{1}(pr_1 \circ \nu^0)/\operatorname{Ker} A, \mathbb{R}\right) \simeq L.
\]
On the other hand, under the identification between ${\mathcal M}(pr_1)$ and $\mathbb{R} \times {\mathcal M}^0(pr_1)$, the projection $\mu: {\mathcal M}(pr_1) \to {\mathcal M}^0(pr_1)$ is just the canonical projection on the second factor. Thus, the affine space $\Gamma(\mu)$ is isomorphic to the vector space $C^{\infty}({\mathcal M}^0(pr_1))$ and the linear-affine bracket
\[
\{\cdot, \cdot\}: {\mathcal O} \times \Gamma(\mu) \to \Gamma((pr_1 \circ \nu^0)^*(\Lambda^mT^*M))
\]
given by (\ref{def-bracket}) may be considered as a bracket
\[
\{\cdot, \cdot\}: {\mathcal O} \times C^{\infty}({\mathcal M}^0(pr_1)) \to C^{\infty}({\mathcal M}^0(pr_1))
\]
defined by
\[
\{\alpha^0, H\} = \left\langle d\alpha^0, \sharp^0(d^vH) \right\rangle , \; \; \mbox{ for } \alpha^0 \in {\mathcal O} \mbox{ and } H \in C^{\infty}({\mathcal M}^0(pr_1)).
\]

\subsection{Continuum Mechanics}\label{CM_section}

In this section we develop the formulation of Continuum Mechanics as a Canonical Hamiltonian Field Theory. This covers the case of fluid mechanics and nonlinear elasticity. We shall assume that the reference configuration of the continuum is described by a manifold $B$ of dimension $N$, $N=2,3$ possibly with boundary, and we suppose that the continuum evolves in a $N$ dimensional manifold $Q$ without boundary, the ambient manifold, typically $ \mathbb{R} ^N$. The elements $x\in B$ denote the labels of the material points of the continuum, whereas the elements $u\in Q$ denote the current positions of these material points. The evolution of the continuum is described by a map $\varphi: [0,T]\times B \rightarrow Q$, where $[0,T]$ is the interval of time. Hence, $u=\varphi(t,x)$ describes the position of the material point $x$ at time $t$.
We shall assume that for each $t$ fixed, the map $x\in B \mapsto\varphi(t,x)\in Q$ is a smooth embedding. Boundary conditions will be described in \S\ref{BC_CM}

\subsubsection{Lagrangian and Hamiltonian formulations in continuum mechanics.}

Continuum mechanics is usually written either as a Lagrangian field theory or as an infinite dimensional classical Lagrangian or Hamiltonian system. While the infinite dimensional description is more classical, the field-theoretic description is especially useful for the derivation of multisymplectic integrators for fluid and elasticity (\cite{DeGB2021,DeGB2022,DeGBRa2014,LeMaOrWe2003,MaPaSh1998}).

In the field theoretic Lagrangian description, the map $\varphi$ is interpreted as a section $\phi \in \Gamma(\pi)$ of the trivial fiber bundle $\pi: E=M \times Q \rightarrow M$, $M=[0,T]\times B$, by writing $\phi(t,x)= (t,x, \varphi(t,x))$. The equations of motion are given by the Euler-Lagrange equations for a given Lagrangian density $\mathcal{L}:J^1\pi\rightarrow \Lambda^mT^*M$, $m=N+1$. Since the bundle is trivial, we have $J^1\pi_{(t,x,u)} \simeq T_{(t,x)}^*M \otimes T_{u}Q$. We denote by $(t,x^i, u^\alpha, V^\alpha, F^\alpha_i)$ the local coordinates. Writing locally the Lagrangian density as $\mathcal{L}=\bar{\mathcal{L}}(t,x^i,u^\alpha, V^\alpha, F^\alpha_i)dt\wedge d^Nx$, the Euler-Lagrange equations are given by
\[
\frac{\partial}{\partial t}\frac{\partial \bar{\mathcal{L}}}{\partial V^\alpha} + \frac{\partial}{\partial x^i}\frac{\partial \bar{\mathcal{L}}}{\partial F^\alpha_i}=\frac{\partial \bar{\mathcal{L}}}{\partial u^\alpha}.  
\]

In the infinite dimensional classical Lagrangian description, the map $\varphi$ is interpreted as a curve $\varphi(t)$ in the infinite dimensional manifold $\operatorname{Emb}(B,Q)$ of smooth embeddings of $B$ into $Q$. The equations are given by the (classical) Euler-Lagrange equations for the Lagrangian function $L:T \operatorname{Emb}(B,Q)\rightarrow \mathbb{R}$ defined from $\mathcal{L}$ as
\[
L(\varphi,V)=\int_B i_{\partial_t}\mathcal{L}(x, \varphi(x), V(x), T_x\varphi)=\int_B \bar{\mathcal{L}}(x^i, V^\alpha(x), \varphi^\alpha_{,i}(x))d^Nx,
\]
where we assumed that the Lagrangian density $\mathcal{L}$ does not depend explicitly on the time $t$, and $T_X\varphi:T_XB \rightarrow T_{\varphi(X)} Q$ denotes the tangent map to $\varphi \in \operatorname{Emb}(B,M)$, i.e. locally $T_x\varphi= \varphi^\alpha_{,i} \frac{\partial}{\partial u^\alpha}\otimes dx^i$. When $L$ is hyperregular, to this classical Lagrangian description is formally associated a classical Hamiltonian description with respect to the Hamiltonian $H:T^*\operatorname{Emb}(B,Q)\rightarrow \mathbb{R}$ defined on the (regular) cotangent bundle of $\operatorname{Emb}(B,Q)$. The Hamiltonian is defined by
\[
H(\varphi,M)= \int_B i_{\partial _t} \Big(\frac{\partial \mathcal{L}}{\partial V^\alpha} (x, \varphi(x), V(x),T_x\varphi)V^\alpha - \mathcal{L}(x, \varphi(x), V(x),T_x\varphi)\Big),
\]
where $V\in T_\varphi \operatorname{Emb}(B,Q)$ is such that $i_{\partial _t} \frac{\partial \mathcal{L}}{\partial V}=M\in T^*_\varphi\operatorname{Emb}(B,Q)$. In this case, the associated equations can formally be written $\dot F=\{F,H\}_{\rm can}$ with respect to the canonical Poisson bracket on $T^*\operatorname{Emb}(B,Q)$.

The Hamiltonian formulation that we present below is different from this one, since it is associated to the field theoretic Lagrangian formulation. Roughly speaking, while the canonical Hamiltonian formulation recalled above is based on a Legendre transform with respect to the time direction only, the canonical Hamiltonian field theoretic description that we will describe below is based on a Legendre transform with respect to all the variables in the base manifold $M$.

We warn the reader that the coordinates $x^i$ that were used in the previous sections for the base manifold $M$ are here given by $(t,x^i)$ for $M=[0,T]\times B$. The coordinates $u^\alpha_i$ used earlier on the fiber of $J^1\pi$ are here given by $(V^\alpha, F^\alpha_i)$ and represent the material velocity and the deformation gradient of the continuum.

\subsubsection{Lagrangian density and Legendre transform}

The Lagrangian density of continuum mechanics is defined with the help of given tensor fields on $B$ and $Q$. In order to treat both fluid dynamics and elasticity from a unified perspective, we shall consider here a Riemannian metric $G$ on $B$, two volume forms $\varrho$ and $\varsigma$ on $B$, and a Riemannian metric $g$ on $Q$. Additional tensor fields can be introduced to describe electromagnetic effects or microstructures. The volume forms $\varrho$ and $\varsigma$ are the mass density and the entropy density in the reference configuration and are locally written as $\varrho=\bar\varrho  d^NX$ and $\varsigma=\bar\varsigma  d^NX$. The potential energy density is a bundle map
\[
\mathcal{E}: (T_x^* B\otimes T_uQ) \times \Lambda^NT^*_x B\times \Lambda^NT^*_x B \times S^2 T_x B\times S^2 T_u ^*Q \rightarrow  \Lambda^NT^*_x B,
\]
covering the projection $B\times Q\rightarrow B$. In local coordinates, it reads
\[
\bar{\mathcal{E}}(x^i,u^\alpha,  F^\alpha_i, \bar\varrho, \bar\varsigma , G^{ij},g_{\alpha\beta})d^NX.
\]
This is a general form of potential energy density for continua, including fluid and elasticity, which may describe both internal and stored energies.

The associated Lagrangian density $\mathcal{L}:J^1\pi\rightarrow \Lambda^mT^*M$ is given by the kinetic minus the potential energy, and reads
\begin{equation}\label{Lagrangian_density_CM}
\begin{aligned}
\mathcal{L}(t,x^i, u^\alpha, V^\alpha, F^\alpha_i)=&\frac{1}{2} g_{\alpha\beta}(u) V^\alpha V^\beta \bar\varrho (x) dt\wedge d^N x\\
&- \bar{\mathcal{E}} ( x^i, u^\alpha , F^\alpha_i, \bar\varrho (x),\bar\varsigma (x),G^{ij}(x),g_{\alpha\beta}(u))dt\wedge d^N x
\end{aligned}
\end{equation}
in local coordinates. Note that the Lagrangian is defined with the help of the given tensor fields $\varrho=\bar\varrho d^Nx$, $\varsigma=\bar\varsigma d^Nx$, $G^{-1}= G^{ij} \frac{\partial}{\partial x^i}   \frac{\partial}{\partial x^j}$, and $g=g_{\alpha\beta} dx^\alpha dx^\beta$. We chose to work with the cometric $G^{ij}$ associated to $G_{ij}$, in order to directly get the Finger deformation (or left Cauchy-Green) tensor $b^{\alpha\beta}$, rather than its inverse, later.

The restricted multimomentum bundle for continuum mechanics is given by
\[
{\mathcal M}^0_{(t,x,u)}\pi \simeq T_{(t,x)} M\otimes T^*_uQ  = \operatorname{Lin} (T^*_{(t,x)}M, T^*_uQ )
\]
with coordinates $(t,x^i, u^\alpha, M_\alpha, P^i_\alpha)$. The restricted Legendre transform of the Lagrangian density is
\begin{equation}\label{restricted_Legendre}
Leg_{ \mathcal{L} }: J^1\pi\rightarrow {\mathcal M}^0\pi,\quad (t,x^i,u^\alpha,V^\alpha, F^\alpha_i)\mapsto (t,x^i, u^\alpha, M_\alpha, P^i_\alpha),
\end{equation}
with $M_\alpha$ and $P^i_\alpha$ given by
\begin{equation}\label{momenta}
M_\alpha = g_{\alpha\beta}(u) V^\beta\bar \varrho(x),\quad P^i_\alpha = -\frac{\partial \bar{\mathcal{E}}}{\partial F^\alpha_i},
\end{equation}
with $M_\alpha$ the momentum density (in the Lagrangian description) and $P^i_\alpha$ is the Piola-Kirchoff stress tensor density. Note that the coordinates $(M_ \alpha ,  P_ \alpha ^i)$ on the fiber of $ \mathcal{M} ^0 \pi $ correspond to the coordinates denoted $p_ \alpha ^i$ earlier.

The Eulerian versions of these tensor densities are the Eulerian momentum density $m_\alpha$ and the Cauchy stress tensor density $\sigma^{\alpha\beta}$ given by the Piola transformation
\begin{equation}\label{Eulerian_momenta}
m_\alpha= M_\alpha \operatorname{det}(F)^{-1}\quad\text{and}\quad \sigma^{\alpha\beta}=- F_i^\alpha P_\gamma ^i g^{\beta\gamma}\operatorname{det}(F)^{-1}.
\end{equation}
From the second relation, we have
\begin{equation}\label{P_sigma}
P_\alpha ^i =  - \operatorname{det}(F) ( F ^{-1} )_\gamma ^i  \sigma ^{\gamma\beta} g_{\beta\alpha}.
\end{equation}
Note that the first relation in \eqref{momenta} is always invertible, but the invertibility of the second relation depends in the potential energy density $\mathcal{E}$.
As we shall illustrate below, relation \eqref{P_sigma} is extremely useful to check if the restricted Legendre transform \eqref{restricted_Legendre} is an isomorphism, in which case we say that the Lagrangian density is hyperregular.

\subsubsection{The Hamiltonian density and the linear-affine bracket for Continuum Mechanics}

By assuming that $\mathcal{L}$ is hyperregular, we get the Hamiltonian $H\in C^\infty(\mathcal{M}^0\pi)$
\begin{equation}\label{HD_continuum}
H(t,x ^i ,u ^\alpha , M_\alpha, P_\alpha ^i )=\frac{1}{2} g^{\alpha\beta}(u)M_\alpha M_\beta \frac{1}{\bar \varrho(x)}+\bar{\mathcal{E}}-\frac{\partial \bar{\mathcal{E}}}{\partial F_i ^\alpha }F_i ^\alpha ,
\end{equation}
where $F_i^\alpha $ is expressed in terms of the variables in $\mathcal{M}^{0}\pi$ by inverting the second relation in \eqref{momenta}.

A section of the restricted momentum bundle $\mathcal{M}^0\pi$ is locally given by
\[
s_0(t,x^i )= \left( t, x^i , \varphi^\alpha (t, x^i ), M_\alpha(t,x^i),P_\alpha ^i (t,x^i ) \right)
\]
and, in the hyperregular case, the Euler-Lagrange equation are equivalent to the Hamilton-deDonder-Weyl equations given by
\begin{equation}\label{HdDW} 
\frac{\partial \varphi^\alpha }{\partial t }= \frac{\partial H}{\partial M_\alpha }  , \quad \frac{\partial \varphi^\alpha }{\partial x^i  }= \frac{\partial H}{\partial P_\alpha ^i } ,\quad\frac{\partial M_\alpha}{\partial t} + \frac{\partial P _\alpha ^i }{\partial x^i }= -\frac{\partial H}{\partial u ^\alpha }.
\end{equation} 
These equations admit the canonical linear-affine bracket formulation, that is, a section $(t, x^{i}) \to s^0(t, x^{i}) = (t, x^{i}, u^\alpha(t, x), M_\alpha(t, x), P_\alpha^{i}(t, x))$ is a solution of the previous equations if
\begin{equation}\label{Bracket_CM} 
(s^0)^*(d \alpha^0)  = \{\alpha^0, h \} \circ s^0, \;\;\forall\; \alpha^0 \in \mathcal{O},
\end{equation} 
where the currents $\alpha^0\in \mathcal{O}$ for Continuum Mechanics are of the form
\begin{align*} 
&\alpha^0(t,x ^i, u^\alpha, M_\alpha, P^i_\alpha)\\
&= \alpha^{00}(t,x ^i, u^\alpha, M_\alpha, P^i_\alpha)d^Nx -  \alpha^{0i}(t,x ^i, u^\alpha, M_\alpha, P^i_\alpha)dt \wedge d^{N-1}x_i\\
&= \big(Y^\alpha(t,x,u) M_\alpha + \beta^0(t,x,u) \big) d^Nx  -\big(Y^\alpha(t,x,u) P^i_\alpha + \beta^i(t,x,u) \big) dt \wedge d^{N-1}x_i
\end{align*}
with $d^{N-1}x_i= i_{\frac{\partial}{\partial x^i} } d^Nx$.
The canonical linear-affine bracket $\{\cdot, \cdot\}: \mathcal{O} \times C^\infty( \mathcal{M}^0\pi)\rightarrow C^\infty( \mathcal{M}^0\pi)$ is given by
\[
\begin{array}{l}
\displaystyle\{\alpha^{00} d^Nx -  \alpha^{0i}dt \wedge d^{N-1}x_i, H \} \\[8pt]
\displaystyle= \frac{\partial \alpha^{00} }{\partial t} + \frac{\partial \alpha^{0i} }{\partial x^i}+ \frac{\partial \alpha ^{00}}{\partial u^ \alpha } \frac{\partial H}{\partial M_\alpha }- \frac{1}{m} \frac{\partial H}{\partial u^\alpha }   \frac{\partial \alpha^{00} }{\partial M_ \alpha } + \frac{\partial \alpha ^{0i}}{\partial u^ \alpha } \frac{\partial H}{\partial P^i_ \alpha } - \frac{1}{m} \frac{\partial H}{\partial u^\alpha }  \frac{\partial \alpha^{0i} }{\partial P_ \alpha ^i}.
\end{array}
\]

This formulation assumes that the Legendre transform is invertible. Except in some simple situations, this invertibility is a priori difficult to check. We shall show below how to facilitate the approach by using two symmetries of the potential energy density $ \mathcal{E} $. The first one, the material covariance, is related to the isotropy of the continuum, while the second, the material frame indifference, is a general covariance assumption of continuum theories, see \cite{MaHu1983} and \cite{GBMaRa2012}.

We assume that the potential energy density is of the form
\[
\bar{\mathcal{E}} \big( x^i, u^\alpha, F _i ^\alpha   , G^{ij} , \bar\varrho   , \bar\varsigma   , g _{\alpha\beta} \big)=\operatorname{det}(F)\bar\epsilon\big(u^\alpha, \bar\varrho \operatorname{det}(F)^{-1} ,\bar \varsigma \operatorname{det}(F)^{-1} , F^\alpha_i G^{ij} F_j^\beta, g _{\alpha\beta}\big),
\]
where $\bar\epsilon$ is the potential energy density in the Eulerian description. This assumption is compatible with the assumption of material covariance. Here $\epsilon=\bar\epsilon d^Nu$ is a bundle map
\[
\epsilon : \Lambda^NT^*_u Q \times \Lambda^NT^*_u Q \times S^2 T_xQ \times S^2 T_x^*Q\rightarrow \Lambda^N_u Q
\]
covering the identity on $Q$.
From this expression, we compute the momenta from the second equality in \eqref{momenta} as
\begin{equation}\label{momenta_continuum}
\begin{aligned} 
P_\alpha ^i =\left( \Big(\bar \epsilon   -\frac{\partial\bar \epsilon  }{\partial \bar \rho    } \bar \rho      - \frac{\partial\bar \epsilon  }{\partial \bar s   } \bar s    \Big) (F ^{-1} )_\alpha ^i  + 2\frac{\partial\bar \epsilon  }{\partial b^{\alpha\beta}}F_i ^\gamma G^{ij} F_j ^\beta ( F ^{-1} )_\gamma^i    \right) \operatorname{det}(F),
\end{aligned}
\end{equation}  
where we introduced the notations
\[
\bar \rho   = \bar\varrho\operatorname{det}(F)^{-1},\quad \bar s   =\bar\varsigma\operatorname{det}(F)^{-1},\quad b^{\alpha\beta}=F^\alpha_i G^{ij} F_j^\beta.
\]
These are the local expressions of the mass density and entropy density in Eulerian description, and of the Finger deformation (or left Cauchy-Green) tensor.

The associated  Cauchy stress tensor density $\sigma$, see the second equation in \eqref{Eulerian_momenta}, is
\begin{equation}\label{expression_sigma}
\sigma ^{\alpha\beta} = \Big(\bar \epsilon   -\frac{\partial\bar \epsilon  }{\partial \bar \rho    } \bar \rho      - \frac{\partial\bar \epsilon  }{\partial \bar s   } \bar s    \Big) g^{\alpha\beta} + 2\frac{\partial\bar \epsilon  }{\partial b^{\gamma \delta}}F_i ^\alpha G^{ij} F_j ^\delta g^{\gamma\beta}.
\end{equation}
If in addition $\mathcal{E}$ satisfies the material frame indifference, then
\begin{equation}\label{two_symmetries}
\epsilon(\psi^*\rho, \psi^* s, \psi^* b, \psi^*g)= \psi^*\big(\epsilon(\rho,s,b,g)\big),
\end{equation}
for all diffeomorphisms $\psi$ of $Q$ and we have the Doyle-Ericksen formula
\begin{equation}\label{DE}
\sigma^{\alpha\beta}= 2\frac{\partial\bar \epsilon }{\partial g_{\alpha\beta}}.
\end{equation}

By inserting these relations into \eqref{HD_continuum}, we get the following result which is a step towards a more explicit expression of the Hamiltonian density, because in practice $\epsilon$, rather than $\mathcal{E}$, is given.

\begin{proposition} Assume that the Lagrangian is hyperregular and that $\mathcal{E}$ satisfies the two invariance mentioned above, and consider the associated Eulerian potential energy density $\epsilon$. Then, the Hamiltonian density of continuum mechanics  $H\in C^\infty(\mathcal{M}^0\pi)$ is given by
\begin{equation}\label{Ham_density_CM}
H(t,x^i,u^\alpha, M_\alpha, P_\alpha^i)=\frac{1}{2} g^{\alpha\beta}(u)M_\alpha M_\beta \frac{1}{ \bar\varrho (x)}+\Big(\bar \epsilon  - 2\frac{\partial\bar \epsilon  }{\partial g_{\alpha\beta} }g_{\alpha\beta}(u) \Big)\operatorname{det}(F).
\end{equation}
\end{proposition}
 
\subsubsection{Boundary conditions}\label{BC_CM}

We briefly describe two mains boundary conditions used in Continuum Mechanics following \S\ref{BC}. These conditions only arise at the spatial part of the boundary of the base manifold $M$, hence the bundle $B_0$ is over $[0,T] \times \partial B \subset \partial M$ only and the boundary condition reads $s( [0,T] \times \partial B) \subseteq B_0$.

For a continuum moving in a fixed domain $B' \subset Q$ diffeomorphic to $B$, we have the boundary condition $ \varphi (t, \partial B)= \partial B'$ on the motion and, in addition, the boundary condition on the Piola-Kirchhoff stress tensor $P$ given by $P_ \alpha^i(t,x) N^\flat_i(x)|_{TB'}=0$, for all $ x \in  \partial B$, with $N$ the normal vector field to $B$ with respect to $G$. This corresponds to zero tangential traction on the boundary, a condition that vanishes for fluids. In this case, the subbundle $B^0 \rightarrow [0,T] \times \partial B$ is given by
\[
B^0(t,x^i)= \{ (t,x^i,u^ \alpha , M_ \alpha , P_ \alpha ^i) \mid u^ \alpha \in \partial B', \;P_ \alpha ^i N_i^\flat  \in (TB')^\circ\}, \quad (t,x^i) \in [0,T] \times \partial B.
\]
In particular, we have $B_E^0 = [0,T] \times \partial B \times \partial B'$ and $B_0 \rightarrow B_E^0$ is a vector bundle.

For a free boundary continuum we take
\[
B^0(t,x^i)= \{ (t,x^i,u^ \alpha , M_ \alpha , P_ \alpha ^i) \mid P_ \alpha ^i N_i^\flat =0\}, \quad (t,x^i) \in [0,T] \times \partial B,
\]
which corresponds to zero traction on the boundary. This reduces to zero pressure at the boundary for fluids. We have $B_E^0 = [0,T] \times \partial B \times Q= E|_{  [0,T] \times \partial B}$ and $B_0 \rightarrow B_E^0$ again is a vector bundle.

\subsubsection{Fluid dynamics}

In this case the energy density $\epsilon$ only depends on the mass density and entropy density $\rho=\bar \rho   d^Nu$ and $s=\bar s   d^Nu$, so the Cauchy stress density is given by
\[
\sigma ^{\alpha\beta} = \Big(\bar \epsilon   -\frac{\partial\bar \epsilon  }{\partial \bar \rho    } \bar \rho      - \frac{\partial\bar \epsilon  }{\partial \bar s   } \bar s    \Big) g^{\alpha\beta}= - p(\bar \rho   , \bar s   ,g_{\alpha\beta})\sqrt{\det g} g^{\alpha\beta},
\]
see \eqref{expression_sigma}, where $p$ is the pressure of the fluid. In this case \eqref{P_sigma} yields
\[
P^i_\alpha = \operatorname{det}(F) (F^{-1})^i_\alpha p \big(\bar\varrho\operatorname{det}(F)^{-1},\bar\varsigma\operatorname{det}(F)^{-1},g_{\alpha\beta}\big) \sqrt{\det g}.
\]
This relation is of the form
\begin{equation}\label{relation_P_F_fluid}
P^i_\alpha = \mathsf{f}( \operatorname{det}(F))  (F^{-1})^i_\alpha,
\end{equation}
for some function $\mathsf{f}$. If the function $x\mapsto {\mathsf{f}(x)^N}{x}$ is invertible on $]0,\infty[$, with inverse $\mathsf{g}$, then relation \eqref{relation_P_F_fluid} is invertible, with inverse
\begin{equation}\label{inverse}
F^\alpha_i = \mathsf{f}(\mathsf{g}^{-1}(\operatorname{det}(P))) (P^{-1})^\alpha_i.
\end{equation}
In this case the Lagrangian density is hyperregular. Note that the function $\mathsf{f}$, and hence the hyperregularity, depends on the state function of the fluid, i.e., the relation $\epsilon=\epsilon (\rho, s, g)$.

Hyperregularity is satisfied for a large class of state equations, including the important case of a perfect gas for which $\epsilon(\rho, s, g)=\epsilon_0 e^{\frac{1}{C_v}\left(\frac{s}{\rho}- \frac{s_0}{\rho_0}\right)}\big(\frac{\rho}{\rho_0\mu(g)}\big)^{\gamma}\mu(g)$, where $\gamma= C_p/C_v$ is the adiabatic index and $\mu(g)$ is the volume form associated to $g$, i.e. $\mu(g)= \sqrt{\det g}d^Nu$. In this case, we compute the pressure as $p\sqrt{\det g}= (\gamma - 1)\bar \epsilon $.

Note that, as it should, $\epsilon$ satisfies \eqref{expression_sigma}. Computing the derivative of $\bar\epsilon$ with respect to the Riemannian metric, we get
\[
\frac{\partial\bar \epsilon }{\partial g_{\alpha\beta}}= \frac{1}{2}  (1- \gamma )\bar\epsilon g^{\alpha\beta},
\]
so one directly checks that the Doyle-Ericksen formula \eqref{DE} is verified.

For fluids, the Hamiltonian density is
\[
H(t,x^i,u^\alpha, M_\alpha, P_\alpha^i)=\Big(\frac{1}{2} g^{\alpha\beta}(u)M_\alpha M_\beta \frac{1}{\bar \varrho (x)}+\Big(\bar\epsilon + N p \sqrt{\det g}\Big)\operatorname{det}(F) \Big) dt\wedge d^Nx,
\]
where $\operatorname{det}(F)$ is found from \eqref{inverse}.
In particular, for the perfect gas, we have $\bar\epsilon + N p \sqrt{\det g}=(1+N(\gamma-1))\bar\epsilon$.

The fluid equations can thus be written in the canonical linear-affine bracket form \eqref{Bracket_CM}.

\subsubsection{Nonlinear elasticity}

 In general, the Hamiltonian density in nonlinear elasticity takes a complicate expression due to the dependence of $ \epsilon $ on the Finger deformation tensor $b$.
For example, for the compressible neo-Hookean material (see \cite{SiTaPi1985}, \cite{BaWe2000}), with $N=3$, the energy density is
\[
\epsilon ( \rho , b, g)= \frac{1}{2} \kappa \left( \ln J \right) ^2 \rho  + \frac{1}{2} \mu \left( J ^{-2/3} \operatorname{Tr}_g(b)-3 \right) \rho  , \qquad J:= \frac{ \mu (g)}{ \mu(b ^\flat)} ,
\]
where $ \kappa $ is the bulk modulus, $ \mu $ is the Lam\'e constant, and $ \mu (b^\flat)$ is the volume form associated to the Riemannian metric $b^\flat$, obtained by lowering the indices of $b$. One observes that \eqref{expression_sigma} is satisfied. 
The Doyle-Ericksen formula yields the expression of the stress tensor density
\[
\sigma = 2 \frac{\partial \epsilon }{\partial g}= \kappa (\ln J ) g^\sharp \rho+ \mu J^{-2/3} \left( b- \frac{1}{3} \operatorname{Tr}_g(b) g^\sharp\right) \rho  . 
\]
We thus get $ 2 \frac{\partial \epsilon }{\partial g}\!:\!g= 3 \kappa \ln J\rho $ which can then be inserted in \eqref{Ham_density_CM} to yield the Hamiltonian density.

We shall illustrate the derivation of the Hamiltonian density by considering the simplified situation $ \epsilon ( \rho ,b, g)= \frac{1}{2} \operatorname{Tr}_g(b) \rho$. In this case $ \sigma = b \rho $, so we get the momenta
\[
P_\alpha ^i =- G^{ij} F_j ^\beta g_{\alpha\beta}\bar\varrho.
\]
Using this and $ \sigma \!: \!g= \operatorname{Tr}_g(b) \rho $, we get the Hamiltonian density
\[
H(t,x^i ,u ^\alpha , M_\alpha, P_\alpha ^i )= \left(\frac{1}{2} g^{\alpha\beta}(u)M_\alpha M_\beta \frac{1}{ \varrho _{\rm loc}(x)}- \frac{1}{2} P_\alpha^iG_{ij}P_\beta^jg_{\alpha\beta} \frac{1}{\varrho _{\rm loc}(x)}\right)dt\wedge d^3x.
\]
The nonlinear elasticity equations can thus be written in the canonical linear-affine bracket form \eqref{Bracket_CM}.

\color{black}

\subsection{Yang-Mills theory}
Yang-Mills theory may be considered as a singular Lagrangian field theory of first order associated with a principal $G$-bundle over an oriented Riemannian (or a Lorentzian manifold) space $M$ (possibly with boundary) of dimension $m$ and where $G$ is a compact Lie group of dimension $n$ (we will follow \cite{IbSp1}).

We will denote by $g$ the metric on $M$. For simplicity, we will assume that the principal bundle is trivial, $g$ is a Riemannian metric and $\partial M = \phi$. 

Under the previous conditions, the configuration bundle of the theory is the vector bundle 
\[
\pi_{M, \frak{g}}: E: =T^*M \otimes \frak{g} \to M
\]
where $\frak{g}$ is the Lie algebra of $G$.

Then, we will proceed as follows. We will introduce a Lagrangian density on the $1$-jet bundle of the fibration $\pi_{M, \frak{g}}: T^*M \otimes \frak{g} \to M$. This Lagrangian density is singular. In fact, the image of the corresponding Legendre transformation is a proper submanifold $\mathcal{M}^1$ of the restricted multimomentum bundle ${\mathcal M}^0\pi_{M, \frak{g}}$. Using the restricted and the extended Legendre transformation, we will construct a constrained Hamiltonian section $h_1: \mathcal{M}^1 \to \mu^{-1}(\mathcal{M}^1) \subseteq {\mathcal M}\pi_{M, \frak{g}}$ of the fibration $\mu_{|\mu^{-1}(\mathcal{M}^1)}: \mu^{-1}(\mathcal{M}^1) \to \mathcal{M}^1$. Now, if we consider an (arbitrary) hamiltonian section $h: {\mathcal M}^0\pi_{M, \frak{g}} \to {\mathcal M}\pi_{M, \frak{g}}$, whose restriction to $\mathcal{M}^1$ coincides with $h_1$, we will obtain a Hamiltonian field theory in such a way that the solutions of the Hamilton-deDonder-Weyl equations for $h$ which are contained in $\mathcal{M}^1$ are just the solutions of the corresponding Yang-Mills theory.

\subsubsection{The Lagrangian formalism}
Note that the sections of the vector bundle $\pi_{M, \frak{g}}: T^*M \otimes \frak{g} \to M$ are the principal connections on the trivial principal bundle $pr_1: M \times G \to M$. As we know, the $1$-jet bundle $J^1 \pi_{M, \frak{g}}$ is an affine bundle over $T^*M \otimes \frak{g}$. The key point is that there is a canonical epimorphism off affine bundles (over the vector bundle projection $\pi _{M, \mathfrak{g} }: T^*M \otimes \mathfrak{g} \rightarrow M$), $F: J^1 \pi_{M, \frak{g}} \rightarrow \Lambda ^2T^*M \otimes \mathfrak{g} $, which is characterized by the condition
\[
F(j^1 \Theta (x))=  {\rm d} \Theta (x)+ [ \Theta(x) , \Theta (x)] , \; \; \mbox{ for } x\in M,
\]
for all principal connections $ \Theta $. In other words, the image by $F$ of the $1$-jet bundle of a principal connection is just the curvature of the connection.

If $(x^{i})$ are local coordinates on $M$ and $\{e_\alpha\}$ is a basis of $\frak{g}$, we have the corresponding local coordinates $(x^{i}, u^\alpha_i)$ on $E$ and $(x^{i}, u^\alpha_i, u^\alpha_{ij})$ on $J^1\pi_{M, \frak{g}}$. Moreover,
\[
F(x^{i}, u^\alpha_i, u^\alpha_{ij}) = \frac{1}{2} F_{kl}^\gamma(x^{i}, u^\alpha_i, u^\alpha_{ij})(dx^k \wedge dx^l) \otimes e_\gamma
\]
with
\begin{equation}\label{curvature_local} 
F_{kl}^\gamma(x^{i}, u^\alpha_i, u^\alpha_{ij}) =  u^\gamma_{lk} - u^\gamma_{kl} + c_{\alpha\beta}^\gamma u^\alpha_k u^\beta_l.
\end{equation} 
Here, $c_{\alpha\beta}^\gamma$ are the structure constants of the Lie algebra $\frak{g}$ with respect to the basis $\{e_\gamma\}$.

Next, we will introduce the Lagrangian density
\[
\mathcal{L}: J^1\pi_{M, \frak{g}} \to \pi_{M, \frak{g}}^*(\Lambda^mT^*M).
\]
First of all, since the manifold $M$ is oriented, the vector bundle $\pi_{M, \frak{g}}^*(\Lambda^mT^*M) \to E$ is the trivial line bundle $E \times \mathbb{R} \to E$. So, the Lagrangian density ${\mathcal L}$ is, in fact, a real $C^\infty$-function $L: 
J^1\pi_{M, \frak{g}} \to \mathbb{R}$.

In addition, we will fix an $Ad$-invariant scalar product $\langle 	\cdot, \cdot \rangle$ on $\frak{g}$ (which is possible, since $G$ is compact). Then, the scalar product on $\frak{g}$ and the Riemannian metric on $M$ induce a bundle metric on the vector bundle 
\[
pr_1: E \times_M \Lambda^2T^*M \otimes \frak{g} \to E.
\]
So, we can consider the real function $L: J^1\pi_{M, \frak{g}} \to \mathbb{R}$ given by
\[
L(z) = \displaystyle \frac{1}{4} \|F(z)\|^2, \; \; \mbox{ for } z \in J^1\pi_{M, \frak{g}},
\]
where the norm is taken with respect to the bundle metric on the vector bundle $
pr_1: E \times_M \Lambda^2T^*M \otimes \frak{g} \to E.$

The local expression of $L$ is
\[
L(x^{i}, u^\alpha_i, u^\alpha_{ij}) = \displaystyle \frac{1}{4} F_{kl}^\gamma(x^{i}, u^\alpha_i, u^\alpha_{ij}) F^{kl}_\gamma(x^{i}, u^\alpha_i, u^\alpha_{ij})
\]
with 
\[
F^{kl}_\gamma = F_{mn}^\beta g^{km}g^{ln}\langle\cdot, \cdot\rangle_{\beta \gamma}
\]
and $(g_{ij})$ the matrix of the coefficients of $g$, $(g^{ij})$ the inverse matrix and 
\[
\langle\cdot, \cdot\rangle_{\beta \gamma} = \langle e_\beta, e_\gamma \rangle.
\]
Thus, the Euler-Lagrange equations for $L$
\[
\displaystyle \frac{\partial}{\partial x^j}\left(\frac{\partial L}{\partial u^\alpha_{ij}}\right) - \displaystyle \frac{\partial L}{\partial u^\alpha_i} = 0, \; \; u^\alpha_{ij} = \displaystyle \frac{\partial u^\alpha_i}{\partial x^j}
\]
are, in this case, the well-known Yang-Mills equations
\begin{equation}\label{Yang-Mills-equation}
\displaystyle \sum_i \frac{\partial F^{ij}_\alpha}{\partial x^{i}} + c_{\alpha\beta}^\gamma u^\beta_i F^{ij}_\gamma =0, \; \; 
u^\alpha_{ij} = \displaystyle \frac{\partial u^\alpha_i}{\partial x^j}.
\end{equation}

\subsubsection{The Legendre transformations and the constrained Hamiltonian formalism}
First of all, we will consider the restricted Legendre transformation associated with $L$ 
\[
leg_L: J^1\pi_{M,\frak{g}} \to \mathcal{M}^0\pi_{M, \frak{g}}.
\]
Note that, since $M$ is oriented, the restricted multimomentum bundle $\mathcal{M}^0\pi_{M, \frak{g}}$ may be identified with the dual bundle $V^*(J^1\pi_{M,\frak{g}})$ of  $V(J^1\pi_{M,\frak{g}})$. So,
\[
\mathcal{M}^0\pi_{M, \frak{g}} \simeq V^*(J^1\pi_{M,\frak{g}}) \simeq E \times_M (TM \otimes TM \otimes \frak{g}^*).
\]
The transformation $leg_L$ is given by
\[
leg_L(z)(v) = \frac{d}{dt}_{|t=0} L(z + tv),
\]
for $z \in J^1_y\pi_{M, \frak{g}}$, $v \in \mathcal{M}_y^0\pi_{M, \frak{g}}$ and $y \in E$. The local expression of $leg_l$ is
\[
leg_L(x^{i}, u^\alpha_i, u^\alpha_{ij}) = (x^{i}, u^\alpha_i, \frac{\partial L}{\partial u^\alpha_{ij}}) = (x^{i}, u^\alpha_i, -F^{kl}_\gamma(x^{i}, u^\alpha_i, u^\alpha_{ij})).
\]
This implies that the image of $leg_L$ is the vector subbundle ${\mathcal M}^1$ (over $E$) of $\mathcal{M}^0\pi_{M, \frak{g}}$
\[
\mathcal{M}^1 \simeq E \times_M (\Lambda^2TM \otimes \frak{g^*}).
\]
Thus, the map $leg_1: J^1\pi_{M,\frak{g}} \to \mathcal{M}^1$ is a submersion with connected fibers and $L$ is almost regular.

On the other hand, we can consider the extended Legendre transformation $Leg_L: J^1\pi_{M,\frak{g}} \to \mathcal{M}\pi_{M, \frak{g}} \simeq  {\rm Aff}(J^1\pi_{M,\frak{g}}, \mathbb{R})$ associated with $L$ defined by
\[
Leg_L(z)(z') = \displaystyle \frac{d}{dt}_{|t=0} L(z + t(z' -z)), \; \; \mbox{ for } z, z' \in J^1_y\pi_{M,\frak{g}} \mbox{ and } y \in E.
\]
The local expression of $Leg_L$ is
\[
Leg_L(x^{i}, u^\alpha_i, u^\alpha_{ij}) = (x^{i}, u^\alpha_i, -E_L(x^{i}, u^\alpha_i, u^\alpha_{ij}), \frac{\partial L}{\partial u^\beta_{kl}}) = (x^{i}, u^\alpha_i, -E_L(x^{i}, u^\alpha_i, u^\alpha_{ij}), -F_{kl}^\beta(x^{i}, u^\alpha_i, u^\alpha_{ij})),
\]
where
\[
E_L(x^{i}, u^\alpha_i, u^\alpha_{ij}) = \displaystyle \frac{1}{4} F_{ij}^\alpha(x^{k}, u^\beta_k, u^\beta_{kl}) F^{ij}_\alpha(x^{k}, u^\beta_k, u^\beta_{kl}) - \displaystyle \frac{1}{2} c_{\alpha\beta}^\gamma u^\alpha_i u^\beta_j F^{ij}_\gamma(x^{k}, u^\beta_k, u^\beta_{kl}).
\]
Note that if $\mu: \mathcal{M}\pi_{M, \frak{g}} \to \mathcal{M}^0\pi_{M, \frak{g}}$ is the canonical projection then the image of $Leg_L$ is a submanifold $\mathcal{M}$ of $\mu^{-1}(\mathcal{M}^1) \subseteq \mathcal{M}\pi_{M, \frak{g}}$ which is diffeomorphic to $\mathcal{M}^1$, via the restriction of $\mu$ to $\mathcal{M}$. The following diagram illustrates the situation
$$
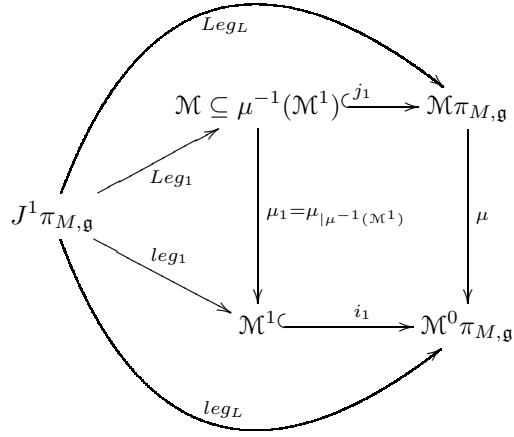
\xymatrix{&&\\&{\mathcal M} \subseteq \mu^{-1}({\mathcal M}^1)\ar@{^(->}[r]^{j_1}\ar[dd]^{\mu_1=\mu_{|\mu^{-1}({\mathcal M}^1)}}&{\mathcal M}\pi_{
M,{\mathfrak g}}\ar[dd]^\mu\\ J^1\ar[ru]_{Leg_1}\pi_{M,{\mathfrak g}}
 \ar@/_-5pc/@{->}[urr]_{Leg_L}  \ar@/_5pc/@{->}[drr]^{leg_L}
	\ar[rd]^{leg_1}&&\\&{\mathcal M}^1 \ar@{^(->}[r]^{i_1}&{\mathcal M}^0\pi_{M,{\mathfrak g}}}$$


The maps $Leg_1: J^1\pi_{M,\frak{g}} \to \mathcal{M}$ and $leg_1: J^1\pi_{M,\frak{g}} \to \mathcal{M}^1$ are surjective submersions and, thus, we have a constrained Hamiltonian field theory. As a consequence (see, for instance, \cite{DeMaMa1}), one may introduce a constrained Hamiltonian section 
\[
h_1: \mathcal{M}^1 \to \mu^{-1}(\mathcal{M}^1)
\]
in such a way that
\[
h_1 \circ leg_1 = Leg_1.
\]
In fact, $h_1= (\mu |_\mathcal{M}) ^{-1}: \mathcal{M} ^1 \rightarrow \mathcal{M} $. In addition, if $(x^{i}, u^\alpha_i, p, p_\alpha^{ij})$ and $(x^{i}, u^\alpha_i, p_\alpha^{ij})$ are the standard local coordinates on  
$\mathcal{M}\pi_{M, \frak{g}}$ and $\mathcal{M}^0\pi_{M, \frak{g}}$, respectively, we can take local coordinates
\[
(x^{i}, u^\alpha_i, p, \pi_\alpha^{ij}) \; \mbox{ and } (x^{i}, u^\alpha_i, \pi_\alpha^{ij})
\]
on $\mu^{-1}(\mathcal{M}^1)$ and $\mathcal{M}^1$, respectively, with
\begin{equation}\label{pi-p}
\frac{1}{2} \pi^{ij}_\alpha = p^{ij}_\alpha, \; \; \mbox{ for } i < j.
\end{equation}
Then,
\[
h_1(x^{i}, u^\alpha_i, \pi_\alpha^{ij}) = (x^{i}, u^\alpha_i, -H_1(x^{i}, u^\alpha_i, \pi_\alpha^{ij}), \pi_\alpha^{ij})
\]
where 
\begin{equation}\label{Expresion-local-H-1}
H_1(x^{i}, u^\alpha_i, \pi_\alpha^{ij}) = \displaystyle \frac{1}{4} p_{ij}^\alpha p^{ij}_\alpha + \displaystyle \frac{1}{2} c_{\alpha\beta}^\gamma u^\alpha_i u^\beta_j p^{ij}_\gamma = \displaystyle \frac{1}{16} \pi_{ij}^\alpha \pi^{ij}_\alpha + \displaystyle \frac{1}{4} c_{\alpha\beta}^\gamma u^\alpha_i u^\beta_j \pi^{ij}_\gamma.
\end{equation}
Note that we are assuming
\[
\pi^{ij}_\alpha = -\pi^{ji}_\alpha \; \; \mbox{ if } i > j \; \; \mbox{ and } \pi^{ii}_\alpha = 0.
\]
In addition,
\[
\pi_{ij}^\alpha = \pi^{kl}_\beta g_{ik}g_{jl} \langle \cdot, \cdot \rangle^{\alpha \beta}.
\]
It is clear that
\[
E_L= H_1  \circ leg_1.
\]
Moreover, on $\mathcal{M}^1$
\begin{equation}\label{Yang-Mills-1}
- F^{ij}_\alpha= \frac{1}{2} \pi^{ij}_\alpha  \circ leg_1.
\end{equation}
So, using (\ref{curvature_local}), (\ref{Yang-Mills-equation}) and (\ref{Yang-Mills-1}), we obtain that
\begin{equation}\label{Yang-Mills-2}
\displaystyle \frac{\partial \pi^{ij}_\alpha}{\partial x^{i}} + c_{\alpha \beta}^{\gamma}u^\beta_{i} \pi_\gamma^{ij} = 0.
\end{equation}
(\ref{Yang-Mills-1}) and (\ref{Yang-Mills-2}) are just the Yang-Mills equations for the Yang-Mills theory in the Hamiltonian side. Thus, Yang-Mills theory may be considered as a constrained (singular) Hamiltonian field theory. 

On the other hand, if $h: \mathcal{M}^0\pi_{M, \frak{g}} \to \mathcal{M}\pi_{M, \frak{g}}$ is  a Hamiltonian section which extends $h_1$ (that is, $h_{|\mathcal{M}^1} = h_1$), then we may consider the corresponding Hamiltonian field theory associated with $h$. Furthermore, using the classical results on  singular Lagrangian field theories (see \cite{DeMaMa1}), if $s^0: U\subseteq M \to \mathcal{M}^0\pi_{M,\frak{g}}$ is a solution of the Hamilton-deDonder-Weyl equations for $h$ which is contained in $\mathcal{M}^1$ then $s^0$ is just a solution of the Yang-Mills equations.

The following diagram illustrates the situation
$$
\xymatrix{&&\\&{\mathcal M} \subseteq \mu^{-1}({\mathcal M}^1)\ar@{^(->}[r]^{j_1}\ar[dd]^{\mu_1=\mu_{|\mu^{-1}({\mathcal M}^1)}}&{\mathcal M}\pi_{
M,{\mathfrak g}}\ar[dd]^\mu\\ J^1\ar[ru]_{Leg_1}\pi_{M,{\mathfrak g}}
	\ar[rd]^{leg_1}\ar[rdd]_{(\pi_{M,{\mathfrak g}})_{0,1}}&&\\&{\mathcal M}^1\ar@<1ex>[uu]^{h_1} \ar@{^(->}[r]^{i_1}&{\mathcal M}^0\pi_{M,{\mathfrak g}}\ar@<1ex>[uu]^{h}\\&M\ar[u]_{s^0}\ar[ur]_{i_1\circ s_0}&}$$

\subsubsection{The Lie algebra of currents and the linear-affine bracket for the extended Hamiltonian field theory}
After the previous subsections, we could apply all the machinery in this paper for the extended Hamiltonian field theory  and, as a consequence, we could deduce results on the Yang-Mills theory. This will be the subject of a future research. Anyway, we will remark a couple of general facts on the Lie algebra of currents, the linear-affine bracket (in Section \ref{suitable_LA_bracket}) and the Yang-Mills equations as constrained Hamilton-deDonder-Weyl equations:
\begin{itemize}
\item
First of all, following the proof of Theorem \ref{Lie-algebra-O}, we have that the space of currents, as a $C^\infty(E)$-module, may be identified with the product $\Gamma(V\pi_{M, \frak{g}}) \times \Gamma(\Lambda^{m-1}_1T^*E)$. But, since the configuration bundle $\pi_{M, \frak{g}}: T^*M \otimes \frak{g} \to M$ is a vector bundle, we have that $\Gamma(V\pi_{M, \frak{g}})$ is generated by vertical lifts of sections of the projection $\pi_{M,\frak{g}}$ (see Appendix \ref{ver-lift-section-vector-bundle}). In fact, if
\[
\theta = \theta_i(x)dx^{i}
\]
is a $1$-form on $M$ and $\xi\in \frak{g}$ then the local expression of the vertical lift of the section $s= \theta \otimes \xi$ is
\[
(\theta \otimes \xi)^{\bf v} (x^{i}, u^\alpha_i) = \theta_i(x) \xi^\alpha \frac{\partial}{\partial u^\alpha_i}.
\]
So, if we chose a local basis of $1$-forms and $(m-1)$-forms on $M$
\[
\{\theta_1, \dots, \theta_m\}, \; \; \{\alpha_1, \dots, \alpha_m\},
\]
respectively, we have a local basis 
\[
\{(\theta_i \otimes e_\gamma)^{\bf v}, \pi_{M,\frak{g}}^*(\alpha_i)\}_{i=1,1\dots,m, \gamma = 1, \dots, n} 
\]
of the space $\Gamma(V\pi_{M, \frak{g}}) \times \Gamma(\Lambda^{m-1}_1T^*E) \simeq \mathcal{O}$. Moreover, following the proof of Theorem \ref{Lie-algebra-O}, we also deduce that the Lie brackets in $\mathcal{O}$ between the previous sections are all zero. Note that if $s_1, s_1$ are sections of the vector bundle $\pi_{M, \frak{g}}: T^*M \otimes \frak{g} \to M$ then
\[
[s_1^{\bf v}, s_2^{\bf v}] = 0, \; \; i_{s_1^{\bf v}}d(\pi_{M,\frak{g}}^*\alpha) = 0,
\]
for $\alpha \in \Gamma(\Lambda^{m-1}T^*M)$.
\item 
Consider the linear-affine bracket
\[
\{\cdot, \cdot\}: \mathcal{O} \times \Gamma(\mu) \to \Gamma((\pi \circ \nu^0)^*(\Lambda^mT^*M))
\]
for the Hamiltonian field theory which extends Yang-Mills theory. We  have that if $\theta = \theta_i(x)dx^{i}$ is a $1$-form on $M$, $\xi = \xi^\alpha e_\alpha \in \frak{g}$, $\beta = \beta^{i}(x) d^{m-1}x_i$ is a $(m-1)$-form on $M$ and
\[
h(x^{i}, u^\alpha_i, p^{ij}_\alpha) = (x^{i}, u^\alpha_i, -H(x^{j}, u^\beta_j, p^{jk}_\beta), p^{ij}_\alpha)
\]
is a section of $\mu: \mathcal{M}\pi_{M, \frak{g}} \to \mathcal{M}^0\pi_{M, \frak{g}}$ then
\[
\begin{array}{rcl}
\{\widehat{(\theta \otimes \xi)^{\bf v}}, h \} & = &\left(  \displaystyle \frac{\partial \theta_i}{\partial x^j}(x) \xi^\alpha p^{ij}_\alpha - \displaystyle \frac{\partial H}{\partial u^\alpha_i}(x^j, u^\beta_j, p^{ij}_\beta)\theta_i(x) \xi^\alpha \right) d^mx \\[8pt]

\{\pi_{M, \frak{g}}^*\beta, h\} & = & \displaystyle \frac{\partial \beta^i}{\partial x^{i}}d^mx.

\end{array}
\]
In particular, if $h$ is an extension of the Yang-Mills Hamiltonian section $h_1: \mathcal{M}^1 \to \mathcal{M}\pi_{M, \frak{g}}$ then, $s^0: U \subseteq M \to \mathcal{M}^0\pi_{M, \frak{g}}$ is a solution of the Yang-Mills equations if and only if
\[
s^0(U) \subseteq \mathcal{M}^1
\]
and 
\[
\begin{array}{rcl}
(s^0)^*(d\widehat{(\theta \otimes \xi)^{\bf v}}) & = & \{\widehat{(\theta \otimes \xi)^{\bf v}}, h \} \circ s^0 \\[8pt]
(s^0)^*(d\beta)  & = & \{\pi_{M, \frak{g}}^*\beta, h\} \circ s^0,
\end{array}
\]
for $\theta$ a $1$-form on $M$, $\xi \in \frak{g}$ and $\beta$ a $(m-1)$-form on $M$.
\end{itemize}


\section{Conclusions and future work}\label{Conclusions-Future}
In this paper, we have developed a completely canonical geometric formulation of Hamiltonian Classical Field Theories of first order which is analogous to the canonical Poisson formulation of time-independent Hamiltonian Mechanics. This formulation is valid for any configuration bundle and is independent of any external structures such as connections or volume forms. We have defined a space of currents and endowed it with a Lie algebra structure, and we have shown that the bracket induces an affine representation of the Lie algebra of  currents on the affine space of Hamiltonian sections. An important difference with the case of time-independent Hamiltonian Mechanics is the linear-affine character of our bracket, which is consistent with the fact that the set of currents and the set of Hamiltonian sections are linear and affine spaces, respectively. We have applied our results to several examples and we have proved their effectiveness.

The results of this paper open some interesting future directions of research: 
\begin{itemize}
\item
Develop appropriate processes of reduction by symmetry for Hamiltonian Classical Field Theories of first order by exploiting the canonical linear-affine bracket formulation proposed in this paper.
\item
Include boundary conditions in the geometric formulation (see \cite{MaVi}) and discuss the relation between the resultant construction and the theory of Peierls brackets \cite{Pe} in the space of solutions (see \cite{FoRo,FoSa}; see also the recent papers \cite{CiDiIbMaScZa1,CiDiIbMaScZa2} and the references therein). 
\item
Discuss a canonical affine formulation of Lagrangian Classical Field Theories of first oder and obtain the equivalence with the Hamiltonian formulation for the case when Lagrangian density is almost regular (this is, for instance, the case of Yang-Mills theories discussed in this paper).
\end{itemize}

\appendix
\section{The $1$-jet bundle associated with a fibration} \label{1-jet-bundle}
Let $\pi: E \to M$ be a fibration, that is, a surjective submersion from $E$ to $M$. Assume that $dim \; M = m$ and $dim \; E = m+n$.

The $1$-jet bundle $J^1\pi$ associated with $\pi$ is the affine bundle over $E$ whose fiber at the point $y \in E$ is
\[
J_y\pi = \{ z: T_{\pi(y)}M \to T_yE \mid z \mbox{ is linear and } T_y\pi \circ z = id \}.
\]
$J^1\pi$ is modelled over the vector bundle 
\[
V(J^1\pi) = \pi^*(T^*M) \otimes V\pi \simeq \cup_{y \in E} {\rm Lin} (T_{\pi(y)}M, V_y\pi),
\]
where $V\pi$ is the vertical bundle of $\pi$.

In fact, if $z \in J^1_y\pi$ and $v: T_{\pi(y)}M \to V_y\pi \in V_y(J^1\pi)$ then one may define, in a natural way, the element $z + v$ of $J^1_y\pi$ as the linear map
\[
(z + v)(u) = z(u) + v(u), \; \; \mbox{ for all } u \in T_{\pi(y)}M.
\]
If $(x^{i}, u^{\alpha})$ are local coordinates on $E$ which are adapted to the fibration $\pi$ then one may consider the corresponding local coordinates on $J^1\pi$. Indeed, if $z \in J^1_y\pi$, $y \in E$ has local coordinates $(x^{i}, u^{\alpha})$ and
\[
z\left(\frac{\partial}{\partial x^{i}}_{|\pi(y)}\right) = \frac{\partial}{\partial x^{i}}_{|y} + u^{\alpha}
_i \frac{\partial}{\partial u^{\alpha}}_{|y}, \; \; \forall i \in \{1, \dots, m\},
\]
then $z$ has local coordinates $(x^{i}, u^{\alpha}, u^{\alpha}_i)$.

Sections of the affine bundle $\pi_{1, 0}: J^1\pi \to E$ may be identified with Ehresmann connections on the fibration $\pi: E \to M$.

We recall that an Ehresmann connection on $\pi: E \to M$ is a vector subbundle $H$ over $E$ of $TE$ satisfying
\[
TE = H \oplus V\pi.
\]
Note that if $H$ is an Ehresmann connection on $\pi: E \to M$, then one may define the horizontal lift associated with $H$ as a vector bundle morphism
\[
^{H}: E \times_M TM \to H \subseteq TE
\]
In fact, if $(y, u) \in E \times T_{\pi(y)}M$ then $(y, u)^H$ is the unique vector in $H_y$ whose projection over $M$ is $u$, that is,
\[
(T_y\pi)((y, u)^H) = u.
\]
It is clear that the horizontal lift induces, in a natural way, a section $s^H$ of the $1$-jet bundle $J^1\pi$. The previous correspondence is one-to-one.
Indeed, if
\[
\left(\frac{\partial}{\partial x^{i}}\right)^H = \frac{\partial}{\partial x^{i}} + H^{\alpha}_i(x, u) \frac{\partial}{\partial u^{\alpha}},
\]
then
\[
s^H(x^{i}, u^{\alpha}) = (x^{i}, u^{\alpha}, H^{\alpha}_i(x, u)).
\]
A (local) section $s: U \subseteq M \to E$ is said to be horizontal with respect to the Ehresmann connection $H$ if
\[
(T_xs)(u) \in H(s(x)), \; \; \forall x \in U, \; \; \forall u \in T_xM.
\]
So, if 
\[
s(x^{i}) = (x^{i}, u^{\alpha}(x)),
\]
then $s$ is a horizontal section of $H$ if and only if it satisfies the following system of partial differential equations
\[
\frac{\partial u^{\alpha}}{\partial x^i} = H^{\alpha}_{i}(x, u(x)), \; \; \forall\; i, \alpha.
\]
The Ehresmann connection $H$ is said to be integrable if the distribution $H$ is completely integrable. In such a case, for every point $y \in E$ there exists a unique horizontal (local) section $s: U \subseteq M \to E$ of $H$ such that $\pi(y) \in U$ and
\[
s(\pi(y)) = y.
\]

\section{The vertical lift of a section of a vector bundle}\label{ver-lift-section-vector-bundle}\label{ver-lift-sec}
In this appendix, we review the definition of the vertical lift of a section of a vector bundle. 

Let $\tau: W \to M$ be a vector bundle and $s: M \to W$ a smooth section of $\tau: W \to M$. Then, one may define the vertical lift $s^{\bf v}$ of $s$ as a vector field on $W$ given by
\[
s^{\bf v}(w) = \frac{d}{dt}_{|t=0}(w + t s(x)), \; \; \,\mbox{ for } w \in W_x \mbox{ and } x \in M.
\]
It is clear that $s^{\bf v}$ is vertical with respect to the vector bundle projection $\tau$.

Next, we will obtain a local expression of $s^{\bf v}$. For this purpose, we consider local coordinates $(x^{i})$ in an open subset $O$ 
of $M$ and a local basis $\{e_a\}$ of $\Gamma(W)$ in $O$. Then, we will denote by $(x^{i}, z^{a})$ the corresponding local coordinates on $W$. Moreover, if
\[
s(x^{i}) = s^{a}(x^{i})e_a(x^{i}), \; \; \mbox{ with } s^{a} \in C^{\infty}(O),
\]
then
\[
s^{\bf v}(x^{i}, z^{a}) = s^{a}(x^{i})\frac{\partial}{\partial z^{a}}_{|(x, z)}.
\]
A particular case of the previous situation is the following one. 

Let $\pi: E \to M$ be a fibration, with $dim M = m$, and $\nu: {\mathcal M}\pi = \Lambda_2^m(T^*E) \to E$ the multimomentum bundle associated with the fibration $\pi: E \to M$. It is a vector bundle over $E$ (see Section \ref{res-ext-multimomentum-bundle}).

Now, if $\Omega \in \Gamma(\Lambda^mT^*M)$ is a $m$-form on $M$ then $\pi^*\Omega$ is a $m$-form on $E$. Moreover, $\pi^*\Omega$ also is a section of the vector bundle $\nu: {\mathcal M}\pi = \Lambda^m_2(T^*E) \to E$. Thus, one may consider the vertical lift $\Omega^{\bf v} := (\pi^*\Omega)^{\bf v}$ of $\pi^*\Omega$ as a $\nu$-vertical vector field on ${\mathcal M}\pi$.
In fact, if $(x^{i}, u^{\alpha}, p, p_\alpha^{i})$ are canonical coordinates on ${\mathcal M}\pi$ as in Section \ref{res-ext-multimomentum-bundle}, then
\begin{equation}\label{vertical-lift-dmx}
(d^mx)^{\bf v} = (dx^1 \wedge \cdots \wedge dx^{m})^{\bf v} = \frac{\partial}{\partial p}.
\end{equation}

\section{Diagram}\label{diagram}

The diagram below illustrates some of the objects used in the paper and their relations. The arrows with label $(\;)^+$ (resp. $(\;) ^*$) indicate that we associate the vector bundle dual to a given affine (resp. vector) bundle.
The four dashed arrows pointing down indicate the associated model vector bundle to a given affine bundle.
The diagram illustrates the isomorphism $\sharp^{\rm aff}$ of affine bundles and three vector bundle isomorphisms naturally associated to it:
\begin{itemize}
\item[\rm (1)] the vector bundle isomorphism $\sharp=-(\sharp^{\rm aff})^+$ obtained by taking (minus) the affine dual to $\sharp^{\rm aff}$;\\
\item[\rm (2)] the vector bundle isomorphism $\sharp^{\rm lin}$ naturally induced by $\sharp^{\rm aff}$ on the associated model vector bundles;\\
\item[\rm (3)] the vector bundle isomorphism $-(\sharp^{\rm lin})^*$ obtained by taking (minus) the dual to $\sharp^{\rm lin}$.
\end{itemize}
Given a Hamiltonian section $h \in \Gamma ( \mu )$, an observable $ \alpha ^0 \in \mathcal{O} $ and a section $ \mathcal{F} _0 \in \Gamma((\pi \circ \nu^0)^*(\Lambda^mT^*M))$ (the space of sections of the vector bundle associated to $ \Gamma ( \mu )$), the diagram also illustrates the sections $ dh$, $ d \alpha ^0$, $d^l \mathcal{F} _0$, $ \mu _0(d \alpha ^0)$ and the corresponding objects obtained by applying sharp operators, namely $\sharp^{\rm aff} dh$, $ \sharp d \alpha ^0$, $\sharp^{\rm lin} d^l \mathcal{F} _0$, $ -( \sharp^{\rm lin}) ^* \mu _0(d \alpha ^0)$. These objects enter into the definition of the three structures $\{\cdot,\cdot\}$, $\{\cdot,\cdot\}_l$, $\{\cdot,\cdot\}_\mathcal{O}$ as 
\begin{align*} 
\{\alpha_0, h\}&= \langle d\alpha^0, \sharp^{\rm aff} dh  \rangle = - \langle\sharp d\alpha^0, dh \rangle \\
\{\alpha_0, \mathcal{F}_0\}_l&= \langle \mu_0d\alpha^0, \sharp^{\rm lin} d^l\mathcal{F}_0  \rangle = - \langle\mathcal{H}_{\alpha^0}, d^l\mathcal{F}_0  \rangle \\
\{\alpha_0, \beta_0\}_ \mathcal{O} &= i_{\mathcal{H}_{\beta^0}} d\alpha^0 =  - i_{\mathcal{H}_{\alpha^0}} d\beta^0 .
\end{align*}
\[
\hspace{-3cm}\begin{xy}
\xymatrix{
& V(\pi\circ \nu)\;\stackrel{{\tilde{\mathcal I}}}{\simeq}\;\ar[dd]^{\mathcal{J}} \widetilde{ \mathbb{P}(\pi)}^+
\ar@/^2.5pc/[rrrdd]^{\bar{\flat}}& & & \mathcal{M}(\pi\circ\nu^0)=J^1(\pi\circ\nu^0)^+&\\
& &  \widetilde{ \mathbb{P}(\pi)}\ar@{~>}[lu]^{(\;)^+}\ar[dd]_<<<<<<<{\tilde\mu}& &  &J^1(\pi\circ\nu^0)
\ar[dd]^p\ar@{~>}[lu]^{(\;)^+}\ar@/_2.5pc/[llldd]_{ \quad \;A}\\
& \frac{V(\pi\circ \nu)}{(\pi \circ \nu)^*(\Lambda^mT^*M)} \;\stackrel{{\mathcal I}}{\simeq}\;\mathbb{P}(\pi)^+\ar@{.>}[ddddd]&  &  &L=\left( \frac{J^1(\pi\circ\nu^0)}{ \operatorname{ker} A}\right) ^+\ar_{\hspace{1cm}\sharp=\flat^{-1}= - (\sharp^{\rm aff})^+}[lll] \ar@{^{(}->}[uu]\ar@{.>}[ddddd]& \\
& &  \mathbb{P}(\pi)\ar@{~>}[lu]_{(\;)^+}\ar_{\hspace{1cm}\sharp^{\rm aff}= (\hat A)^{-1}}[rrr]\ar@{.>}[ddddd]&  & & \frac{J^1(\pi\circ\nu^0)}{ \operatorname{ker} A} \ar@{~>}[lu]_{(\;)^+}\ar@{.>}[ddddd] \\
& & & & &\\
& & & \mathcal{M}^0\pi\ar@{-->}@/^0.5pc/[uul]_{dh}\ar@{-->}@/_0.5pc/[uurr]_{\Gamma_h=\sharp^{\rm aff}dh}\ar@{-->}@/^0.5pc/[lluuu]^{\sharp d\alpha^0}\ar@{-->}@/_0.5pc/[uuur]^<<<<<<<<<<<{ d\alpha^0} \ar@{-->}@/_0.5pc/[ddll]_{\mathcal{H}_{\alpha^0}= - (\sharp^{\rm lin })^*\mu_0 (d\alpha^0)}\ar@{-->}@/_1pc/[dddl]^<<<<<<<<{d^l\mathcal{F}_0}\ar@{-->}@/^0.5pc/[ddr]_{\mu_0(d\alpha^0)}\ar@{-->}@/^1.5pc/[dddrr]^{\sharp^{\rm lin}(d^l\mathcal{F}_0)}& &\\
& & & & &\\
& V(\pi \circ\nu^0)=V(\mathbb{P}(\pi))^*& & & V\left( \frac{J^1(\pi\circ\nu^0)}{ \operatorname{ker}A}\right)^*\ar_{\hspace{1.5cm}- (\sharp^{\rm lin})^*}[lll]&\\
& &V(\mathbb{P}(\pi)) \ar@{~>}[lu]^{(\;)^*}\ar_{\sharp^{\rm lin}= (\flat^{\rm lin})^{-1}}[rrr]& & &V\left( \frac{J^1(\pi\circ\nu^0)}{ \operatorname{ker} A}\right) \ar@{~>}[lu]^{(\;)^*}
}
\end{xy}
\]

\end{document}